\documentclass[12pt]{article}

\usepackage[multiple]{footmisc}
\usepackage{mathrsfs}
\usepackage{pdfsync}

\usepackage{fullpage}
\usepackage{amssymb}
\usepackage[round]{natbib}
\usepackage{amsmath}
\usepackage{amsthm}
\usepackage{amsfonts}
\usepackage{mathtools}
\usepackage{graphicx}
\usepackage{bm}
\usepackage{array}
\usepackage{comment}
\usepackage{xcolor}
\usepackage[ruled]{algorithm2e}
\usepackage{bezier, latexsym}
\usepackage[unicode=true,pdfusetitle,
 bookmarks=true,bookmarksnumbered=false,bookmarksopen=false,
 breaklinks=true,pdfborder={0 0 0},pdfborderstyle={},backref=false,colorlinks=true]
 {hyperref}
\hypersetup{linkcolor=blue, citecolor=blue, urlcolor=blue}

\usepackage{tikz}
\usepackage[hypcap=false]{caption}
\usetikzlibrary{positioning}
\usepackage[para, flushleft]{threeparttable}
\usepackage{pifont}

\usepackage{tabularx, booktabs} 

\newenvironment{subproof}[1][\proofname]{%
  \begin{proof}[#1]%
}{%
  \end{proof}%
}

\newtheorem{theorem}{Theorem}
\newtheorem{proposition}{Proposition}
\newtheorem{lemma}{Lemma}
\newtheorem{corollary}{Corollary}

\theoremstyle{definition}

\newtheorem{claim}{Claim}
\newtheorem{example}{Example}

\newtheorem{fact}{Fact}

\newtheoremstyle{propertystyle}
  {\topsep}       
  {\topsep}       
  {\normalfont}   
  {}              
  {\bfseries}     
  {.}             
  {.5em}          
  {\thmnote{#3}}  

\theoremstyle{propertystyle}
\newtheorem*{property}{}

\usepackage{centernot}

\usepackage[normalem]{ulem} 
\usepackage{xcolor} 

\newcommand{\prof}[1]{\mathsf{#1}}
\newcommand{\defterm}[1]{\emph{#1}}
\newcommand{\defprop}[1]{\emph{#1}}

\newcommand{\prop}[1]{\textsl{#1}}

\usepackage{xspace}
\newcommand{\abbrev}[1]{\textsc{#1}}

\DeclareRobustCommand{\PTB}{\abbrev{PTB}\xspace}
\DeclareRobustCommand{\FTE}{\abbrev{FTE}\xspace}
\DeclareRobustCommand{\SP}{\abbrev{SP}\xspace}
\DeclareRobustCommand{\NB}{\abbrev{NB}\xspace}
\DeclareRobustCommand{\NW}{\abbrev{NW}\xspace}
\DeclareRobustCommand{\IPLB}{\abbrev{IPLB}\xspace}
\DeclareRobustCommand{\GSP}{\abbrev{GSP}\xspace}
\DeclareRobustCommand{\PE}
{\abbrev{PE}\xspace}

\DeclareRobustCommand{\IR}{\abbrev{IR}\xspace}

\begin{document}

\title{Breaking Ties among Equals: Serial and Sequential Dictatorships in House Allocation\thanks{
We are grateful to William Thomson and Bettina Klaus for detailed comments, and to P\'eter Bir\'o, Yang Chen, Ravi Jagadeesan, and Subhajit Pramanik for helpful discussions. We gratefully acknowledge financial support from the Swiss National Science Foundation (Project 100018-231777); Feng also acknowledges support from
the National Natural Science Foundation of China (No.~72503028), and the Liaoning Provincial Natural Science Foundation of China (No.~2026-BS-0755).
This material is based upon work supported by the National Science Foundation under Grant No.~DMS-1928930 and by the Alfred P.~Sloan Foundation under grant G-2021-16778, while the authors were in residence at the Simons Laufer Mathematical Sciences Institute in Berkeley, California, during the Fall 2023 semester. Part of this research was conducted while Feng was visiting the 
Department of Informatics at Kyushu University, and while Coreno was in residence at the Hausdorff Research Institute for Mathematics in Bonn, funded by the Deutsche Forschungsgemeinschaft (DFG, German Research Foundation) under Germany's Excellence Strategy -- EXC-2047/2 -- 390685813.} 
}

\author{Di Feng\textsuperscript{a,$\dagger$} \and Jacob Coreno\textsuperscript{b,}\thanks{Corresponding authors.}}

\renewcommand{\thefootnote}{\alph{footnote}}

\date{\today}

\maketitle

\footnotetext[1]{Dongbei University of Finance and Economics, China. Contact: \protect\href{mailto:dfeng@dufe.edu.cn}{dfeng@dufe.edu.cn}.}

\footnotetext[2]{University of Lausanne, Switzerland. Contact: \protect\href{mailto:jacob.coreno@unil.ch}{jacob.coreno@unil.ch}.}

\renewcommand{\thefootnote}{\arabic{footnote}}
\setcounter{footnote}{0}

\vspace{-2.4em}
\begin{abstract}
\linespread{1}\selectfont

We study the deterministic allocation of indivisible objects to agents under unit demand and strict preferences, without monetary transfers. When every agent must receive an object, two agents who report identical preferences cannot be treated equally, so every mechanism embodies a rule for breaking such ties. We introduce two invariance requirements on this rule. \defprop{Fair treatment of equals} requires the favored member of every preference-identical pair to be the same at all profiles; the weaker \defprop{pairwise tie-breaking} allows the favored member to depend on the reports of the other agents but not on the pair's common report. Combined with non-wastefulness and strategy-proofness, (i) fair treatment of equals characterizes the serial dictatorships, (ii) pairwise tie-breaking and non-bossiness characterize the sequential dictatorships, and (iii) pairwise tie-breaking and the identical-preference lower bound---which compares each agent's assignment with her outcome when everyone shares her ranking---characterize the serial dictatorships once again. When agents may remain unassigned, the first two characterizations extend upon adding individual rationality, but the identical-preference lower bound no longer forces a fixed order. Among sequential dictatorships, the identical-preference lower bound is equivalent to \defterm{boundedness}: there is a priority whose $k$th-ranked agent is called after at most $k-1$ real objects have been assigned.

\medskip

\medskip

\end{abstract}

\pagebreak

\section{Introduction}
A finite set of heterogeneous and indivisible objects---such as houses, offices, school seats, or tasks---is to be assigned to a group of agents, one object per agent, without monetary transfers. We study deterministic mechanisms for this problem when preferences are strict. In this environment, a basic tension arises before any question of design is posed: agents who report identical preferences cannot be treated equally. If two agents submit the same ranking, they must nevertheless receive different objects, so one necessarily prefers the other's assignment---an asymmetry that cannot be traced to anything the two agents reported differently. Every mechanism therefore embodies, implicitly or explicitly, a rule for resolving such ties. We ask what restrictions on this local tie-breaking behavior imply about the global structure of a mechanism once standard efficiency and incentive requirements are imposed.

We study two invariance requirements on tie-breaking. Both restrict the grounds on which one of two preference-identical agents may be favored. Differences in reported preferences may provide a report-based reason for favoring one agent over another; a common report provides none. Under \prop{fair treatment of equals}, the favored member of every preference-identical pair must be fixed across all profiles. The weaker \prop{pairwise tie-breaking} permits the tie-break to depend on the reports of the other agents, but not on the pair's common report. Both requirements are local: they impose no direct restriction when all agents report distinct preferences, which is the case at all but a vanishing share of profiles as the number of objects grows (see footnote~\ref{fn:vanishing}).

\prop{Fair treatment of equals} implies the existence of a single priority order governing the relative treatment of preference-identical agents. When imposed with respect to a priority that reflects some defensible notion of justice---seniority, waiting time, urgency, or a lottery drawn in advance---it also has a direct fairness interpretation: the higher-priority member of an identically reporting pair must receive the preferred object. Without such a justification, both properties are simply invariance requirements on tie-breaking. The restriction to tied reports leaves genuine room for other objectives. For example, \prop{fair treatment of equals} permits rank-efficient allocations that \prop{weak fairness} can rule out.\footnote{\prop{Weak fairness} requires higher-priority agents to never envy lower-priority agents, even when their reports differ \citep{svensson1994}. Section~\ref{sec:related} defines rank efficiency, explains its compatibility with \prop{fair treatment of equals}, and gives a three-agent example illustrating the contrast with \prop{weak fairness} (footnote~\ref{fn:rank-efficient-example}).}

These local requirements, combined with familiar incentive and efficiency conditions, pin down the global structure of the mechanism. \prop{Non-wastefulness}, \prop{strategy-proofness}, and \prop{fair treatment of equals} characterize the \defterm{serial dictatorships}: agents are called in one fixed order, and each in turn receives her most-preferred object among those still available (Theorem~\ref{thm0}). Replacing \prop{fair treatment of equals} with \prop{pairwise tie-breaking} and strengthening the incentive requirement to \prop{group strategy-proofness}---equivalently, \prop{strategy-proofness} and \prop{non-bossiness} in this environment---characterizes the larger class of  \defterm{sequential dictatorships}: the next agent called may depend on the assignments made earlier, but not on other features of the reported preference profile (Theorem~\ref{thm1}). Finally, retaining \prop{non-wastefulness}, \prop{strategy-proofness}, and \prop{pairwise tie-breaking}, but replacing \prop{non-bossiness} with the \prop{identical-preference lower bound}---a solidarity property guaranteeing every agent an assignment at least as good as the one she would receive under maximal competition, when every other agent shares her preferences---eliminates this dependence on earlier assignments. The resulting requirements characterize the serial dictatorships once again (Theorem~\ref{thm:iplb}).

Theorem~\ref{thm1} also locates the sequential dictatorships within the broader class of \prop{group strategy-proof} and \prop{Pareto efficient} mechanisms---the trading cycles mechanisms of \citet{pycia2017}, which generalize the hierarchical exchange mechanisms of \citet{papai2000}. The sequential dictatorships are exactly the trading cycles mechanisms in which, after every realized suballocation, a single agent owns all remaining objects \citep{pycia2017,pycia2023}. Theorem~\ref{thm1} shows that a single property---\prop{pairwise tie-breaking}---separates the sequential dictatorships from all other trading cycles mechanisms, without assuming \prop{Pareto efficiency} or prior membership in the class of trading cycles mechanisms.

The characterizations admit both constructive and impossibility interpretations. Constructively, they show that the local invariance requirements, combined with standard efficiency and incentive conditions, recover the familiar serial and sequential dictatorships. Read negatively, they identify the price of departing from those mechanisms. Every nonserial mechanism that is \prop{non-wasteful} and \prop{strategy-proof} must reverse the relative treatment of some pair across two profiles at which that pair reports identically. If the nonserial mechanism also satisfies the \prop{identical-preference lower bound}, then it exhibits such a reversal even with all other agents' reports held fixed. Likewise, every \prop{non-wasteful} and \prop{group strategy-proof} mechanism that is not a sequential dictatorship must reverse the relative treatment of some preference-identical pair solely on the basis of the preference they share.

The baseline model requires that every agent receive a real object. Section~\ref{sec:exten} allows agents to remain unassigned, either because objects are scarce or because they find no available object acceptable. In this environment, the characterizations of serial and sequential dictatorships extend after adding \prop{individual rationality} (Theorems~\ref{thm:serial-dictatorship-null} and~\ref{thm:sequential-dictatorship-null}). The \prop{identical-preference lower bound} behaves differently, however. Among the sequential dictatorships, it characterizes the \defterm{bounded} ones: there exists a priority whose $k$th-ranked agent is always called after at most $k-1$ real objects have been assigned (Proposition~\ref{prop:bounded-sequential} and Corollary~\ref{cor:bounded-characterization}). Since choosing nonassignment does not exhaust a real object, the bounded sequential dictatorships permit departures from a fixed calling order.

\subsection{Organization}
The paper proceeds as follows. Section~\ref{sec:model} introduces the model, Section~\ref{sec:property} defines the properties of allocations and mechanisms, and Section~\ref{dictatorships} defines serial dictatorships and sequential dictatorships. Section~\ref{sec:result} presents the characterizations and discusses the independence of their properties. Section~\ref{sec:exten} extends the analysis to a model with an outside option and characterizes the bounded sequential dictatorships through the \prop{identical-preference lower bound}. Section~\ref{sec:related} discusses the related literature. Appendix~\ref{appendix:independence} contains the formal independence examples, Appendix~\ref{appendix:proof} contains the omitted proofs for Sections~\ref{sec:property}--\ref{sec:result}, and Appendix~\ref{appendix:proofs-extension} contains those for Section~\ref{sec:exten}.

\section{The model}
\label{sec:model}

Let $N = \{1, \dots, n\}$ be a finite set of $n \geq 2$ \defterm{agents}, and let $O$ be a finite set of heterogeneous, indivisible \defterm{objects} with $|O| \geq n$.\footnote{In Section \ref{sec:exten}, we allow for fewer objects than agents and permit agents to remain unassigned.} An \defterm{allocation} is an injective function $\mu : N \to O$. We often represent an allocation as $\mu = (\mu_i)_{i \in N}$, where $\mu_i = \mu(i)$ denotes agent~$i$'s \defterm{assignment} at $\mu$. For each $S \subseteq N$, $\mu(S)=\{\mu_i \mid i\in S\}$ denotes the set of objects assigned to the agents in $S$. Let $\mathcal{M}$ denote the set of all allocations.

Each agent $i$ reports strict \defterm{preferences}, represented by a linear order $R_i$ on $O$.\footnote{Formally, $R_i$ is a complete, transitive, and antisymmetric binary relation on $O$.} Let $\mathcal{R}$ denote the set of all preference relations (i.e., linear orders) on $O$, and let $\mathcal{R}^N$ denote the set of \defterm{preference profiles} $\prof{R} = (R_i)_{i \in N}$. When no agent is specified, we write $\succeq, \succeq', \succeq^*, \dots$ for preference relations and $\succ, \succ', \succ^*,\dots$ for their strict parts. For each $\succeq \in \mathcal{R}$, let $\prof{R}(\succeq) \coloneqq (\succeq)_{i \in N}$ denote the profile at which every agent reports $\succeq$; we call profiles of this form \defterm{identical-preference profiles}.

For an agent-indexed relation $R_i \in \mathcal{R}$, its strict part is denoted by $P_i$: for all $x, y \in O$, $x \mathrel{P_i} y$ if and only if $x \mathrel{R_i} y$ and $x \neq y$. Subscripts, superscripts, and accents are carried over to the corresponding strict relations. We often represent preferences as ordered lists: $R_i : o_1, o_2, \dots, o_{|O|}$ means that $o_1 \mathrel{P_i} o_2 \mathrel{P_i} \cdots \mathrel{P_i} o_{|O|}$, while a truncated list $R_i : o_1, o_2, \dots, o_k,\dots$ means that $o_1 \mathrel{P_i} o_2 \mathrel{P_i} \cdots \mathrel{P_i} o_k$ and that every omitted object is ranked below $o_k$, in an order left unspecified. For every nonempty $A \subseteq O$, $\max_{R_i}(A)$ denotes the most-preferred object in $A$ according to $R_i$; that is, $\max_{R_i}(A) \in A$ and $\max_{R_i}(A) \mathrel{R_i} o$ for all $o \in A$. The same conventions also apply to relations denoted by $\succeq$ and their restrictions.

A \defterm{mechanism} is a function $f:\mathcal{R}^N\to \mathcal{M}$ that selects an allocation $f(\prof{R}) \in \mathcal{M}$ at each preference profile $\prof{R} \in \mathcal{R}^N$. For each $i\in N$, $f_i(\prof{R})$ denotes agent $i$'s \defterm{assignment} at $f(\prof{R})$.

\section{Properties of allocations and mechanisms}
\label{sec:property}

Throughout this section, let $f : \mathcal{R}^N \to \mathcal{M}$ be a mechanism.

\subsection{Efficiency}
An allocation $\mu$ is \defprop{non-wasteful} at a preference profile $\prof{R}$ if no agent prefers an unassigned object to her own assignment: for each $i \in N$ and each $o \in O \setminus \mu(N)$, $\mu_i \mathrel{R_i} o$. Otherwise, $\mu$ is \defprop{wasteful} at $\prof{R}$.

\begin{property}[Non-wastefulness (\NW)]
For every $\prof{R}\in \mathcal{R}^N$, the allocation $f(\prof{R})$ is \prop{non-wasteful} at $\prof{R}$.
\end{property}

If there are exactly as many objects as agents ($|O|=n$), then every allocation $\mu \in \mathcal{M}$ satisfies $\mu(N) = O$, so every mechanism is trivially \prop{non-wasteful}.

An allocation $\mu$ is \defprop{Pareto efficient} at $\prof{R}$ if there is no allocation $\mu'$ such that $\mu'_i \mathrel{R_i} \mu_i$ for every $i \in N$ and $\mu'_j \mathrel{P_j} \mu_j$ for some $j \in N$.

\begin{property}[Pareto efficiency (\PE)]
For every $\prof{R} \in \mathcal{R}^N$, the allocation $f(\prof{R})$ is \prop{Pareto efficient} at $\prof{R}$.
\end{property}

An allocation $\mu$ is \defprop{pairwise efficient} at $\prof{R}$ if no two agents can both benefit by exchanging their assignments: there do not exist $i,j\in N$ such that $\mu_i \mathrel{P_j} \mu_j$ and $\mu_j \mathrel{P_i} \mu_i$.\footnote{In the object reallocation model, \citet{ekici2022} calls this property \prop{pair efficiency} and characterizes the Top Trading Cycles mechanism by \prop{individual rationality}, \prop{strategy-proofness}, and \prop{pair efficiency}.}

\begin{property}[Pairwise efficiency]
For every $\prof{R}\in \mathcal{R}^N$, the allocation $f(\prof{R})$ is \prop{pairwise efficient} at $\prof{R}$. 
\end{property}

\prop{Pareto efficiency} implies both \prop{pairwise efficiency} and \prop{non-wastefulness}. Conversely, a \prop{pairwise efficient} and \prop{non-wasteful} allocation need not be \prop{Pareto efficient}, because a Pareto improvement may require a cycle of exchanges among three or more agents.

\subsection{Incentives}

Let $\prof{R} \in \mathcal{R}^N$ and $S \subseteq N$. We write $\prof{R}_S = (R_i)_{i \in S} \in \mathcal{R}^S$ for the restriction of $\prof{R}$ to $S$ and, given $\prof{R}'_S \in \mathcal{R}^S$, we write $(\prof{R}'_S, \prof{R}_{N \setminus S})$ for the preference profile in which every agent~$i \in S$ reports $R'_i$ and every agent~$j \in N\setminus S$ reports $R_j$. If $S = \{i \}$, we denote  $(\prof{R}'_{S}, \prof{R}_{N \setminus S})$ by $(R'_i, \prof{R}_{-i})$.

An agent~$i$ can \defterm{manipulate} mechanism $f$ at preference profile $\prof{R} \in \mathcal{R}^N$ if there exists $R'_i \in \mathcal{R}$ such that $f_i(R'_i, \prof{R}_{-i}) \mathrel{P_i} f_i(\prof{R})$. A mechanism is \defprop{strategy-proof} if no agent can manipulate it at any preference profile.

\begin{property}[Strategy-proofness (\SP)]
For every $\prof{R} \in \mathcal{R}^N$ and every $i\in N$, there is no $R'_i\in \mathcal{R}$ such that $f_i(R'_i, \prof{R}_{-i}) \mathrel{P_i} f_i(\prof{R})$.
\end{property}

More generally, a coalition $S \subseteq N$ can \defterm{manipulate} mechanism $f$ at preference profile $\prof{R} \in \mathcal{R}^N$ if there exists $\prof{R}'_S \in \mathcal{R}^S$ such that every member of $S$ weakly prefers her assignment at $f(\prof{R}'_S, \prof{R}_{N \setminus S})$ to her assignment at $f(\prof{R})$, and at least one member strictly prefers it. A mechanism is \defprop{group strategy-proof} if no coalition can manipulate it at any preference profile.

\begin{property}[Group strategy-proofness (\GSP)]
For every $\prof{R} \in \mathcal{R}^N$ and every $S\subseteq N$, there is no $\prof{R}'_S \in \mathcal{R}^S$ such that (i) for each $i \in S$, $f_i(\prof{R}'_S,\prof{R}_{N \setminus S}) \mathrel{R_i} f_i(\prof{R})$, and (ii) for some $j \in S$, $f_j(\prof{R}'_S,\prof{R}_{N \setminus S}) \mathrel{P_j} f_j(\prof{R})$.
\end{property}

It is immediate from the definitions that \prop{group strategy-proofness} implies \prop{strategy-proofness}. It also implies \defprop{non-bossiness}, an invariance property that ensures no individual agent can influence the allocation without changing her own assignment.

\begin{property}[Non-bossiness (\NB)]
For every $\prof{R} \in \mathcal{R}^N$, every $i \in N$, and every $R'_i\in \mathcal{R}$, $f_i(R'_i,\prof{R}_{-i}) = f_i(\prof{R})$ implies $f(R'_i,\prof{R}_{-i}) = f(\prof{R})$.
\end{property}

In this environment, \prop{group strategy-proofness} is in fact equivalent to the conjunction of \prop{strategy-proofness} and \prop{non-bossiness} \citep{papai2000}. Fact~\ref{fact} in Appendix~\ref{appendix:proof} records this equivalence together with a further equivalent formulation, \prop{(Maskin) monotonicity}, which we use in the proofs.

\subsection{Fairness and tie-breaking}
In this deterministic environment, competition for scarce objects makes standard fairness guarantees such as ``envy-freeness'' and ``equal treatment of equals'' unattainable.\footnote{By contrast, randomization allows agents who report identical preferences to receive the same lottery over objects, even though their realized assignments differ.}  When two agents rank the same object first, assigning it to one of them necessarily leaves the other envious. The conflict is starkest when two agents report identical preferences: since their assignments are distinct, both agree on which agent receives the better object, and nothing in their common report could justify this unequal treatment. Our next two properties discipline how this unavoidable asymmetry is resolved: they restrict how the favored member of a preference-identical pair may vary across profiles.

The first property requires the resolution of such ties to be invariant across profiles: whenever two agents report identical preferences, the mechanism favors the same member of the pair, independently of their common preference and of the reports of all other agents.

\begin{property}[Fair treatment of equals (\FTE)]
    For all distinct $i,j \in N$ and all $\prof{R}, \prof{R}' \in \mathcal{R}^N$ such that $R_i = R_j$ and $R'_i = R'_j$,
\begin{equation}\label{FTE tie-break formulation}
f_i(\prof{R}) \mathrel{P_i} f_j(\prof{R}) \implies f_i(\prof{R}') \mathrel{R'_i} f_j(\prof{R}').
\end{equation}
\end{property}

Because distinct agents must receive distinct objects, one member of every preference-identical pair is strictly favored, and (\ref{FTE tie-break formulation}) requires that the identity of this member never changes. In the present model, this is equivalent to requiring a single priority to govern all such comparisons.

A \defterm{priority} is a strict linear order $\pi$ on $N$. Agent~$i$ has higher priority than agent~$j$ if $i \mathrel{\pi} j$. For each $k \in \{1,\dots, n\}$, let $\pi_k$ denote the agent with the $k$th-highest priority under $\pi$.\footnote{Formally, $\pi_k \in N$ satisfies $|\{j \in N \mid j \mathrel{\pi} \pi_k\}| = k -1$.} The \defterm{identity priority} is the priority $\pi$ such that $i = \pi_i$ for every $i \in N$. Let $\Pi$ denote the set of all priorities.

Mechanism $f$ satisfies \prop{fair treatment of equals} if and only if there exists $\pi \in \Pi$ such that, for every $\prof{R} \in \mathcal{R}^N$ and every $i,j\in N$, 
\begin{equation}\label{priority exists}
    (i \mathrel{\pi} j \text{ and } R_i=R_j) \implies f_i(\prof{R}) \mathrel{R_i} f_j(\prof{R}).
\end{equation}
When $\pi \in \Pi$ satisfies (\ref{priority exists}), we say that $f$ satisfies \prop{fair treatment of equals} with respect to $\pi$.\footnote{To obtain such a priority from (\ref{FTE tie-break formulation}), fix an identical-preference profile and let $\pi$ order the agents according to the common preference's ranking over their assignments. Their assignments are distinct, and (\ref{FTE tie-break formulation}) propagates each strict comparison at this profile to every other profile where a pair reports identically. Conversely, a priority satisfying (\ref{priority exists}) precludes any reversal in the relative treatment of a preference-identical pair.} The priority formulation makes the connection with serial dictatorships explicit, and it gives \prop{fair treatment of equals} a direct fairness interpretation when $\pi$ reflects a defensible notion of justice, as discussed in the introduction. \prop{Fair treatment of equals} is thus a deterministic relaxation of the unattainable equal treatment of equals. It also weakens the \defprop{weak fairness} condition of \citet{svensson1994}, which requires a priority to be respected at \emph{every} profile (see Section~\ref{sec:related}). 

\defprop{Pairwise tie-breaking} relaxes \prop{fair treatment of equals} along one dimension. It allows the favored member of a preference-identical pair to depend on the reports of agents outside the pair, while requiring it to remain independent of the pair's common report.

\begin{property}[Pairwise tie-breaking (\PTB)]
For all distinct $i,j \in N$ and all $\prof{R}, \prof{R}' \in \mathcal{R}^N$ such that $R_i = R_j$, $R'_i = R'_j$, and $\prof{R}_{N \setminus\{i,j\}} =  \prof{R}'_{N \setminus \{i,j\}}$, 
    $$
    f_i(\prof{R}) \mathrel{P_i} f_j(\prof{R}) \implies f_i(\prof{R}') \mathrel{R'_i} f_j(\prof{R}').
    $$
\end{property}
Thus, holding all other reports fixed, a change in an identical-preference pair's common report cannot by itself induce a reversal in their relative treatment. 

\prop{Fair treatment of equals}  implies \prop{pairwise tie-breaking} because the former implies that the favored member of any preference-identical pair is independent of all reports. When $n = 2$, the converse also holds, since there are no reports outside the pair on which the tie-break could depend; when $n \geq 3$, the converse fails (see Example~\ref{independence:IPLB}). Neither property restricts the mechanism at profiles where no two agents report the same preferences.\footnote{\label{fn:vanishing}Let $M \coloneqq |O|!$ be the number of strict preference relations over $O$. Among all $M^n$ profiles, the fraction at which all $n$ reports are distinct is
$$\frac{M (M - 1) \cdots (M - n + 1)}{M^n} = \prod_{k = 0}^{n-1} \left(1 - \frac{k}{M}\right),$$ which converges to one as $|O|$ grows (even if $n$ also grows, provided $n \leq |O|$). Thus the profiles at which the two tie-breaking requirements bind form a vanishing share of the domain. Nevertheless, together with the other properties, these requirements determine the mechanism's structure throughout the domain.}

The final property is of a different nature. Rather than restricting how pairwise ties are resolved, it compares each agent's assignment with the  assignment she would receive at the identical-preference profile in which every other agent shares her preferences.

\begin{property}[Identical-preference lower bound (\IPLB)]
For every $\prof{R} \in \mathcal{R}^N$ and every $i\in N$, $f_i(\prof{R}) \mathrel{R_i} f_i(\prof{R}(R_i))$.
\end{property}

The \prop{identical-preference lower bound} expresses a form of solidarity. For each agent~$i$ with preference $R_i$, the benchmark profile $\prof{R}(R_i)$ subjects her to maximal competition, and the bound requires that she is not harmed when the others' preferences depart from hers. Preference diversity therefore affects all agents in the same weak direction: it may benefit some more than others, but it creates no losers. \citet{moulin1990} calls these benefits \emph{positive preference externalities}.\footnote{Lower-bound properties of this kind descend from the fair-share guarantee of \citet{Steinhaus1948}. For the \prop{identical-preference lower bound} in models with indivisible objects, see \citet{bevia1996}, \citet{fujinaka2007}, and \citet{hashimoto2018}; \citet{thomson2023} surveys related lower-bound properties.}\textsuperscript{,}\footnote{The \prop{identical-preference lower bound} is reminiscent of the maximin-share guarantee of \citet{budish2011combinatorial}, but the two requirements are distinct. In the present single-unit environment, an agent's maximin share is her $n$th-ranked object, and \prop{non-wastefulness} alone guarantees an object among her $n$ most-preferred. Combining the \prop{identical-preference lower bound} with \prop{non-wastefulness} yields a stronger guarantee: the objects assigned at $\prof{R}(R_i)$ are exactly the $n$ objects that $R_i$ ranks highest, so agent~$i$'s assignment at $\prof{R}(R_i)$ is weakly better than her $n$th-ranked object---and strictly better unless $i$ is the agent who receives this object.}

\section{Serial and sequential dictatorships}\label{dictatorships}

We now introduce the two classes of mechanisms characterized in this paper: \defterm{serial dictatorships} and \defterm{sequential dictatorships}. 

Let $\pi$ be a priority. The \textit{serial dictatorship associated with $\pi$}, denoted  $f^\pi$, calls the agents in the order $\pi_1, \dots, \pi_n$; when called, each agent receives her most-preferred object among those not already assigned. Formally, for every $k \in \{1, \dots, n\}$ and every $\prof{R} \in \mathcal{R}^N$,
\begin{equation}
f_{\pi_k}^{\pi}\left(\prof{R}\right)	=\max_{R_{\pi_k}}\left(O \setminus \bigcup_{t=1}^{k-1}\left\{ f_{\pi_t}^{\pi}\left(\prof{R}\right) \right\} \right),
\end{equation}
where the union is empty if $k = 1$. A mechanism is a \textit{serial dictatorship} if it equals $f^\pi$ for some priority $\pi$.

A sequential dictatorship also calls the agents one at a time and assigns to each her most-preferred remaining object, but the identity of the next agent to be called may depend on the history of assignments realized earlier. We describe such histories using suballocations. A \defterm{suballocation} is an injective function $\varsigma : S \to O$ for some subset $S \subseteq N$. Let $\mathcal{S}$ be the set of all suballocations, and let $\varnothing$ denote the empty suballocation. For  each suballocation $\varsigma : S \to O$, define 
$$
N(\varsigma) \coloneqq S, \qquad O(\varsigma) \coloneqq \varsigma(S), \qquad |\varsigma| \coloneqq |S|.
$$
Thus $N(\varsigma)$ and $O(\varsigma)$ are the sets of assigned agents and assigned objects under $\varsigma$. Write $A(\varsigma) \coloneqq O \setminus O(\varsigma)$ for the set of objects remaining after $\varsigma$. We identify $\varsigma$ with its graph $\{(i, \varsigma(i)) \mid i  \in S\}$. Accordingly, we write $\varsigma ' \subseteq \varsigma$ if $\varsigma$ extends the suballocation $\varsigma'$. Note that $\mathcal{M} = \{ \varsigma \in \mathcal{S} \mid N(\varsigma) = N\}$.

A \defterm{dictator mapping} is a function $\kappa : \mathcal{S} \setminus \mathcal{M} \to N$ such that 
$$
\kappa(\varsigma) \in N \setminus N(\varsigma) \quad \text{for every } \varsigma \in \mathcal{S} \setminus \mathcal{M}.
$$
Given a dictator mapping $\kappa$ and a preference profile $\prof{R} \in \mathcal{R}^N$, recursively define $\varsigma^\kappa_0(\prof{R}) \coloneqq \varnothing$ and, for each $t \in \{1,\dots, n\}$,
$$
i^\kappa_t(\prof{R}) \coloneqq \kappa\big(\varsigma^\kappa_{t-1}(\prof{R})\big) \quad \text{and} \quad \varsigma^{\kappa}_t(\prof{R}) \coloneqq \varsigma^\kappa_{t-1}(\prof{R}) \cup \left\{ \left(i^\kappa_t(\prof{R}), \max_{R_{i^\kappa_t(\prof{R})}} A \big(\varsigma^\kappa_{t-1}(\prof{R}) \big) \right) \right\}.
$$
Thus, at step~$t$, agent~$i^\kappa_t(\prof{R})$ is called and assigned her most-preferred remaining object. A suballocation $\varsigma$ is \defterm{$\kappa$-reachable}, or \defterm{reachable under $\kappa$}, if $\varsigma^\kappa_{|\varsigma|}(\prof{R}) = \varsigma$ for some $\prof{R} \in \mathcal{R}^N$; we say that such a profile \defterm{reaches} $\varsigma$ under $\kappa$. By construction, at any two profiles reaching the same proper suballocation $\varsigma$ after $t$ steps, the agent called at step~$t+1$ is $\kappa(\varsigma)$.

The \defterm{sequential dictatorship associated with $\kappa$} is the mechanism $f^\kappa$ defined by $f^\kappa(\prof{R}) = \varsigma^\kappa_n(\prof{R})$ for every $\prof{R} \in \mathcal{R}^N$. A mechanism is a \defterm{sequential dictatorship} if it equals $f^\kappa$ for some dictator mapping $\kappa$.\footnote{A dictator mapping $\kappa$ assigns a next dictator to every proper suballocation, including unreachable ones. For example, if $\kappa(\varnothing) = i$, then $N(\varsigma^\kappa_1(\prof{R})) = \{i\}$ for every $\prof{R} \in \mathcal{R}^N$, so no nonempty suballocation $\varsigma$ with $i \notin N(\varsigma)$ is $\kappa$-reachable. The values of $\kappa$ at unreachable suballocations do not affect $f^\kappa$. Appendix~\ref{sec: Partial dictator mappings} extends reachability to partial dictator mappings.}

Every serial dictatorship is a sequential dictatorship. Indeed, for any priority $\pi$, the serial dictatorship $f^\pi$ equals a sequential dictatorship $f^\kappa$ whose associated dictator mapping $\kappa$ satisfies
\begin{equation}\label{eq:serial-as-sequential}
i^\kappa_t(\prof{R}) = \pi_t \quad \text{for every } t \in \{1, \dots, n\} \text{ and every } \prof{R} \in \mathcal{R}^N.
\end{equation}
The inclusion is strict when $n \geq 3$ (see Example~\ref{independence:IPLB}). When $n = 2$, however, every sequential dictatorship is a serial dictatorship, since the first dictator is fixed and the identity of the second is then determined.

\section{Main results}
\label{sec:result}

We present three related characterizations. All three share \prop{non-wastefulness} and \prop{strategy-proofness} and differ in the properties added to this common core: \prop{fair treatment of equals} yields the serial dictatorships (Theorem~\ref{thm0}); \prop{non-bossiness} and \prop{pairwise tie-breaking} yield the sequential dictatorships (Theorem~\ref{thm1}); and the \prop{identical-preference lower bound} and \prop{pairwise tie-breaking} yield the serial dictatorships once again (Theorem~\ref{thm:iplb}). \prop{Non-wastefulness} matters only when $|O| > n$; when $|O| = n$, it holds automatically, and the remaining properties alone characterize the respective class.

\subsection{Fair treatment of equals}

\begin{theorem}\label{thm0}
	A mechanism satisfies \prop{non-wastefulness}, \prop{strategy-proofness}, and \prop{fair treatment of equals}
	if and only if it is a serial dictatorship.
\end{theorem}

The proof shows how the three properties interact. \prop{Non-wastefulness} ensures that the highest-priority agent's most-preferred object is assigned to someone. If another agent received it, \prop{strategy-proofness} would allow that agent to retain it after adopting the highest-priority agent's report, contradicting \prop{fair treatment of equals}. Hence the highest-priority agent receives her most-preferred object, and the argument proceeds recursively down the priority.

Since serial dictatorships are \prop{Pareto efficient} and \prop{group strategy-proof}, Theorem~\ref{thm0} derives both properties rather than assuming them. In particular, it dispenses with \prop{non-bossiness}, which frequently appears---either directly or through a stronger requirement such as \prop{consistency}---in characterizations of \prop{strategy-proof} object-allocation mechanisms (see, e.g., \citealp{papai2000,papai2001,klaus2002,pycia2017}). \prop{Non-bossiness} is imposed directly in \citeauthor{svensson1999}'s (\citeyear{svensson1999}) characterization of serial dictatorships by \prop{strategy-proofness}, \prop{non-bossiness}, and \prop{neutrality}, and in the characterization of \citet{sen2026priority} that imposes \prop{neutrality} only at identical-preference profiles (see Proposition~\ref{prop:another_characterization}); it appears implicitly in the characterization by \prop{weak Pareto efficiency} and pairwise relaxations of \prop{consistency} and \prop{neutrality} in the variable-population framework of \citet{ergin2000}.\footnote{\label{fn:consistency}In a variable-population environment, a mechanism is \prop{consistent} if, whenever a set of agents is removed from a problem together with their assignments, the remaining agents receive the same assignments in the reduced problem as in the original one. If agent~$i$'s report could change the allocation without changing her own assignment, then removing agent~$i$ together with that common assignment from the two problems would yield the same reduced problem, to which \prop{consistency} assigns two different allocations---a contradiction. Hence a \prop{bossy} mechanism cannot be \prop{consistent}.}

We close this subsection with a companion to Theorem~\ref{thm0} that trades the reach of the tie-breaking requirement against the strength of the incentive requirement. It isolates an intermediate condition implied by both \prop{fair treatment of equals} and \prop{neutrality} (see Section~\ref{sec:related}). A mechanism satisfies \defprop{identical-preference fair treatment of equals} with respect to a priority $\pi \in \Pi$ if condition (\ref{priority exists}) holds at every identical-preference profile: for every $\succeq \in \mathcal{R}$ and every $i, j \in N$, $i \mathrel{\pi} j$ implies $f_i(\prof{R}(\succeq)) \mathrel{\succeq} f_j(\prof{R}(\succeq))$. Both \prop{fair treatment of equals} and \prop{neutrality} imply this condition.\footnote{\label{fn:neutrality}To see that \prop{neutrality} implies this condition, fix an identical-preference profile and rank the agents according to the ranks of their assignments. Every other identical-preference profile is obtained by some relabeling of the objects, and \prop{neutrality} preserves the assigned ranks. The argument uses \prop{neutrality} only at identical-preference profiles.}  This restriction of \prop{fair treatment of equals} regains its full force once \prop{non-bossiness} is also imposed.

\begin{proposition}\label{prop:another_characterization}
Let $\pi \in \Pi$. A mechanism satisfies \prop{non-wastefulness}, \prop{strategy-proofness}, \prop{non-bossiness}, and \prop{identical-preference fair treatment of equals} with respect to $\pi$ if and only if it equals the serial dictatorship $f^\pi$.
\end{proposition}

Under \prop{non-wastefulness}, \prop{identical-preference fair treatment of equals} with respect to some priority is equivalent to the \defprop{unanimous profile neutrality} of \citet{sen2026priority}, which imposes \prop{neutrality} only at identical-preference profiles. \prop{Non-wastefulness} makes the assignments the $n$ highest-ranked objects at each identical-preference profile, so fixing the agents' relative order is equivalent to fixing each agent's assigned rank. Consequently, Proposition~\ref{prop:another_characterization} is the universal-domain specialization of \citet[Theorems~3.1 and~5.1]{sen2026priority}, who characterize the preference domains on which these properties single out the serial dictatorships. Appendix~\ref{app:another_characterization_proof} contains a short proof, which also covers the $|O| = n = 2$ case excluded in their analysis.

Theorem~\ref{thm0} and Proposition~\ref{prop:another_characterization} characterize the same class under two different incentive requirements: with \prop{strategy-proofness} alone, the priority must be respected whenever two agents report identically; with \prop{group strategy-proofness}, it suffices that it is respected when all agents report identically. Section~\ref{sec:related} further develops this comparison with characterizations via \prop{neutrality} \citep[e.g.,][]{svensson1999,sen2026priority}.

\subsection{Pairwise tie-breaking}
We next weaken \prop{fair treatment of equals} to \prop{pairwise tie-breaking} and strengthen \prop{strategy-proofness} to \prop{group strategy-proofness} (equivalently, \prop{strategy-proofness} and \prop{non-bossiness}).
The characterized class of mechanisms expands from the serial to the sequential dictatorships. The proof of the characterization uses the following efficiency implication.

\begin{lemma}\label{lemma:pe}
If a mechanism satisfies \prop{pairwise tie-breaking}, \prop{strategy-proofness}, and \prop{non-bossiness}, then it satisfies \prop{pairwise efficiency}.
\end{lemma}

The proof is short and illustrates how \prop{pairwise tie-breaking} acquires efficiency content. If two agents could both gain by exchanging their assignments $x$ and $y$, then \prop{strategy-proofness} and \prop{non-bossiness}---through their equivalence with \prop{monotonicity} (Fact~\ref{fact})---imply that the outcome is preserved when both agents report $\succeq : x,y,\dots$ and when both report $\succeq': y, x, \dots$. The first common report favors one agent and the second favors the other, violating  \prop{pairwise tie-breaking}.

In fact, every mechanism satisfying \prop{non-wastefulness}, \prop{strategy-proofness}, \prop{non-bossiness}, and \prop{pairwise tie-breaking} is a sequential dictatorship.

\begin{theorem}\label{thm1}
	A mechanism satisfies \prop{non-wastefulness}, \prop{strategy-proofness}, \prop{non-bossiness}, and \prop{pairwise tie-breaking} if and only if it is a sequential dictatorship. 
\end{theorem}

Since sequential dictatorships are \prop{Pareto efficient}, the properties in Theorem~\ref{thm1} imply \prop{Pareto efficiency}, although only \prop{non-wastefulness} is imposed directly. The same conclusion can be reached without the characterization. In independent work, \citet{demeulemeester2025pairwise} show that a \prop{monotonic mechanism}---equivalently, here, a \prop{group strategy-proof} one (Fact~\ref{fact})---is \prop{Pareto efficient} if and only if it is \prop{pairwise efficient} and \prop{non-wasteful}. Together with Lemma~\ref{lemma:pe}, this yields \prop{Pareto efficiency} directly from the properties of Theorem~\ref{thm1}.\footnote{When $|O| > n$, \prop{non-wastefulness} alone implies surjectivity, i.e., that every allocation is selected at some profile. Fix an allocation $\mu$ and  $o \in O \setminus \mu (N)$, and let $\prof{R}$ be a profile at which each agent~$i$ ranks $\mu_i$ first and $o$ second. \prop{Non-wastefulness} forces each agent~$i$ to receive $\mu_i$ at $\prof{R}$. A \prop{group strategy-proof} mechanism that selects every allocation at some profile is \prop{Pareto efficient}, since a Pareto improvement would permit a profitable deviation by the grand coalition. Thus Lemma~\ref{lemma:pe} is needed for this conclusion only when $|O| = n$.}\textsuperscript{,}\footnote{\label{fn:no PE}This suggests an alternative strategy for proving Theorem~\ref{thm1}: having obtained \prop{Pareto efficiency}, one could invoke the characterization of the \prop{group strategy-proof} and \prop{Pareto efficient} mechanisms as trading cycles mechanisms \citep{pycia2017} and then identify the trading cycles mechanisms that satisfy \prop{pairwise tie-breaking}. We do not follow this route. The argument we give is self-contained and invokes neither \prop{Pareto efficiency} nor membership in the class of trading cycles mechanisms. A similar argument is used to prove Theorem~\ref{thm:iplb}, where \prop{non-bossiness} is not assumed, so these intermediate implications and the resulting trading cycles representation cannot be invoked at that stage of the argument.}

The proof of the only-if direction constructs the dictator mapping recursively: given that the mechanism acts as a sequential dictatorship for the first $k$ assignments, we identify a unique next dictator after every reachable suballocation $\varsigma$ with $k$ assignments. Fix a profile $\prof{R}^\varsigma$ that reaches $\varsigma$, and let all agents unassigned at $\varsigma$ change their report to $\succeq \in \mathcal{R}$. \prop{Non-wastefulness} implies that these agents receive the $\succeq$-best remaining objects; label these agents as $j_1(\succeq), \dots, j_m(\succeq)$, where $j_\ell(\succeq)$ receives the $\ell$th-best remaining object according to $\succeq$. \prop{Strategy-proofness}, \prop{non-bossiness},  \prop{pairwise tie-breaking}, and \prop{pairwise efficiency} (Lemma~\ref{lemma:pe}) imply that agents $j_1(\succeq), \dots, j_{m-1}(\succeq)$ retain the same assignments after agent $j_m(\succeq)$ changes her preference from $\succeq$ to any preference in $\mathcal{R}$. Iterating this argument then frees agents $j_{m-1}(\succeq), \dots, j_2(\succeq)$ from the common report $\succeq$ one at a time, from the bottom of the list, and throughout this process agent~$j_1(\succeq)$ retains the same assignment. Hence agent~$j_1(\succeq)$ retains the $\succeq$-best remaining object whenever she reports $\succeq$ and the $\varsigma$-assigned agents report as they do in $\prof{R}^\varsigma$, regardless of the remaining agents' reports. \prop{Non-bossiness} then ensures that agent~$j_1(\succeq)$ retains this object whenever she reports $\succeq$
and the $\varsigma$-assigned agents deviate from $\prof{R}^\varsigma$ to \emph{any profile that reaches $\varsigma$}; hence $j_1(\succeq)$ is the \defterm{$\succeq$-dictator} after $\varsigma$. Finally, \prop{pairwise tie-breaking} implies that the $\succeq$-dictator after $\varsigma$ is the same agent for every $\succeq$; this agent is the next dictator after $\varsigma$.

When $n = 2$, the only-if part of Theorem~\ref{thm1} follows from Theorem~\ref{thm0}: \prop{pairwise tie-breaking} coincides with \prop{fair treatment of equals}, and every sequential dictatorship is a serial dictatorship; hence Theorem~\ref{thm0} implies that sequential dictatorships are characterized by \prop{non-wastefulness}, \prop{strategy-proofness}, and \prop{pairwise tie-breaking}. Neither equivalence used in this argument survives when $n \geq 3$.

\subsection{The identical-preference lower bound}
Our final characterization replaces \prop{non-bossiness} in Theorem~\ref{thm1} with the \prop{identical-preference lower bound}. We first establish that serial dictatorships satisfy this property.

\begin{proposition}
	Serial dictatorships satisfy the \prop{identical-preference lower bound}. \label{prop:IPLB}
\end{proposition}

In a model in which agents may remain unassigned, \citet{hashimoto2018} characterizes a broader class of mechanisms satisfying \prop{strategy-proofness}, \prop{Pareto efficiency}, and the \prop{identical-preference lower bound}; the class includes the \defterm{You Request My House--I Get Your Turn (YRMH--IGYT)} mechanisms \citep{abdulkadirouglu1999} and, in particular, the serial dictatorships. Because his model differs from ours, Appendix~\ref{appendix:proof} contains a direct proof of the proposition.

By contrast, sequential dictatorships that are not serial dictatorships violate the \prop{identical-preference lower bound}: if the identity of a later dictator depends on the assignment of an earlier one, then an agent who is served early at her identical-preference profile can be demoted after a change in an earlier dictator's assignment (Example~\ref{independence:IPLB}). The next proposition shows that this failure is general.

\begin{proposition}\label{prop:seq + IPLB}
A sequential dictatorship satisfying the \prop{identical-preference lower bound} is a serial dictatorship.
\end{proposition}

Proposition~\ref{prop:seq + IPLB} implies that adding the \prop{identical-preference lower bound} to the properties of Theorem~\ref{thm1} immediately yields a characterization of the serial dictatorships. Theorem~\ref{thm:iplb} strengthens this characterization by dispensing with \prop{non-bossiness}.

\begin{theorem}\label{thm:iplb}
	A mechanism satisfies \prop{non-wastefulness}, \prop{strategy-proofness}, \prop{identical-preference lower bound}, and \prop{pairwise tie-breaking} if and only if it is a serial dictatorship.
\end{theorem}

The proof of the only-if direction again identifies a unique next dictator after every reachable suballocation $\varsigma$, but unlike that of Theorem~\ref{thm1} it cannot appeal to \prop{non-bossiness}. Instead, for any linear order $L$ over the objects unassigned at $\varsigma$, we consider the identical-preference profile that ranks the objects assigned at $\varsigma$ first, in the order in which they were assigned, and the remaining objects according to $L$; this profile reaches $\varsigma$. \prop{Non-wastefulness} ensures that the $L$-best remaining object is assigned to an agent unassigned at $\varsigma$; call this agent the \defterm{$L$-dictator} after $\varsigma$. Using each such profile as a benchmark, the \prop{identical-preference lower bound} together with \prop{pairwise tie-breaking} implies that the $L$-dictator is the same agent for every $L$. \prop{Strategy-proofness} and the \prop{identical-preference lower bound} then imply that this same agent receives her most-preferred remaining object at every profile reaching $\varsigma$. The recursive argument therefore yields a sequential dictatorship.

Proposition~\ref{prop:seq + IPLB} supplies the second part of the argument. Suppose the mechanism acts as a serial dictatorship for the first $k$ assignments. We first show that the next dictator is the same after any two reachable length-$k$ suballocations that differ for exactly one agent; otherwise, the dictator after one suballocation would be demoted after the other, violating the \prop{identical-preference lower bound} at a suitably chosen preference profile. This  invariance extends to any two reachable length-$k$ suballocations---adjacent or not---because one can pass from one to the other through a sequence of one-agent changes. Hence, the next dictator is constant across all reachable length-$k$ suballocations, and the sequential dictatorship is a serial dictatorship.

When $n = 2$, \prop{strategy-proofness} is redundant in Theorem~\ref{thm:iplb}: \prop{non-wastefulness}, \prop{identical-preference lower bound}, and \prop{pairwise tie-breaking} already characterize serial dictatorships.\footnote{To see this, let $N = \{1, 2\}$ and let $f$ be a mechanism satisfying \prop{non-wastefulness}, the \prop{identical-preference lower bound}, and \prop{pairwise tie-breaking}. Because $|N| = 2$, $f$ satisfies \prop{fair treatment of equals}. Assume $f$ satisfies \prop{fair treatment of equals}  with respect to priority $\pi = (1,2)$. Fix $\prof{R} \in \mathcal{R}^N$. At the identical-preference profile $\prof{R}(R_1)$, \prop{non-wastefulness} and \prop{fair treatment of equals} imply that agent~$1$ receives $\max_{R_1}(O)$. By the \prop{identical-preference lower bound},
$$
f_1(\prof{R}) \mathrel{R_1} f_1(\prof{R}(R_1)) = \max_{R_1}(O),
$$
and consequently $f_1(\prof{R}) = \max_{R_1}(O)$. Once agent~$1$ receives $\max_{R_1}(O)$, \prop{non-wastefulness} forces agent~$2$ to receive his most-preferred object among $O \setminus \{f_1(\prof{R})\}$. Thus $f(\prof{R}) = f^\pi(\prof{R})$.}

\subsection{Independence of the properties}\label{sec:independent}

The properties in Theorems~\ref{thm0}--\ref{thm:iplb} are independent, subject to the population-size and supply qualifications discussed below. Appendix~\ref{appendix:independence} exhibits, for each property,  a mechanism that violates it while satisfying the remaining properties of the corresponding theorem; Table~\ref{tab:properties_L} summarizes the constructions, which we briefly describe here.

\prop{Non-wastefulness} is needed whenever $|O| > n$ (when $|O| = n$, it holds automatically): a serial dictatorship restricted to a fixed proper subset of the objects inherits every other property in Theorems~\ref{thm0}--\ref{thm:iplb} but never assigns an excluded object (Example~\ref{independence:NW}).

\prop{Pairwise tie-breaking} (and hence \prop{fair treatment of equals}) can fail for the YRMH--IGYT mechanisms, which are \prop{Pareto efficient}, \prop{group strategy-proof}, and satisfy the \prop{identical-preference lower bound}. In Example~\ref{ex: notPTB}, the favored member of a preference-identical pair of agents is determined by ownership (although our model has no exogenous ownership structure), so the tie-break depends on the pair's common preference report.

Example~\ref{ex:notSP} constructs a manipulable mechanism satisfying \prop{Pareto efficiency}, \prop{non-bossiness}, \prop{fair treatment of equals}, and the \prop{identical-preference lower bound}. At every profile, the outcome is determined by lexicographically maximizing a rank-count score subject to the \prop{identical-preference lower bound}; we refer to the appendix for details.\footnote{One can construct simpler examples to establish the independence of \prop{strategy-proofness} in Theorems~\ref{thm0} and~\ref{thm1}, for any $n \geq 2$. For Theorem~\ref{thm:iplb}, any such example requires $n \geq 3$. Example~\ref{ex:notSP} is chosen because it violates only \prop{strategy-proofness} among all the properties in Table~\ref{tab:properties_L}, and therefore establishes independence in all three theorems at once.}

A mechanism that selects between two serial dictatorships, both placing agent~$1$ first, according to agent~$1$'s second-ranked object is \prop{Pareto efficient}, \prop{strategy-proof}, and satisfies \prop{pairwise tie-breaking}, but it is \prop{bossy} (Example~\ref{exampleNB}).

Finally, any nonserial sequential dictatorship satisfies the properties of Theorem~\ref{thm1} but violates the \prop{identical-preference lower bound} (Example~\ref{independence:IPLB}; see also Proposition~\ref{prop:seq + IPLB}).

Examples~\ref{ex:notSP},~\ref{exampleNB}, and~\ref{independence:IPLB} use $n \geq 3$ agents. When $n = 2$, \prop{non-bossiness} and \prop{identical-preference lower bound} are redundant in Theorems~\ref{thm1} and~\ref{thm:iplb} respectively, since \prop{non-wastefulness}, \prop{strategy-proofness}, and \prop{pairwise tie-breaking} already characterize the serial dictatorships in that case. Likewise, when $n = 2$, \prop{strategy-proofness} is redundant in Theorem~\ref{thm:iplb}.

\begin{table}[ht]
\caption{Properties of the mechanisms in the independence examples}
\label{tab:properties_L}
\centering
\renewcommand{\arraystretch}{1.15}
\begin{tabular}{l*{7}{c}}
\hline
Mechanism & \NW & \SP & \NB & \FTE & \PTB & \IPLB & \PE \\
\midrule
$f^\pi$ (serial dictatorship)                     & \checkmark & \checkmark & \checkmark & \checkmark & \checkmark & \checkmark & \checkmark           \\
$f^{\neg \NW}$ &  & \checkmark & \checkmark & \checkmark & \checkmark & \checkmark & \\
$f^{\neg \PTB}$ (YRMH-IGYT)  & \checkmark & \checkmark & \checkmark & & & \checkmark & \checkmark \\
$f^{\neg \SP}$  & \checkmark & & \checkmark & \checkmark & \checkmark & \checkmark & \checkmark \\
$f^{\neg \NB}$  & \checkmark & \checkmark & & & \checkmark & & \checkmark \\
$f^{\neg \IPLB}$ (nonserial sequential dictatorship) & \checkmark & \checkmark & \checkmark & & \checkmark & & \checkmark \\
\bottomrule
\end{tabular}

\end{table}

\section{Extension: an outside option}
\label{sec:exten}

We now allow agents to remain unassigned, either because there are too few objects or because they prefer not to receive any available object. The characterizations in Theorems~\ref{thm0} and~\ref{thm1} extend once we impose \prop{individual rationality} (Theorems~\ref{thm:serial-dictatorship-null} and~\ref{thm:sequential-dictatorship-null}). The \prop{identical-preference lower bound} instead identifies a larger class of sequential dictatorships, defined by a bound on the number of real objects assigned before each agent is called (Proposition~\ref{prop:bounded-sequential}); without \prop{non-bossiness}, the other properties need not even imply a sequential dictatorship  (Example~\ref{ex:bounded-non-bossiness}). We first describe the changes to the model and properties, retaining the notation and conventions of the baseline model unless otherwise stated.

\subsection{Model}

The set of agents remains $N = \{1, \dots, n\}$. Let $O$ be a nonempty, finite set of \defterm{real objects}, and let $0 \notin O$ be a \defterm{null object} representing nonassignment. The set of all objects is denoted $O_0 \coloneqq O \cup \{0\}$. We impose no further restriction on $|O|$; in particular, there may be fewer real objects than agents. An \defterm{allocation} is a function $\mu : N \to O_0$ such that $\mu_i = \mu_j \in O$ implies $i = j$. Thus each real object has capacity one, while any number of agents may receive $0$. As before, $\mu(S) = \{\mu_i \mid i \in S\}$ is the image of $S \subseteq N$; in particular, it includes $0$ when some member of $S$ is unassigned. Let $\mathcal{M}$ denote the set of allocations in this extended model.

Each preference relation $R_i$ is a linear order on $O_0$, with strict part $P_i$. A real object $o$ is \defterm{acceptable} at $R_i$ if $o \mathrel{P_i} 0$. We impose no restriction on the position of $0$, so an agent may find any subset of the real objects acceptable. In this section, $\mathcal{R}$ denotes the set of these extended preferences, $\prof{R} = (R_i)_{i \in N} \in \mathcal{R}^N$ denotes a preference profile, and $f : \mathcal{R}^N \to \mathcal{M}$ denotes a mechanism. The remaining notation of Section~\ref{sec:model} applies to the enlarged domain.

\subsection{Properties}
The definitions of \prop{strategy-proofness}, \prop{group strategy-proofness}, \prop{non-bossiness}, \prop{Pareto efficiency}, \prop{pairwise efficiency}, \prop{pairwise tie-breaking}, and the \prop{identical-preference lower bound} apply verbatim to the enlarged domain and allocation set. \prop{Non-wastefulness} continues to refer to unassigned real objects: an allocation $\mu$ is \prop{non-wasteful} at $\prof{R}$ if $\mu_i \mathrel{R_i} o$ for every $i \in N$ and every $o \in O \setminus \mu(N)$, and a mechanism is \prop{non-wasteful} if it selects a \prop{non-wasteful} allocation at every profile. Unlike in the baseline model, $|O| = n$ does not make this requirement automatic, since some agents may remain unassigned while real objects are unused.

Non-wastefulness, as it is defined here, does not protect an agent against receiving an unacceptable object. We therefore add \defprop{individual rationality}: an allocation $\mu$ is \prop{individually rational} at $\prof{R}$ if $\mu_i \mathrel{R_i} 0$ for every $i \in N$.\footnote{Equivalently, one may treat the null object as always available and require $\mu_i \mathrel{R_i} o$ for every $i \in N$ and every $o \in (O \setminus \mu(N)) \cup \{0\}$; this single condition is the conjunction of \prop{non-wastefulness} and \prop{individual rationality}. We keep the two apart for comparability with \citet{pycia2022}, and because the proofs invoke them separately.}

\begin{property}[Individual rationality (\IR)]
    For every $\prof{R} \in \mathcal{R}^N$ and every $i \in N$, $f_i(\prof{R}) \mathrel{R_i} 0$.
\end{property}

\prop{Fair treatment of equals} also applies verbatim in the form (\ref{FTE tie-break formulation}): the implication is silent when the two agents receive the same object, so two identically reporting agents may both receive $0$. In Section~\ref{sec:property}, condition~(\ref{FTE tie-break formulation}) is equivalent to the priority formulation (\ref{priority exists}) because one member of every preference-identical pair must be strictly favored. With an outside option, a priority satisfying (\ref{priority exists}) still yields (\ref{FTE tie-break formulation}), but the converse can fail in the absence of further properties. Under \prop{individual rationality} and \prop{non-wastefulness}, however, (\ref{FTE tie-break formulation}) again yields a priority satisfying (\ref{priority exists}) (Lemma~\ref{lemma:priority-exists} in Appendix~\ref{appendix:proofs-extension}). 

Fact~\ref{fact} remains valid in the extended model; see \citet[Proposition~1]{demeulemeester2025pairwise}, whose capacity formulation includes the present model by giving $0$ capacity~$n$. The proofs below continue to invoke \prop{monotonicity} through \prop{strategy-proofness} and \prop{non-bossiness}.

\subsection{Serial and sequential dictatorships}

Serial and sequential dictatorships operate as before, except that nonassignment is always available. A \defterm{suballocation} is a function $\varsigma : S \to O_0$, for some $S \subseteq N$, that assigns each real object to at most one agent. We again identify a suballocation with its graph and write
$$ N(\varsigma) \coloneqq S, \qquad O(\varsigma) \coloneqq \varsigma(S) \cap O, \qquad |\varsigma| \coloneqq |S|.$$
Thus $N(\varsigma)$ includes every agent whose assignment is specified, including those assigned to $0$, whereas $O(\varsigma)$ includes only the assigned real objects. The set of objects remaining after $\varsigma$ is now 
\begin{equation}\label{eq:remaining-set}
    A(\varsigma) \coloneqq \big( O \setminus O(\varsigma) \big) \cup \{0\}.
\end{equation}
In this section and Appendix~\ref{appendix:proofs-extension}, $\mathcal{S}$ denotes the set of these suballocations. Then $\mathcal{M} = \{\varsigma \in \mathcal{S} \mid N(\varsigma) = N\}$. A dictator mapping is again a function $\kappa : \mathcal{S} \setminus \mathcal{M} \to N$ with $\kappa(\varsigma) \in N \setminus N(\varsigma)$. The recursion in Section~\ref{dictatorships} defines $i^\kappa_t(\prof{R})$, $\varsigma^\kappa_{t}(\prof{R})$, and reachability using the remaining object set (\ref{eq:remaining-set}). In particular, at step~$t$, agent~$i^\kappa_t(\prof{R}) = \kappa(\varsigma^\kappa_{t-1}(\prof{R}))$ receives $\max_{R_{i^\kappa_t(\prof{R})}} A (\varsigma^\kappa_{t-1}(\prof{R}))$ and hence
$$
\varsigma^{\kappa}_t(\prof{R}) \coloneqq \varsigma^\kappa_{t-1}(\prof{R}) \cup \left\{ \left(i^\kappa_t(\prof{R}), \max_{R_{i^\kappa_t(\prof{R})}} A \big(\varsigma^\kappa_{t-1}(\prof{R}) \big) \right) \right\}.
$$
An agent who chooses $0$ is not called again, but $0$ remains available to every subsequent agent. As before, the sequential dictatorship $f^\kappa$ is defined by $f^\kappa(\prof{R}) = \varsigma^\kappa_n(\prof{R})$ for every $\prof{R} \in \mathcal{R}^N$. The serial dictatorship $f^\pi$ equals a sequential dictatorship whose dictator mapping $\kappa$ satisfies (\ref{eq:serial-as-sequential}). Equivalently, for every $k \in \{1, \dots, n\}$ and every $\prof{R} \in \mathcal{R}^N$,
\begin{equation}\label{eq:serial-recursion}
    f^\pi_{\pi_k}(\prof{R}) = \max_{R_{\pi_k}} \left( \left( O \setminus \bigcup_{\ell=1}^{k-1}\{f^\pi_{\pi_\ell}(\prof{R})\} \right) \cup \{0\} \right),
\end{equation}
where the union is empty when $k = 1$.

\subsection{Extended results}
Theorems~\ref{thm0} and~\ref{thm1} extend once \prop{individual rationality} is also imposed.

\begin{theorem}\label{thm:serial-dictatorship-null}
    In the model with an outside option, a mechanism satisfies \prop{individual rationality}, \prop{non-wastefulness}, \prop{strategy-proofness}, and \prop{fair treatment of equals} if and only if it is a serial dictatorship.
\end{theorem}

The proof first recovers the priority formulation (\ref{priority exists}) of \prop{fair treatment of equals} (Lemma~\ref{lemma:priority-exists}) and then proceeds recursively down the priority as in the proof of Theorem~\ref{thm0}. If no remaining real object is acceptable to the next agent, then she is assigned $0$ by \prop{individual rationality}; otherwise, the baseline argument applies.

\begin{theorem}\label{thm:sequential-dictatorship-null}
    In the model with an outside option, a mechanism satisfies \prop{individual rationality}, \prop{non-wastefulness}, \prop{strategy-proofness}, \prop{non-bossiness}, and \prop{pairwise tie-breaking} if and only if it is a sequential dictatorship.
\end{theorem}

The only-if proof uses the same recursive construction as Theorem~\ref{thm1}: assuming that the mechanism acts as a sequential dictatorship for the first $k$ assignments, it identifies a next dictator after every reachable suballocation with~$k$ assignments. The argument identifying this dictator differs from that of Theorem~\ref{thm1} and is simpler. It considers profiles at which every remaining agent shares a common report with a single acceptable real object: \prop{individual rationality} and \prop{non-wastefulness} assign that object to exactly one of these agents, \prop{pairwise tie-breaking} shows that this agent is the same no matter which real object is ranked first, and \prop{strategy-proofness} and \prop{non-bossiness} then make her the next dictator at every profile that reaches the suballocation. In particular, \prop{pairwise efficiency} (Lemma~\ref{lemma:pe}) plays no role.

Consider a serial dictatorship modified only at its last step:
the last agent called receives her most-preferred remaining \emph{real} object
if one remains, and the null object otherwise.
This mechanism satisfies all properties of
Theorems~\ref{thm:serial-dictatorship-null}
and~\ref{thm:sequential-dictatorship-null}
except \prop{individual rationality}. With minor modifications, the independence examples for Theorems~\ref{thm0}--\ref{thm:iplb} also establish the independence of the remaining properties in Theorems~\ref{thm:serial-dictatorship-null} and~\ref{thm:sequential-dictatorship-null}; we omit the details.

\subsection{The identical-preference lower bound}\label{sec:extension-iplb}
Theorem~\ref{thm:iplb} does not permit the same extension. In particular, the \prop{identical-preference lower bound} no longer identifies serial dictatorships among the sequential dictatorships. The following example illustrates why the relevant constraint concerns the number of preceding real-object assignments, rather than the number of preceding turns.

\begin{example}[A nonserial sequential dictatorship satisfying \IPLB]\label{ex:nonserial-sequential-IPLB}
Let $N = \{1,2,3,4\}$ and $|O| \geq 4$, and fix an object
$a \in O$. For every profile $\prof{R} \in \mathcal{R}^N$, define $f(\prof{R}) = f^{\pi(\prof{R})}(\prof{R})$, where
$$
\pi(\prof{R}) = \begin{cases}
    (1, 2, 3, 4), & \text{if } \max_{R_1}(O_0) = 0,\\
    (1, 3, 4, 2), & \text{if } \max_{R_1}(O_0) = a \text{ and } \max_{R_3}(O_0 \setminus \{a\}) = 0,\\
    (1, 3, 2, 4), & \text{otherwise}.
\end{cases}
$$
Agent~$1$ is called first. If she receives $0$, the remaining order is fixed; otherwise, agent~$3$ is called next, and the order of agents~$2$ and~$4$ depends only on the first two assignments. Hence $f$ is a sequential dictatorship.

Let $\prof{R} \in \mathcal{R}^N$, $\pi^* = (1, 3, 2, 4)$, and let $i \coloneqq \pi^*_k$. At every profile, at most $k-1$ real objects are assigned before agent~$i$ is called. Hence, at the identical-preference profile $\prof{R}(R_i)$, agent~$i$ receives her $k$th-ranked real object if at least $k$ real objects are acceptable to her, and she receives $0$ otherwise. In the former case, at least one of her $k$ most-preferred real objects remains available when she is called at $\prof{R}$, so $f_i(\prof{R}) \mathrel{R_i}  f_i(\prof{R}(R_i))$. In the latter case, \prop{individual rationality} implies $f_i(\prof{R}) \mathrel{R_i} 0 = f_i(\prof{R}(R_i))$. Hence $f$ satisfies the \prop{identical-preference lower bound}.
\hfill $\diamond$
\end{example}

\subsubsection{Bounded sequential dictatorships}

Let $\pi \in \Pi$. A dictator mapping $\kappa$ is \defterm{$\pi$-bounded} if, for every $\kappa$-reachable suballocation $\varsigma$,
\begin{equation}\label{eq:bounded}
\{\pi_1, \dots, \pi_{|O(\varsigma)|} \} \subseteq N(\varsigma),
\end{equation}
where the set on the left is empty if $O(\varsigma) = \emptyset$.
A sequential dictatorship is \defterm{$\pi$-bounded} if it equals $f^\kappa$ for some $\pi$-bounded dictator mapping $\kappa$, and \defterm{bounded} if it is $\pi'$-bounded for some $\pi' \in \Pi$.

Thus $\kappa$ is $\pi$-bounded if, whenever $k$ real objects have already been assigned, agents $\pi_1, \dots, \pi_k$ have already been called. Equivalently, at every profile, agent~$\pi_k$ is called after at most $k-1$ real objects have been assigned. Her turn may thus be postponed beyond her priority rank, but if she is called at step $t > k$, at least $t - k$ of the preceding agents must have received $0$. 

Every serial dictatorship $f^\pi$ is $\pi$-bounded, since $N(\varsigma) = \{\pi_1, \dots, \pi_{|\varsigma|}\}$ for every reachable $\varsigma$ and $|O(\varsigma)| \leq |\varsigma|$. The mechanism of Example~\ref{ex:nonserial-sequential-IPLB} is $\pi$-bounded for $\pi = (1, 3, 2,4)$. Lemma~\ref{lemma:bounded-real-serial} in Appendix~\ref{appendix:proofs-extension} shows that a $\pi$-bounded sequential dictatorship calls the agents in the order $\pi_1, \pi_2, \dots$ as long as a real object remains available and every dictator receives a real object.

\begin{proposition}\label{prop:bounded-sequential}
    A sequential dictatorship satisfies the \prop{identical-preference lower bound} if and only if it is bounded.
\end{proposition}

Combining the proposition with Theorem~\ref{thm:sequential-dictatorship-null} yields the following characterization.

\begin{corollary}\label{cor:bounded-characterization}
    In the model with an outside option, a mechanism satisfies \prop{individual rationality}, \prop{non-wastefulness}, \prop{strategy-proofness}, \prop{non-bossiness}, \prop{pairwise tie-breaking}, and the \prop{identical-preference lower bound} if and only if it is a bounded sequential dictatorship.
\end{corollary}

Unlike in Theorem~\ref{thm:iplb}, \prop{non-bossiness} cannot be omitted. Example~\ref{ex:bounded-non-bossiness} in Appendix~\ref{appendix:independence} gives a mechanism satisfying all the other properties but not \prop{non-bossiness}, for any $|O| \geq 2$. An agent who chooses the null object can change the order of the others without changing her own assignment, while her opt-out leaves enough real objects available to preserve the \prop{identical-preference lower bound}.\footnote{The three-agent example exploits agent~$1$'s ranking of the objects she finds unacceptable, but this is not essential. With four agents and at least two objects, one can construct analogous mechanisms where the order in which agents are called depends only on the agents' rankings over acceptable objects.}

\section{Related literature and discussion}\label{sec:related}

\paragraph{Characterizations of serial dictatorships.} \citet{svensson1999} characterizes serial dictatorships by \prop{strategy-proofness}, \prop{non-bossiness}, and \prop{neutrality}. \prop{Non-wastefulness} is not imposed in his model, which permits more objects than agents but requires every agent to receive a real object. \citet{sen2026priority} weaken \prop{neutrality} to \prop{unanimous profile neutrality}, its restriction to identical-preference profiles, and identify the preference domains on which \prop{non-wastefulness}, \prop{strategy-proofness}, \prop{non-bossiness}, and this weaker requirement still characterize the serial dictatorships; the universal domain is one of them (Proposition~\ref{prop:another_characterization}). \citet{ergin2000} works in an equal-supply environment where the population may vary and characterizes serial dictatorships by \prop{weak Pareto efficiency} together with pairwise versions of \prop{consistency} and \prop{neutrality}; no incentive property is imposed in that characterization, although \prop{consistency} itself implies a form of \prop{non-bossiness} (footnote~\ref{fn:consistency}). Theorem~\ref{thm0} differs from all three in that it imposes \prop{strategy-proofness} as its only incentive requirement. 

\prop{Neutrality} is an object-side invariance requirement: problems that differ only in the labeling of the objects must have correspondingly relabeled solutions. Our tie-breaking properties are agent-side invariance requirements; they are silent about the object labels and regulate only the relative treatment of agents who report identical preferences. \prop{Neutrality} and \prop{fair treatment of equals} are logically independent,\footnote{The restricted serial dictatorship of Example~\ref{independence:NW} satisfies \prop{fair treatment of equals} but not \prop{neutrality}, since it treats the fixed subset $X$ of objects asymmetrically. Conversely, when $n \geq 3$, a mechanism that selects between two serial dictatorships according to whether two fixed agents report identical preferences is \prop{neutral} but violates \prop{fair treatment of equals}.} but both imply \prop{identical-preference fair treatment of equals} with respect to some priority, which under \prop{non-wastefulness} is equivalent to \prop{unanimous profile neutrality} (Section~\ref{sec:result}). The proofs of \citet{svensson1999}, of \citet{sen2026priority}, and of Proposition~\ref{prop:another_characterization} share a common architecture: \prop{neutrality}, its restriction to identical-preference profiles, or \prop{identical-preference fair treatment of equals} fixes the order of the agents across identical-preference profiles, and \prop{group strategy-proofness} propagates this order to arbitrary profiles. In \citeauthor{svensson1999}'s (\citeyear{svensson1999}) model, \prop{non-wastefulness} is not required, since \prop{neutrality} and \prop{group strategy-proofness} directly imply full \prop{Pareto efficiency}; \citet{sen2026priority} impose it when objects outnumber agents. The stronger incentive requirement is essential to this propagation.\footnote{\citet[p. 563]{svensson1999} constructs a mechanism that is \prop{strategy-proof}, \prop{neutral}, \prop{Pareto efficient}, and \prop{bossy}, which is not a serial dictatorship.} In Theorem~\ref{thm0}, where \prop{non-bossiness} is unavailable, we must instead exploit the broader reach of \prop{fair treatment of equals} to propagate the order to arbitrary profiles using \prop{strategy-proofness} alone; \prop{non-wastefulness} must be imposed separately (Example~\ref{independence:NW}).

\paragraph{Priority-based fairness and tie-breaking.}
\prop{Fair treatment of equals} resembles familiar priority-based relaxations of envy-freeness. \citet{svensson1994} studies \prop{weak fairness}, which requires a single priority to be respected at \emph{every} profile. With a priority common to all objects, this is the house-allocation analogue of \prop{no justified envy}, a common desideratum in school-choice problems \citep{abdulkadirouglu2003}. 
On the strict preference domain, \prop{weak fairness} and \prop{non-wastefulness} already determine the associated serial dictatorship profile by profile, with no incentive condition at all \citep[cf.][Lemma~1]{ergin2000}.\footnote{Suppose $f$ is \prop{non-wasteful} and \prop{weakly fair} with respect to the identity priority $\pi$. Let $\prof{R} \in \mathcal{R}^N$. Agent~$1$'s assignment is weakly $R_1$-preferred to every other agent's assignment (\prop{weak fairness}) and to every unassigned object (\prop{non-wastefulness}), hence to every object; so $f_1(\prof{R}) = \max_{R_1}(O)$. Inductively, agent~$k$'s assignment is weakly $R_k$-preferred to the assignment of every agent $j > k$ and to every unassigned object, so $f_k(\prof{R}) = \max_{R_k} \left(O \setminus \{f_1(\prof{R}), \dots, f_{k-1}(\prof{R})\}\right)$. Thus $f(\prof{R}) = f^\pi(\prof{R})$.} Theorem~\ref{thm0} shows that the much more local \prop{fair treatment of equals}, which restricts only pairs reporting identical preferences, has the same force once \prop{strategy-proofness} is imposed.

The relaxation is nevertheless substantive before incentive compatibility is imposed. For example, \prop{fair treatment of equals} is compatible with \prop{rank efficiency} \citep{featherstone2020rank}.\footnote{An allocation $\mu$ is \defprop{rank efficient} if no other allocation weakly increases, for every $k$, the number of agents receiving one of their $k$ most-preferred objects and strictly increases it for some $k$ \citep{featherstone2020rank}. \prop{Rank-efficient} allocations are \prop{Pareto efficient} and hence \prop{non-wasteful}.} Indeed, swapping the assignments of two preference-identical agents leaves the distribution of assigned ranks unchanged, so at every profile there exists a rank-efficient allocation that orders each group of preference-identical agents by any given priority $\pi$, and selecting one such allocation yields a rank-efficient mechanism satisfying \prop{fair treatment of equals}. \prop{Weak fairness}, by contrast, can force a rank-dominated allocation.\footnote{\label{fn:rank-efficient-example}\prop{Weak fairness} and \prop{non-wastefulness} force the serial dictatorship associated with the priority. For example, let $N = \{1, 2, 3\}$, $O = \{a, b, c\}$, $\pi = (1, 2, 3)$, $R_1 = R_2 : a, b, c$, and $R_3 : b, a, c$. The $\pi$-serial dictatorship selects $(a, b, c)$, which assigns one first choice, one second choice, and one third choice. The allocation $(a, c, b)$ assigns two first choices and one third choice; hence it rank-dominates $(a, b, c)$. Note that $(a, c, b)$ is compatible with \prop{fair treatment of equals} as it respects $\pi$ for the only identical-preference pair $\{1,2\}$.}

In school choice, priorities are usually primitive data, and a related literature asks how constraints generated by those priorities may be relaxed or redesigned. \citet{dur2019} specify which priority violations are admissible and study assignments satisfying the resulting \prop{partial fairness}. \citet{ehlers2020legal} count a priority claim as relevant only when it can be redressed by another valid assignment, leading to their setwise notion of \defterm{legal} assignments. \citet{celebi2023} instead design priorities that adapt to the composition of the eventual allocation in order to balance scores and diversity. Closest to our approach, \citet{dougan2018} characterize Immediate Acceptance mechanisms based on endogenously determined priorities using, among other properties, an invariance requirement on the resolution of pairwise ties called \prop{weak uniform tie-breaking}.

\paragraph{Group-strategy-proof and efficient mechanisms.}
As discussed in the introduction, the serial and sequential dictatorships are focal subclasses of the incentive-compatible mechanisms studied by \citet{papai2000} and \citet{pycia2017}. \citet{pycia2017} characterize the \prop{group strategy-proof} and \prop{Pareto efficient} mechanisms as the trading cycles mechanisms, in which agents trade from (artificial) endowments determined by a structure of ownership and brokerage rights; \citet{papai2000} characterizes the subclass without brokerage---the hierarchical exchange mechanisms---using \prop{reallocation-proofness} as an additional property. The sequential dictatorships are the especially simple trading cycles mechanisms in which, after every realized suballocation, a single agent owns all remaining objects \citep{pycia2017,pycia2023}. Theorem~\ref{thm1} shows that \prop{pairwise tie-breaking} alone separates the sequential dictatorships from all other trading cycles mechanisms.

Lemma~\ref{lemma:pe} also connects the result to recent work on weaker efficiency tests. \citet{demeulemeester2025pairwise} show that \prop{monotonic} mechanisms are \prop{Pareto efficient} if and only if they are \prop{pairwise efficient} and \prop{non-wasteful}. As noted after Theorem~\ref{thm1}, their result and Lemma~\ref{lemma:pe} provide an alternative route to \prop{Pareto efficiency}, though our proofs invoke neither \prop{Pareto efficiency} nor the representation as a trading cycles mechanism directly.

\paragraph{Characterizations of sequential dictatorships.}
\citet{pycia2023} offer a complementary characterization of sequential dictatorships in a general setting with unit demand. They ask when a mechanism can be rationalized by a strict social welfare function satisfying completeness, unanimity, and independence of irrelevant alternatives, and they show that the only \prop{strategy-proof} mechanisms satisfying this global coherence requirement are the almost sequential dictatorships---and exactly the sequential dictatorships when objects outnumber agents \citep[cf.][]{gibbard1973,satterthwaite1975}. Their rationalizability condition is unlike other properties considered here: it restricts the mechanism indirectly by imposing constraints on an underlying social welfare function.

\citet{fujinaka2011secure} characterize the sequential dictatorships when $|O| = n$: they are the only \prop{Pareto efficient} mechanisms that are \defprop{securely implementable} in the sense of \citet{saijo2007secure}. Secure implementability is equivalent to \prop{strategy-proofness} together with the \defprop{rectangular property}, which strengthens \prop{non-bossiness}.\footnote{Formally, a mechanism $f$ satisfies the \prop{rectangular property} if, for all $\prof{R}, \prof{R}' \in \mathcal{R}^N$, $f(\prof{R}) \neq f(\prof{R}')$ implies $f_i(\prof{R}) \neq f_i(R'_i, \prof{R}_{-i})$ for some $i \in N$. \prop{Non-bossiness} is the special case in which $\prof{R}$ and $\prof{R}'$ differ in exactly one agent's report.} Thus, when $|O| = n$, the \prop{rectangular property} separates the sequential dictatorships from the other trading cycles mechanisms, as \prop{pairwise tie-breaking} does in Theorem~\ref{thm1}. The two properties restrict different aspects of a mechanism: \prop{rectangularity} restricts how it responds to joint changes in the reports of several agents, while \prop{pairwise tie-breaking} restricts the relative treatment of identically reporting agents; neither property implies the other.

With an outside option, \citet{pycia2022} characterize the \defterm{binary serial dictatorships} by \prop{individual rationality}, \prop{non-wastefulness}, \prop{group strategy-proofness}, and \prop{neutrality}: these are the sequential dictatorships in which the next dictator depends on which agents have already been called and which of them received the null object, but not on the identities of any preceding real-object assignments.
Boundedness in Section~\ref{sec:exten} imposes a different restriction: it limits the number of real objects assigned before an agent is called, rather than the information used to determine the next dictator. The mechanism in Example~\ref{ex:nonserial-sequential-IPLB} is a bounded sequential dictatorship that fails \prop{neutrality},\footnote{Whether agent~$4$ is called before agent~$2$ depends on whether agent~$1$ receives the particular object~$a$.} so it is not a binary serial dictatorship.

A separate literature studies environments in which agents may receive a bundle of objects. On the unrestricted domain of preferences over bundles, \citet{papai2001} characterizes the sequential dictatorships by \prop{strategy-proofness}, \prop{non-bossiness}, and \prop{citizen sovereignty}.\footnote{\prop{Citizen sovereignty} requires that every feasible allocation is selected at some profile. On an unrestricted domain, it is implied by \prop{Pareto efficiency}. Conversely, \prop{group strategy-proofness} and \prop{citizen sovereignty} directly imply \prop{Pareto efficiency}.} Related characterizations are obtained under responsive preferences \citep{ehlers2003scw}, under fixed bundle sizes on domains containing the additive preferences \citep{hatfield2009strategy}, and under bounded bundle sizes with all strict preferences \citep{root2020incentives}.  The contrast with unit demand is striking: \prop{group strategy-proofness} and \prop{Pareto efficiency}, which characterize the sequential dictatorships under multi-unit demand, characterize the much larger class of trading cycles mechanisms under unit demand \citep{pycia2017}. The reason is that, under unit demand, an agent who trades leaves the market, so exchange from dispersed endowments is \prop{group strategy-proof} \citep[p.~1425]{papai2000}; under multi-unit demand, an agent who trades one object may influence her own later trading opportunities, and---when the domain is rich enough---\prop{group strategy-proofness} forces all ownership to be concentrated on a single agent at each step.

\paragraph{Lower bounds.}
The \prop{identical-preference lower bound} belongs to a tradition of lower-bound guarantees originating in fair division \citep{Steinhaus1948}; see \citet{moulin1990}, \citet{bevia1996}, and \citet{fujinaka2007} for related bounds in models with indivisible objects, and \citet{thomson2023} for a survey. In a model in which agents may remain unassigned, \citet[Theorem~1]{hashimoto2018} characterizes the mechanisms satisfying \prop{strategy-proofness}, \prop{Pareto efficiency}, and the \prop{identical-preference lower bound} as the Top Trading Cycles mechanisms associated with U-symmetric inheritance structures, a class containing all YRMH--IGYT mechanisms and, in particular, the serial dictatorships.\footnote{The serial dictatorship associated with $\pi$ is the YRMH--IGYT mechanism with priority $\pi$ in which every agent is a newcomer, so that no object has an existing tenant.} Theorem~\ref{thm:iplb} shows that, in our model, adding \prop{pairwise tie-breaking} leaves only the serial dictatorships even when \prop{Pareto efficiency} is weakened to \prop{non-wastefulness}.

\section*{Statements and declarations}

\paragraph{Competing interests.}
The authors have no relevant financial or non-financial interests to disclose.

\paragraph{Data availability.}
No datasets were used or produced in this theoretical study.

\paragraph{Generative AI use.}
During the preparation of this work, the authors used ChatGPT, Claude, and Refine.ink to check the exposition for clarity and for consistency of notation and arguments. After using these tools, the authors reviewed and edited the content as needed and take full responsibility for the content of the published article.

\bibliographystyle{chicago}
\bibliography{references}

\appendix

\section{Independence examples}\label{appendix:independence}

This appendix contains the formal constructions summarized in Section~\ref{sec:independent}. For each property appearing among Theorems~\ref{thm0},~\ref{thm1}, and~\ref{thm:iplb}, we specify a mechanism that violates it while satisfying the remaining properties in the relevant theorem. 

Our first example is defined for $|O| > n$.

\begin{example}[\prop{Non-wastefulness}]\label{independence:NW} \ \\
If $|O| = |N|$, then every mechanism is \prop{non-wasteful}. Assume $|O| > |N|$. Fix a subset $X \coloneqq \{o_1,\ldots,o_n\} \subsetneq O$, and let $f^{\neg \NW}$ be the restricted $\pi$-serial dictatorship that only assigns objects from $X$.  Formally, for every $\prof{R} \in \mathcal{R}^N$, $f^{\neg \NW}_{\pi_1}(\prof{R}) = \max_{R_{\pi_1}}(X)$ and, for each $k \in \{2,\dots,n\}$, 
$$
f^{\neg \NW}_{\pi_k}(\prof{R}) = \max_{R_{\pi_k}}\left(X \setminus \bigcup_{\ell = 1}^{k-1} \{f^{\neg \NW}_{\pi_\ell}(\prof{R})\} \right).
$$

It is easy to see that $f^{\neg \NW}$ is \prop{wasteful}, but it satisfies \prop{strategy-proofness}, \prop{non-bossiness}, \prop{fair treatment of equals} (hence \prop{pairwise tie-breaking}), and \prop{identical-preference lower bound}. Therefore \prop{non-wastefulness} is genuinely needed in Theorems~\ref{thm0},~\ref{thm1}, and~\ref{thm:iplb} whenever $|O| > |N|$.
\hfill~$\diamond$
\end{example}

The remaining examples are defined for all environments with $|O| \geq n$; Examples~\ref{ex:notSP}, \ref{exampleNB}, and~\ref{independence:IPLB} assume $n \geq 3$ (see Section~\ref{sec:independent}).

\begin{example}[\prop{Pairwise tie-breaking}]\label{ex: notPTB}\ \\
Let $X \coloneqq \{x_1, \dots ,x_n\} \subseteq O$, and let $\omega$ be the allocation with $\omega_i = x_i$ for all $i \in N$. Let $f^{\neg \PTB}$ be the \defterm{You Request My House--I Get Your Turn (YRMH-IGYT)} mechanism \citep{abdulkadirouglu1999} associated with the initial allocation $\omega$ and the identity priority $\pi$. That is, agent~$i$ is declared the existing tenant of object $x_i$, and property rights over any objects in $O \setminus X$ are granted to agents based on the priority order $\pi$.\footnote{If $X = O$, this is the usual Top Trading Cycles mechanism for the housing market with endowment $\omega$.}

Because $f^{\neg \PTB}$ is a hierarchical exchange mechanism, it is \prop{Pareto efficient} and \prop{group strategy-proof} \citep{papai2000,pycia2017}. Moreover, \citet{hashimoto2018} shows that $f^{\neg \PTB}$ satisfies the \prop{identical-preference lower bound}.

To see that $f^{\neg \PTB}$ fails \prop{pairwise tie-breaking}, consider the identical-preference profile $\prof{R} \in \mathcal{R}^N$ where every agent reports $R_1 : x_1, x_2, \dots, x_n, \dots$. Then $f^{\neg \PTB}(\prof{R}) = \omega$ and hence 
$$
f^{\neg \PTB}_1(\prof{R}) \mathrel{P_1} f^{\neg \PTB}_2(\prof{R}).
$$
Now let  $R'_1 = R'_2 :x_2, x_1, \dots$ and denote $\prof{R}' \coloneqq (R'_1, R'_2 ,\prof{R}_{N \setminus \{1,2\}})$. Then $f^{\neg \PTB}(\prof{R}') = \omega$ and hence 
$$
f^{\neg \PTB}_2(\prof{R}')  \mathrel{P'_1} f^{\neg \PTB}_1(\prof{R}').$$
This violates \prop{pairwise tie-breaking} (and hence \prop{fair treatment of equals}). Thus \prop{pairwise tie-breaking} is indispensable in Theorem~\ref{thm1} and in Theorem~\ref{thm:iplb}. Likewise, \prop{fair treatment of equals} is essential in Theorem~\ref{thm0}.
\hfill~$\diamond$
\end{example}

\begin{example}[\prop{Strategy-proofness}]\label{ex:notSP}\ \\
Assume $n \geq 3$, and fix the identity priority $\pi$. For each $R_i \in \mathcal{R}$ and each $x \in O$, let
$
\rho(x, R_i) \coloneqq \left|\{ o \in O \mid o \mathrel{R_i} x \} \right|
$
be the rank of~$x$ under $R_i$. For each profile $\prof{R} \in \mathcal{R}^N$, let
$$
\mathcal{L}(\prof{R}) \coloneqq \{\mu \in \mathcal{M} \mid \rho(\mu_i, R_i) \leq i \text{ for all } i \in N\}.
$$
Thus $\mathcal{L}(\prof{R})$ is the set of allocations such that every agent~$i$ receives an object weakly better than the rank-$i$ object that she would receive at the identical-preference profile $\prof{R}(R_i)$ under the serial dictatorship $f^\pi$. Note that $\mathcal{L}(\prof{R})$ is nonempty because it contains $f^\pi(\prof{R})$.

For each $\mu \in \mathcal{M}$ and each rank~$k \in \{1,\dots, |O|\}$, let
$$
\#_k(\mu, \prof{R}) \coloneqq \left| \{i \in N \mid \rho(\mu_i, R_i) = k \} \right|
$$
be the number of agents who receive their rank-$k$ object at $\mu$ under $\prof{R}$. Define the score function $s$ from $\mathcal{M} \times \mathcal{R}^N$ to  $\mathbb{Z}^{|O| + n}$, where the score of $\mu$ at $\prof{R}$ is
$$
s(\mu, \prof{R}) \coloneqq \left( \#_1(\mu,\prof{R}), \dots, \#_{|O|}(\mu,\prof{R}),-\rho(\mu_1, R_1), \dots, -\rho(\mu_n, R_n) \right).
$$

Let $\geq_{\mathrm{lex}}$ denote the lexicographic order on $\mathbb{Z}^{|O|+n}$. For every $\prof{R} \in \mathcal{R}^N$, define $f^{\neg \SP}(\prof{R})$
as the unique allocation that lexicographically maximizes $s(\cdot, \prof{R})$ on $\mathcal{L}(\prof{R})$. Formally, $f^{\neg \SP}(\prof{R}) \in \mathcal{L}(\prof{R})$ is such that 
$$s(f^{\neg \SP}(\prof{R}), \prof{R}) \geq_{\mathrm{lex}} s(\nu, \prof{R}) \quad \text{for every } \nu \in \mathcal{L}(\prof{R}).$$
Uniqueness holds because $\geq_{\mathrm{lex}}$ is a linear order and preferences are strict.

By construction, $f^{\neg \SP}$ satisfies the \prop{identical-preference lower bound}. It is \prop{Pareto efficient}: if $\nu$ Pareto dominates $f^{\neg \SP}(\prof{R})$ at $\prof{R}$, then $\nu \in \mathcal{L}(\prof{R})$ and $s(\nu, \prof{R}) >_{\mathrm{lex}} s(f^{\neg \SP}(\prof{R}), \prof{R})$, a contradiction. It also satisfies \prop{fair treatment of equals} with respect to $\pi$: if $R_i = R_j$, $i \mathrel{\pi} j$, and $f^{\neg \SP}_j(\prof{R}) \mathrel{P_i} f^{\neg \SP}_i(\prof{R})$, then swapping the assignments of~$i$ and $j$ yields an allocation $\nu \in \mathcal{L}(\prof{R})$ with $s(\nu, \prof{R}) >_{\mathrm{lex}} s(f^{\neg \SP}(\prof{R}), \prof{R})$, a contradiction. Finally, one can show that $f^{\neg \SP}$ is \prop{non-bossy}.\footnote{To prove \prop{non-bossiness}, fix $\prof{R} \in \mathcal{R}^N$, $i \in N$, and $R'_i \in \mathcal{R}$, and write $\prof{R}' \coloneqq (R'_i, \prof{R}_{-i})$. Suppose $f_i^{\neg \SP}(\prof{R}) = f^{\neg \SP}_i(\prof{R}') = o$. Let $\mathcal{A}(o) \coloneqq \{\mu \in \mathcal{M} \mid \mu_i = o\}$.

Since $f^{\neg \SP}(\prof{R}) \in \mathcal{L}(\prof{R}) \cap \mathcal{A}(o)$ and $f^{\neg \SP}(\prof{R}') \in \mathcal{L}(\prof{R}') \cap \mathcal{A}(o)$, we have $\rho(f^{\neg \SP}_i(\prof{R}'), R_i) = \rho(f^{\neg \SP}_i(\prof{R}), R_i) \leq i$ and  $\rho(f^{\neg \SP}_i(\prof{R}), R'_i) = \rho(f^{\neg \SP}_i(\prof{R}'),R'_i) \leq i$. Hence 
\begin{equation}\label{Useful containment}
f^{\neg \SP}(\prof{R}') \in \mathcal{L}(\prof{R}) \cap \mathcal{A}(o) \quad \text{and} \quad f^{\neg \SP}(\prof{R}) \in \mathcal{L}(\prof{R}') \cap \mathcal{A}(o).
\end{equation}
Because agent~$i$ receives object~$o$ at every allocation in $\mathcal{A}(o)$, there is a constant vector $\mathbf{c} \in \mathbb{Z}^{|O|+n}$ such that
$$
s(\mu, \prof{R}') = s(\mu, \prof{R}) + \mathbf{c} \quad \text{for all } \mu \in \mathcal{A}(o).
$$
By (\ref{Useful containment}) and the definition of $f^{\neg \SP}$, it follows that 
$$
s(f^{\neg \SP}(\prof{R}'), \prof{R}') \geq_{\mathrm{lex}} s(f^{\neg \SP}(\prof{R}), \prof{R}') = s(f^{\neg \SP}(\prof{R}), \prof{R}) + \mathbf{c} \geq_{\mathrm{lex}} s(f^{\neg \SP}(\prof{R}'), \prof{R}) + \mathbf{c} = s(f^{\neg \SP}(\prof{R}'), \prof{R}').
$$
Since $\geq_{\mathrm{lex}}$ is a linear order on $\mathbb{Z}^{|O|+n}$, we must have $s(f^{\neg \SP}(\prof{R}'), \prof{R}') = s(f^{\neg \SP}(\prof{R}), \prof{R}')$ and hence $f^{\neg \SP}(\prof{R}') = f^{\neg \SP}(\prof{R})$. 
} 

We show that $f^{\neg \SP}$ is not \prop{strategy-proof} when $n \geq 3$. Fix distinct objects $a,b,c,o_4,\dots,o_n \in O$, and take $\prof{R} \in \mathcal{R}^N$ with
$$
R_1 : a,b,c,\dots, \quad R_2 : b,c,a,\dots, \quad R_3 : a,b,c,\dots, \quad \text{and} \quad \max_{R_i}(O) = o_i \text{ for all }i \in N \setminus \{1,2,3\}.
$$
Then
$f^{\neg \SP}(\prof{R}) = (a,b,c, o_4, \dots, o_n)$.
However, if instead agent~$3$ reports $R'_3 : b, a, c,\dots$, then
$$
f^{\neg \SP}(R'_3, \prof{R}_{-3}) = (a,c,b,o_4,\dots,o_n),
$$
which violates \prop{strategy-proofness} because $b \mathrel{P_3} c$. Thus \prop{strategy-proofness} is essential in Theorems~\ref{thm0},~\ref{thm1}, and~\ref{thm:iplb}.
\hfill~$\diamond$ \label{exampleSP}
\end{example}


\begin{example}[\prop{Non-bossiness}]\label{exampleNB}\ \\
Assume~$n \geq 3$. Let $\pi=(1,2,\dots,n)$ be the identity priority and let $\pi' = (1, 3, 2, 4, \dots, n)$ be the priority obtained by interchanging the order of agents~$2$ and~$3$. Fix $x \in O$ and define $f^{\neg \NB}$ by
$$
f^{\neg \NB}(\prof{R}) = \begin{cases}
f^\pi(\prof{R}), & \text{if } R_1 \text{ ranks object } x \text{ second,} \\
f^{\pi'}(\prof{R}), & \text{otherwise}.
\end{cases}
$$
 Then $f^{\neg \NB}$ satisfies \prop{non-wastefulness}, \prop{strategy-proofness}, and \prop{pairwise tie-breaking}. 
 
 To see that $f^{\neg \NB}$ fails \prop{non-bossiness}, consider a preference profile $\prof{R} \in \mathcal{R}^N$ with
 $$
 R_1 : a, x, b,\dots, \quad R_2 = R_3 : b, x, a,\dots,
 $$
 and let $R'_1 : a, b, x, \dots$.
 Then 
 $$
 f^{\neg \NB}(\prof{R}) = (a, b, x, \dots) \quad \text{and} \quad f^{\neg \NB}(R'_1, \prof{R}_{-1}) = (a, x, b, \dots),
 $$
 which violates \prop{non-bossiness}. Hence \prop{non-bossiness} is essential in Theorem~\ref{thm1}. \hfill~$\diamond$
\end{example}

\begin{example}[\prop{Identical-preference lower bound}]\label{independence:IPLB}\ \\
Assume~$n \geq 3$. Let $f^{\neg \IPLB}$ be a nonserial sequential dictatorship, i.e., a sequential dictatorship that is not a serial dictatorship. For concreteness, fix $o_1 \in O$ and assume that $f^{\neg \IPLB}$ is a sequential dictatorship whose dictator mapping $\kappa$ satisfies
$$ \kappa(\varsigma) = 
\begin{cases}
1, & \text{if } \varsigma = \varnothing,\\
|\varsigma| + 1, & \text{if } (1, o_1) \in \varsigma \text{ and } |\varsigma| + 1 \in N \setminus N(\varsigma), \\
n - |\varsigma| + 1, & \text{if } (1, x) \in \varsigma \text{ for some } x\in O \setminus \{o_1\} \text{ and } n - |\varsigma| + 1 \in N \setminus N(\varsigma) \\
\min \left({N \setminus N(\varsigma)}\right), & \text{otherwise.}
\end{cases}
$$
In other words, agent~1 is first assigned her most-preferred object. If agent~1 receives~$o_1$, then the remaining agents are assigned objects in the order $2, \dots, n$; otherwise they are assigned in the order $n, n-1, \dots, 2$. By Theorem~\ref{thm1}, $f^{\neg \IPLB}$ satisfies \prop{non-wastefulness}, \prop{strategy-proofness}, \prop{non-bossiness}, and 
\prop{pairwise tie-breaking}. 

Let $\overline{\prof{R}} \in \mathcal{R}^N$ be an identical-preference profile with $\overline{R}_1 : o_1, o_2, \dots, o_{|O|}$. Then $i^\kappa_2(\overline{\prof{R}}) = 2$, which implies that $f^{\neg \IPLB}_2(\overline{\prof{R}}) = o_2$. However, if $R_1 : o_2, o_1, \dots$, then $i^\kappa_{n}(R_1, \overline{\prof{R}}_{-1}) = 2$ and hence $f^{\neg \IPLB}_2(R_1, \overline{\prof{R}}_{-1}) =o_n$. Because $n \geq 3$, this means $f^{\neg \IPLB}_2(\overline{\prof{R}}) \mathrel{\overline{P}_2} f^{\neg \IPLB}_2(R_1, \overline{\prof{R}}_{-1})$, so $f^{\neg \IPLB}$ fails the \prop{identical-preference lower bound}. Hence the \prop{identical-preference lower bound} is essential in Theorem~\ref{thm:iplb}.
\hfill~$\diamond$
\end{example}

The following example uses the model and properties of Section~\ref{sec:exten}.

\begin{example}[\prop{Non-bossiness is essential in Corollary~\ref{cor:bounded-characterization}}]\label{ex:bounded-non-bossiness}\ \\
    In the outside-option model, let $N = \{1, 2, 3\}$ and $|O| \geq 2$. Fix distinct objects $a, b \in O$. For every profile $\prof{R} \in \mathcal{R}^N$, define $f(\prof{R}) = f^{\pi(\prof{R})}(\prof{R})$, where
    $$
    \pi(\prof{R}) = \begin{cases}
        (1, 3, 2), & \text{if } \max_{R_1}(O_0) = 0 \text{ and } b \mathrel{P_1} a,\\
        (1, 2, 3), & \text{otherwise}.
    \end{cases}
    $$
    The mechanism $f$ satisfies \prop{individual rationality}, \prop{non-wastefulness}, \prop{strategy-proofness}, \prop{pairwise tie-breaking}, and the \prop{identical-preference lower bound}, but it is \prop{bossy}: if $R_1 : 0, a, b, \dots$, $R'_1 : 0, b, a, \dots$, and $R_2 = R_3 : a, b, 0, \dots$, then the outcomes at $\prof{R}$ and $(R'_1, \prof{R}_{-1})$ are $(0, a, b)$ and $(0, b, a)$, respectively. Agent~$1$ keeps $0$, while agents~$2$ and~$3$ exchange $a$ and $b$.\hfill $\diamond$  
\end{example}

\section{Proofs} \label{appendix:proof}

\subsection{Auxiliary properties and results}

We first introduce an invariance property used in the proofs. For each $R_i \in \mathcal{R}$ and $o \in O$, let $L(o, R_i) \coloneqq \{o' \in O \mid o \mathrel{R_i} o'\}$ be the \defterm{lower contour set of $R_i$ at $o$}. A preference $R'_i$ is a \defterm{monotonic transformation of $R_i$ at $o$} if $L(o, R_i) \subseteq L(o, R'_i)$. For each $\prof{R} \in \mathcal{R}^N$ and $\mu \in \mathcal{M}$, a profile $\prof{R}'$ is a \defterm{monotonic transformation of $\prof{R}$ at $\mu$} if $R'_i$ is a monotonic transformation of $R_i$ at $\mu_i$ for every $i \in N$.\footnote{In the model of Section~\ref{sec:exten}, $O_0$ replaces $O$ in these definitions, and $\mathcal{R}$ and $\mathcal{M}$ have their extended meanings. Everything in this subsection, including Fact~\ref{fact}, applies there without change.}

\begin{property}[Monotonicity]
    For every $\prof{R} \in \mathcal{R}^N$ and every monotonic transformation $\prof{R}'$ of $\prof{R}$ at $f(\prof{R})$, we have $f(\prof{R}')=f(\prof{R})$.    
\end{property}

The following standard equivalences are useful throughout the proofs \citep{papai2000,takamiya2001,alva2017,pycia2023}.

\begin{fact}\label{fact}
	The following properties of a mechanism $f$ are equivalent:
	\begin{enumerate}
		\item \prop{group strategy-proofness}.
		\item \prop{strategy-proofness} and \prop{non-bossiness}.
		\item \prop{monotonicity}.
	\end{enumerate}
\end{fact}

\subsection{Proof of Lemma~\ref{lemma:pe}}
Let $f$ be a mechanism satisfying \prop{pairwise tie-breaking}, \prop{strategy-proofness}, and \prop{non-bossiness}. By Fact~\ref{fact}, $f$ satisfies \prop{monotonicity}.
    
Suppose toward contradiction that $f$ is not \prop{pairwise efficient}. Then there exist $\prof{R} \in \mathcal{R}^N$ and $i,j \in N$ such that 
$$
f_i(\prof{R}) \mathrel{P_j} f_j(\prof{R}) \quad \text{and} \quad f_j(\prof{R}) \mathrel{P_i} f_i(\prof{R}).
$$
Denote $\mu \coloneqq f(\prof{R})$, and consider preferences 
$$
\hat{R}_i : \mu_j, \mu_i, \dots \quad \text{and} \quad \hat{R}_j : \mu_i, \mu_j, \dots$$
for agents~$i$ and~$j$. Then $(\hat{R}_i, \hat{R}_j,\prof{R}_{N \setminus \{i,j\}})$ is a monotonic transformation of $\prof{R}$ at $\mu$. Hence \prop{monotonicity} implies $f(\hat{R}_i, \hat{R}_j, \prof{R}_{N \setminus \{i,j\}}) = \mu$.

Now consider $\bar{R}_i \coloneqq \hat{R}_j : \mu_i, \mu_j,\dots$. Because $\bar{R}_i$ is a monotonic transformation of $\hat{R}_i$ at $\mu_i$, \prop{monotonicity} implies $f(\bar{R}_i,\hat{R}_j,\prof{R}_{N \setminus \{i,j\}}) = \mu$. Since $\mu_i \mathrel{\bar{P}_i} \mu_j$, we have
\begin{equation}\label{first PTB relation}
\bar{R}_i = \hat{R}_j \quad \text{and} \quad f_i(\bar{R}_i,\hat{R}_j,\prof{R}_{N \setminus \{i,j\}}) \mathrel{\bar{P}_i} f_j(\bar{R}_i,\hat{R}_j,\prof{R}_{N \setminus \{i,j\}}).
\end{equation}

Now consider $\bar{R}_j \coloneqq \hat{R}_i$. Because $\bar{R}_j$ is a monotonic transformation of $\hat{R}_j$ at $\mu_j$, \prop{monotonicity} implies $f(\hat{R}_i,\bar{R}_j,\prof{R}_{N \setminus \{i,j\}}) = \mu$. Since $\mu_j \mathrel{\bar{P}_j} \mu_i$, we have
\begin{equation}\label{second PTB relation}
\bar{R}_j = \hat{R}_i \quad \text{and} \quad f_j(\hat{R}_i,\bar{R}_j,\prof{R}_{N \setminus \{i,j\}}) \mathrel{\bar{P}_j} f_i(\hat{R}_i,\bar{R}_j,\prof{R}_{N \setminus \{i,j\}}).
\end{equation}
Taken together, (\ref{first PTB relation}) and (\ref{second PTB relation}) show that $f$ fails \prop{pairwise tie-breaking}. \hfill \qed

\subsection{Proof of Theorem~\ref{thm0}}

Every serial dictatorship is \prop{non-wasteful}, \prop{strategy-proof}, and satisfies \prop{fair treatment of equals} with respect to its associated priority.

It remains to prove the only-if direction. Let $f$ be a mechanism satisfying \prop{non-wastefulness}, \prop{strategy-proofness}, and \prop{fair treatment of equals} with respect to some priority $\pi$. We show that $f = f^\pi$. Without loss of generality, let $\pi$ be the identity priority.

We show by induction on $k$ that
$$
\text{for every } \prof{R} \in \mathcal{R}^N, \quad f_k(\prof{R}) = \max_{R_k}\left(O \setminus \bigcup^{k-1}_{i = 1}\{f_i(\prof{R})\}\right)
$$
holds for $k = 1, \dots, n$.

\paragraph{Base case.} For $k = 1$, let $\prof{R} \in \mathcal{R}^N$ and $o \coloneqq \max_{R_1}(O)$. If $o$ were unassigned at $f(\prof{R})$, then agent~$1$ would prefer an unassigned object to her assignment, contradicting \prop{non-wastefulness}. So object $o$ is assigned to some agent.

Suppose $f_j(\prof{R}) = o$ for some agent $j > 1$. Let $R'_j \coloneqq R_1$ and consider the profile $\prof{R}' \coloneqq (R'_j , \prof{R}_{-j})$. \prop{Strategy-proofness} implies that $f_j(\prof{R}') \mathrel{R'_j} f_j(\prof{R}) = o$, and since $o = \max_{R'_j} (O)$ we must have $f_j(\prof{R}') = o$.
As a result, $f_1(\prof{R}
')\neq o$ and hence $f_j(\prof{R}') \mathrel{P'_1} f_1(\prof{R}') $.
However, since $R'_1 = R_1=R'_j$, \prop{fair treatment of equals} and $1 \mathrel{\pi} j$ gives $f_1(\prof{R}') \mathrel{R'_1} f_j(\prof{R}')$, a contradiction.

It follows that $f_1(\prof{R}) = o$, and hence $f_1(\prof{R}) = \max_{R_1}(O)$.

\paragraph{Induction hypothesis.} Let $k \geq 2$, and assume that
$$
\text{for every } \prof{R} \in \mathcal{R}^N, \quad f_\ell(\prof{R}) = \max_{R_\ell}\left(O \setminus \bigcup^{\ell-1}_{i = 1}\{f_i(\prof{R})\}\right)
$$
holds for $\ell = 1, \dots, k-1$.

\paragraph{Inductive step.} 
Let $\prof{R} \in \mathcal{R}^N$,
$A \coloneqq O \setminus \bigcup^{k-1}_{i = 1}\{f_i(\prof{R})\}$, and $o \coloneqq \max_{R_k}(A)$. We show that $f_k(\prof{R}) = o = \max_{R_k}(A)$.

If $o$ were unassigned at $f(\prof{R})$, then agent~$k$ would prefer an unassigned object to her assignment, contradicting \prop{non-wastefulness}. So object $o$ is assigned to some agent, and since $o \notin \{f_1(\prof{R}) ,\dots, f_{k-1}(\prof{R})\}$, it is assigned to an agent in $\{k,\dots,n\}$.

Suppose $f_j(\prof{R}) = o$ for some agent~$j > k$. Let $R'_j \coloneqq R_k$ and consider the profile $\prof{R}' \coloneqq (R'_j, \prof{R}_{-j})$. By the induction hypothesis,
$$
f_\ell(\prof{R}') = \max_{R'_\ell}\left(O \setminus \bigcup^{\ell-1}_{i = 1}\{f_i(\prof{R}')\}\right)
$$
holds for $\ell = 1,\dots, k-1$. Since $R'_\ell = R_\ell$ for every agent $\ell \in \{1, \dots, k-1\}$, it follows that $f_\ell(\prof{R}) = f_\ell(\prof{R}')$ for every $\ell \in \{1, \dots, k-1\}$. Therefore
$$
A = O \setminus \bigcup^{k-1}_{i = 1}\{f_i(\prof{R}')\}, 
$$
and hence both $f_j(\prof{R}')$ and $f_k(\prof{R}')$ belong to $A$.

\prop{Strategy-proofness} implies that $f_j(\prof{R}') \mathrel{R'_j} f_j(\prof{R}) = o$. Therefore, since $o = \max_{R'_j}(A)$ and $f_j(\prof{R}') \in A$, we must have $f_j(\prof{R}') = o$.
As a result, $f_k(\prof{R}')\neq o$ and hence $f_j(\prof{R}') \mathrel{P'_k} f_k(\prof{R}') $.
However, since $R'_k = R_k=R'_j$, \prop{fair treatment of equals} and $k \mathrel{\pi} j$ gives $f_k(\prof{R}') \mathrel{R'_k} f_j(\prof{R}')$, a contradiction.
It follows that $f_k(\prof{R}) = o$, and hence $f_k(\prof{R}) = \max_{R_k}(A)$.

This completes the inductive step. Iterating from $k = 2$ through to $k = n$, we see that $f$ is the serial dictatorship associated with $\pi$.

\subsection{Proof of Proposition~\ref{prop:another_characterization}}\label{app:another_characterization_proof}

The serial dictatorship $f^\pi$ satisfies the four properties. For the converse, let $f$ be a mechanism satisfying them. By Fact~\ref{fact}, $f$ is monotonic. Fix $\prof{R} \in \mathcal{R}^N$ and let $\mu \coloneqq f^\pi(\prof{R})$. Choose $\succeq \in \mathcal{R}$ with $\succeq : \mu_{\pi_1}, \dots, \mu_{\pi_n},\dots$, and consider the identical-preference profile $\prof{R}(\succeq)$.

At $\prof{R}(\succeq)$, \prop{non-wastefulness} implies that the assigned objects are the $n$ objects that $\succeq$ ranks highest, namely $\mu_{\pi_1}, \dots, \mu_{\pi_n}$. \prop{Identical-preference fair treatment of equals} then implies that agent~$\pi_k$ receives the $k$th-ranked assigned object, so $f(\prof{R}(\succeq)) = \mu$.

We claim that $\prof{R}$ is a monotonic transformation of $\prof{R}(\succeq)$ at $\mu$, that is, $L(\mu_{\pi_k}, \succeq) \subseteq L(\mu_{\pi_k}, R_{\pi_k})$ for every $k \in \{1, \dots, n\}$. Fix $k \in \{1,\dots, n\}$. Observe that $L(\mu_{\pi_k}, \succeq) = O \setminus \{\mu_{\pi_1}, \dots, \mu_{\pi_{k-1}}\}$ is exactly the set of objects remaining when agent~$\pi_k$ is called at step~$k$ of the $\pi$-serial dictatorship algorithm at $\prof{R}$. Since $\mu_{\pi_k} = \max_{R_{\pi_k}}\big(O \setminus \{\mu_{\pi_1}, \dots, \mu_{\pi_{k-1}}\}\big)$, it follows that $L(\mu_{\pi_k}, \succeq) \subseteq L(\mu_{\pi_k}, R_{\pi_k})$.

Since $f(\prof{R}(\succeq)) = \mu$ and $\prof{R}$ is a monotonic transformation of $\prof{R}(\succeq)$ at $\mu$, \prop{monotonicity} implies $f(\prof{R}) = \mu = f^\pi(\prof{R})$. As $\prof{R}$ was arbitrary, $f = f^\pi$.\hfill \qed

\subsection{Proof of Proposition~\ref{prop:IPLB}}
Let $f^\pi$ be the serial dictatorship associated with priority $\pi$. Without loss of generality, let $\pi$ be the identity priority.

Fix $\prof{R} \in \mathcal{R}^N$ and $i \in N$. Let $R_i$ rank objects in the order $R_i : o_1, o_2, \dots, o_{|O|}$, and consider the identical-preference profile $\prof{R}(R_i)$.

By the definition of $f^\pi$, we have $f^\pi_i(\prof{R}(R_i)) = o_i$. Again, by the definition of $f^\pi$, we have 
$$
f^\pi_i(\prof{R}) = \max_{R_i}\left(O \setminus \bigcup^{i-1}_{j = 1}\{f^\pi_j(\prof{R})\}\right).
$$
Since $\bigcup^{i-1}_{j = 1}\{f^\pi_j(\prof{R})\}$ contains $i - 1$ objects, and the upper contour set $\{o \in O \mid o \mathrel{R_i} o_i\}$ contains $i$ objects, we must have $f^\pi_i(\prof{R}) \in \{o \in O \mid o \mathrel{R_i} o_i\}$. Thus $f^\pi_i(\prof{R}) \mathrel{R_i} f^\pi_i(\prof{R}(R_i))$. \hfill \qed

\subsection{Proofs of Theorem~\ref{thm1}, Proposition~\ref{prop:seq + IPLB},  and Theorem~\ref{thm:iplb}}

The following notation and terminology will be used throughout these proofs and, with the conventions stated in Appendix~\ref{appendix:proofs-extension}, in the proof of Theorem~\ref{thm:sequential-dictatorship-null}.

Recall that $\mathcal{S}$ is the set of all suballocations. For each integer $k$, let $\mathcal{S}_{k}$ be the set of suballocations containing exactly $k$ assignments, and $\mathcal{S}_{<k}$ be the set of suballocations containing fewer than $k$ assignments.
Formally, denote $$\mathcal{S}_{k} \coloneqq \{\varsigma \in \mathcal{S} \mid |\varsigma| = k\} \quad \text{and} \quad \mathcal{S}_{<k} \coloneqq \{\varsigma \in \mathcal{S} \mid |\varsigma| < k\}.$$ Note in particular that $\mathcal{S}_{<0} = \emptyset$, $\mathcal{S}_{<n} = \mathcal{S} \setminus \mathcal{M}$, and $\mathcal{S}_{<k+1} = \mathcal{S}_{<k} \cup \mathcal{S}_k$ for every integer $k$. 

\subsubsection{Partial dictator mappings}\label{sec: Partial dictator mappings}

A \defterm{partial dictator mapping} is a function $\kappa : \mathcal{D} \to N$, where $\mathcal{D} \subseteq \mathcal{S} \setminus \mathcal{M}$, such that
$$
\kappa(\varsigma) \in N \setminus N(\varsigma) \quad \text{for every }  \varsigma \in \mathcal{D}.
$$
Given a partial dictator mapping $\kappa : \mathcal{D} \to N$ and a preference profile $\prof{R} \in \mathcal{R}^N$, we define $\varsigma^\kappa_t(\prof{R})$ and $i^\kappa_t(\prof{R})$ recursively (as in Section~\ref{dictatorships}) for as long as $\kappa(\varsigma_{t-1}^\kappa(\prof{R}))$ is defined. Formally, set $\varsigma_0^\kappa(\prof{R}) \coloneqq \varnothing$ and, for each $t \in \{1,\dots, n\}$ with $\varsigma^\kappa_{t-1}(\prof{R}) \in \mathcal{D}$, define
$$i^\kappa_t(\prof{R}) \coloneqq \kappa\big(\varsigma^\kappa_{t-1}(\prof{R})\big)$$
and
$$\varsigma^{\kappa}_t(\prof{R}) \coloneqq \varsigma^\kappa_{t-1}(\prof{R}) \cup \left\{ \left(i^\kappa_t(\prof{R}), \max_{R_{i^\kappa_t(\prof{R})}}  A(\varsigma^\kappa_{t-1}(\prof{R}) )  \right) \right\}.$$
The terminology of Section~\ref{dictatorships} extends to partial dictator mappings: a suballocation $\varsigma \in \mathcal{S}_{k}$ is called \defterm{$\kappa$-reachable} if $\varsigma^\kappa_k(\prof{R})$ is defined and equals $\varsigma$ for some $\prof{R} \in \mathcal{R}^N$, in which case $\prof{R}$ \defterm{reaches} $\varsigma$ under $\kappa$.

We repeatedly use the following elementary observation: if a profile $\prof{R}$ reaches $\varsigma$ under $\kappa$ and $\prof{R}'$ is a profile with $\prof{R}'_{N(\varsigma)} = \prof{R}_{N(\varsigma)}$, then $\prof{R}'$ also reaches $\varsigma$ under $\kappa$.

In the proofs of Theorems~\ref{thm1},~\ref{thm:iplb}, and~\ref{thm:sequential-dictatorship-null}, we fix a mechanism $f$ satisfying the stated properties and recursively construct a dictator mapping $\kappa : \mathcal{S} \setminus \mathcal{M} \to N$ such that $f = f^\kappa$. The following three lemmas will prove useful. The first shows that there is a unique history associated with every reachable suballocation.

\begin{lemma}\label{suballocation determines sequence}
    Let $\kappa : \mathcal{S}_{<k} \to N$ be a partial dictator mapping, where $k \in \{0,1,\dots,n\}$. Every $\kappa$-reachable $\varsigma \in \mathcal{S}_k$ determines a unique history $$H^\kappa(\varsigma) = ((i_1, o_1), \dots, (i_k, o_k))$$
    of dictator--assignment pairs, which is empty if $k = 0$. That is, for all $\prof{R}, \prof{R}' \in \mathcal{R}^N$ with $\varsigma^\kappa_k(\prof{R}) = \varsigma^\kappa_k(\prof{R}') = \varsigma$, and for every $t \in \{1,\dots, k\}$,
$$
i^\kappa_t(\prof{R}) = i^\kappa_t(\prof{R}') = i_t \quad \text{and} \quad \varsigma^\kappa_t(\prof{R}) = \varsigma^\kappa_t(\prof{R}')  = \{(i_s, o_s) \mid 1 \leq s \leq t\}.
$$
\end{lemma}

\begin{proof}[Proof of Lemma~\ref{suballocation determines sequence}]
The statement is vacuously true if $k = 0$. Assume $k \geq 1$.

For $t = 1$, the first dictator is $i^\kappa_1(\prof{R}) = \kappa(\varnothing) = i^\kappa_1(\prof{R}')$. Since $\varsigma^\kappa_k(\prof{R}) = \varsigma^\kappa_k(\prof{R}')$, this dictator receives the same object in the suballocations $\varsigma^\kappa_1(\prof{R})$ and $\varsigma^\kappa_1(\prof{R}')$.

Now proceed inductively. Let $t \in \{2,\dots,k\}$ and suppose that the first $t-1$ dictator--assignment pairs coincide at $\prof{R}$ and $\prof{R}'$. Then $\varsigma^\kappa_{t-1}(\prof{R}) = \varsigma^\kappa_{t-1}(\prof{R}')$, so
$$i^\kappa_t(\prof{R}) = \kappa(\varsigma^\kappa_{t-1}(\prof{R})) = \kappa(\varsigma^\kappa_{t-1}(\prof{R}')) = i^\kappa_t(\prof{R}').$$
Again, since $\varsigma^\kappa_k(\prof{R}) = \varsigma^\kappa_k(\prof{R}')$, this $t$th dictator receives the same assignment in the suballocations $\varsigma^\kappa_t(\prof{R})$ and $\varsigma^\kappa_t(\prof{R}')$.
\end{proof}

The next lemma records the implication of \prop{non-bossiness} used in the proofs of Theorems~\ref{thm1} and~\ref{thm:sequential-dictatorship-null}: once a mechanism acts as a sequential dictatorship for the first~$k$ assignments, the reports that reach $\varsigma \in \mathcal{S}_k$ affect the outcome only through $\varsigma$.

\begin{lemma}\label{tail assignments depend only on tail agents}
    Let $f$ be a \prop{non-bossy} mechanism, let $k \in \{0, 1, \dots, n\}$, and let $\kappa : \mathcal{S}_{<k} \to N$ be a partial dictator mapping such that
    $$f(\prof{R}) \supseteq \varsigma^\kappa_k(\prof{R}) \quad \text{for every } \prof{R} \in \mathcal{R}^N.$$
    If $\prof{R}, \prof{R}' \in \mathcal{R}^N$ satisfy
    $$ \varsigma^\kappa_k(\prof{R}) = \varsigma^\kappa_k(\prof{R}') = \varsigma \quad \text{and} \quad \prof{R}_{N \setminus N(\varsigma)} = \prof{R}'_{N \setminus N(\varsigma)},$$
    then $f(\prof{R}) = f(\prof{R}')$.
\end{lemma}

\begin{proof}[Proof of Lemma~\ref{tail assignments depend only on tail agents}]
    If $k = 0$, then $\prof{R} = \prof{R}'$ and the conclusion is immediate. Let $k \geq 1$ and let $H^\kappa(\varsigma) = ((i_1, o_1), \dots, (i_k, o_k))$. For each $\ell \in \{0, 1, \dots, k\}$, define the profile $\prof{R}^\ell \coloneqq (\prof{R}'_{\{i_1, \dots, i_\ell\}}, \prof{R}_{N \setminus \{i_1, \dots, i_\ell\}})$, so that $\prof{R}^0 = \prof{R}$ and $\prof{R}^k = \prof{R}'$. Because $\prof{R}$ and $\prof{R}'$ both reach $\varsigma$ under $\kappa$, the reports $R_{i_t}$ and $R'_{i_t}$ both rank~$o_t$ first among $A(\varsigma^\kappa_{t-1}(\prof{R}))$, for every $t \in \{1,\dots, k\}$; hence every $\prof{R}^\ell$ reaches $\varsigma$ under $\kappa$. 
    
    For each $\ell \in \{1, \dots, k\}$, the profiles $\prof{R}^{\ell - 1}$ and $\prof{R}^\ell$ differ only in agent~$i_\ell$'s report, and the hypothesis on $\kappa$ gives $f_{i_\ell}(\prof{R}^{\ell-1}) = o_\ell = f_{i_\ell}(\prof{R}^\ell)$. \prop{Non-bossiness} therefore implies $f(\prof{R}^{\ell - 1}) = f(\prof{R}^\ell)$, and iterating over $\ell = 1, \dots, k$ yields $f(\prof{R}) = f(\prof{R}')$. 
\end{proof}

Given a partial dictator mapping that determines a mechanism's first $k$ assignments, the last lemma extends it by one step.

\begin{lemma}\label{recursive extension lemma}
Fix a mechanism $f$ and $k \in \{0, 1, \dots, n - 1\}$. Let $\kappa : \mathcal{S}_{<k} \to N$ be a partial dictator mapping such that
$$f(\prof{R}) \supseteq \varsigma^\kappa_k(\prof{R}) \quad \text{for every } \prof{R} \in \mathcal{R}^N.$$
Suppose that, for every $\kappa$-reachable suballocation $\varsigma \in \mathcal{S}_k$, there exists an agent $d(\varsigma) \in N \setminus N(\varsigma)$ such that
\begin{equation}\label{local dictator property}
f_{d(\varsigma)}(\prof{R}) = \max_{R_{d(\varsigma)}}A(\varsigma) \quad \text{for every } \prof{R} \in \mathcal{R}^N \text{ such that } \varsigma^\kappa_k(\prof{R}) = \varsigma.
\end{equation}
Then $\kappa$ extends to a partial dictator mapping $\overline{\kappa} : \mathcal{S}_{< k + 1} \to N$ such that
$$f(\prof{R}) \supseteq \varsigma^{\overline{\kappa}}_{k+1}(\prof{R}) \quad \text{for every } \prof{R} \in \mathcal{R}^N.$$
\end{lemma}
\begin{proof}[Proof of Lemma~\ref{recursive extension lemma}]
Define $\overline{\kappa} : \mathcal{S}_{<k+1} \to N$ such that, for every $\varsigma \in \mathcal{S}_{<k+1}$,
$$ \overline{\kappa}(\varsigma) = 
\begin{cases}
\kappa(\varsigma), & \text{if } \varsigma \in \mathcal{S}_{<k}\\
d(\varsigma), & \text{if } \varsigma \in \mathcal{S}_k \text{ and } \varsigma \text{ is } \kappa\text{-reachable}\\
\min \left(N \setminus N(\varsigma) \right), & \text{if } \varsigma \in \mathcal{S}_k \text{ and } \varsigma \text{ is not } \kappa\text{-reachable}. 
\end{cases}
$$
Because $k \leq n-1$, the set $N \setminus N(\varsigma)$ is nonempty for every $\varsigma \in \mathcal{S}_k$. Hence $\overline{\kappa}$ is a partial dictator mapping that extends $\kappa$.

Let $\prof{R} \in \mathcal{R}^N$ and denote $\varsigma \coloneqq \varsigma^{\overline{\kappa}}_{k}(\prof{R})$. Since $\overline{\kappa}$ extends $\kappa$, we have $\varsigma = \varsigma^\kappa_k(\prof{R})$. Since $\varsigma$ is $\kappa$-reachable, $\overline{\kappa}(\varsigma) = d(\varsigma)$ and (\ref{local dictator property}) gives $$f_{\overline{\kappa}(\varsigma)}(\prof{R}) = \max_{R_{\overline{\kappa}(\varsigma)}} A(\varsigma).$$
By hypothesis, $f(\prof{R}) \supseteq \varsigma^\kappa_k(\prof{R}) = \varsigma$, and consequently
$$ \varsigma^{\overline{\kappa}}_{k+1}(\prof{R}) = \varsigma \cup \left\{ \left( \overline{\kappa}(\varsigma), \max_{R_{\overline{\kappa}(\varsigma)}} A(\varsigma) \right) \right\} = \varsigma \cup \left\{ \left( \overline{\kappa}(\varsigma), f_{\overline{\kappa}(\varsigma)}(\prof{R}) \right) \right\} \subseteq f(\prof{R}).$$
\end{proof}

\subsubsection{Proof of Theorem~\ref{thm1}}

It is straightforward to show that every sequential dictatorship satisfies \prop{non-wastefulness}, \prop{strategy-proofness}, \prop{non-bossiness}, and \prop{pairwise tie-breaking}.

It remains to prove the only-if direction. Let $f$ be a mechanism satisfying \prop{non-wastefulness}, \prop{strategy-proofness}, \prop{non-bossiness}, and \prop{pairwise tie-breaking}. By Fact~\ref{fact} and Lemma~\ref{lemma:pe}, $f$ is \prop{monotonic} and \prop{pairwise efficient}.

We shall recursively construct a dictator mapping $\kappa : \mathcal{S} \setminus \mathcal{M} \to N$ such that $f = f^\kappa$. Specifically, we prove by induction that, for each $k \in \{0, 1, \dots, n\}$, there exists a partial dictator mapping $\kappa : \mathcal{S}_{<k} \to N$ satisfying
$$S(k): \qquad f(\prof{R}) \supseteq \varsigma^\kappa_k(\prof{R}) \quad \text{for every } \prof{R} \in \mathcal{R}^N.$$

\paragraph{Base case.} First note that $\mathcal{S}_{<0} = \emptyset$ and the empty function $\kappa : \mathcal{S}_{<0} \to N$ is vacuously a partial dictator mapping. Furthermore, $\kappa$ satisfies $S(0)$ because every allocation extends $\varnothing$.

\paragraph{Induction hypothesis.} Suppose that, for some $k \in \{0, 1, \dots, n-1\}$, we have constructed a partial dictator mapping $\kappa : \mathcal{S}_{<k} \to N$ satisfying $S(k)$.

\paragraph{Inductive step.} We will extend the partial dictator mapping $\kappa$ from $\mathcal{S}_{<k}$ to a partial dictator mapping on $\mathcal{S}_{<k+1}$ satisfying condition $S(k+1)$. By Lemma~\ref{recursive extension lemma}, it suffices to identify a next dictator after every $\kappa$-reachable suballocation in $\mathcal{S}_k$. 

To this end, let $\varsigma \in \mathcal{S}_k$ be $\kappa$-reachable. By Lemma~\ref{suballocation determines sequence}, $\varsigma$ determines a unique history $H^\kappa(\varsigma) = \left((i_1, o_1), \dots, (i_k, o_k) \right)$ of dictator--assignment pairs. Hence $$N(\varsigma) = \{i_1, \dots, i_k\}, \quad O(\varsigma) = \{o_1, \dots, o_k\},$$ and every profile reaching $\varsigma$ has this same history. 
(When $k = 0$, the history is empty and $N(\varsigma) = O(\varsigma) = \emptyset$.) Write
$$ \overline{N(\varsigma)} \coloneqq N \setminus N(\varsigma) \qquad \text{and} \qquad m \coloneqq |\overline{N(\varsigma)}| = n - k.$$

Fix a profile $\prof{R}^\varsigma \in \mathcal{R}^N$ that reaches $\varsigma$. For every $\succeq \in \mathcal{R}$, let $\prof{R}^\succeq \coloneqq (\prof{R}^\varsigma_{N(\varsigma)}, \prof{R}(\succeq)_{\overline{N(\varsigma)}})$ be the preference profile such that every agent $i \in N(\varsigma)$ reports $R^\varsigma_i$ and every agent $j \in \overline{N(\varsigma)}$ reports $\succeq$. Denote $\mu^\succeq \coloneqq f(\prof{R}^\succeq)$. For every $\succeq \in \mathcal{R}$, index the agents in $\overline{N(\varsigma)}$ so that $$\overline{N(\varsigma)} = \{j_1(\succeq), \dots, j_m(\succeq)\} \qquad \text{and} \qquad \mu^\succeq_{j_1(\succeq)} \mathrel{\succeq} \mu^\succeq_{j_2(\succeq)} \mathrel{\succeq} \cdots \mathrel{\succeq} \mu^\succeq_{j_m(\succeq)}.$$
Since $f(\prof{R}^\succeq)$ extends $\varsigma$ and all agents in $\overline{N(\varsigma)}$ report $\succeq$ at $\prof{R}^\succeq$, \prop{non-wastefulness} implies that $\{\mu^\succeq_{j} \mid j \in \overline{N(\varsigma)}\}$ equals the $m$ best objects in $A(\varsigma)$ according to $\succeq$. In particular,
\begin{equation}\label{There is a Q dictator}
\mu^\succeq_{j_1(\succeq)} = \max_\succeq A(\varsigma).
\end{equation}
For each $\ell \in \{2, \dots, m+1 \}$, define $$\mathcal{R}^\succeq_\ell \coloneqq \{\prof{R} \in \mathcal{R}^N \mid \varsigma^\kappa_k(\prof{R}) = \varsigma \text{ and } R_{j_t(\succeq)} = \succeq \text{ for every index } t < \ell \}.$$
Observe that $\mathcal{R}^\succeq_{m+1} \subseteq \mathcal{R}^\succeq_{m} \subseteq \cdots \subseteq \mathcal{R}^\succeq_{2}$.

\paragraph{Step 1: For every $\succeq \in \mathcal{R}$, there exists a ``$\succeq$-dictator'' after $\varsigma$.}

\begin{lemma}\label{thm2:persistence claim}
For every $\succeq \in \mathcal{R}$, every $\ell \in \{2, \dots, m+1 \}$, and every $\prof{R} \in \mathcal{R}^\succeq_\ell$,
$$f_{j_t(\succeq)}(\prof{R}) = \mu^\succeq_{j_t(\succeq)} \quad \text{for every index } t < \ell.$$
\end{lemma}

\begin{proof}[Proof of Lemma~\ref{thm2:persistence claim}]
Let $\succeq \in \mathcal{R}$. To simplify notation throughout this proof, write $j_t \coloneqq j_t(\succeq)$ for every $t \in \{1,\dots,m\}.$ We must show that, for every $\ell \in \{2, \dots, m+1 \}$ and every $\prof{R} \in \mathcal{R}^\succeq_\ell$,
\begin{equation}\label{thm2:persistence conclusion}
f_{j_t}(\prof{R}) = \mu^\succeq_{j_t} \quad \text{for every index } t < \ell.
\end{equation}
We proceed by descending induction on $\ell$ (starting from $\ell = m+1$ and decrementing to $\ell = 2$).

For the base case, let $\prof{R} \in \mathcal{R}^\succeq_{m+1}$. Then $\varsigma^\kappa_k(\prof{R}) = \varsigma^\kappa_k(\prof{R}^\succeq) = \varsigma$ and $\prof{R}_{\overline{N(\varsigma)}} = \prof{R}^\succeq_{\overline{N(\varsigma)}}$,
so Lemma~\ref{tail assignments depend only on tail agents} implies that $f(\prof{R}) = f(\prof{R}^\succeq) = \mu^\succeq$.

For the inductive step, fix $\ell \in \{2, \dots, m\}$ and suppose that the conclusion holds for $\ell + 1.$ Write 
$$ Y \coloneqq \{\mu^\succeq_{j_1}, \dots, \mu^\succeq_{j_{\ell - 1}}\}.$$
The preference $\succeq$ ranks every object in $Y$ above every object in $A(\varsigma) \setminus Y$.

\begin{claim}\label{thm2:Deviator cannot get an object in Y}
For every $\prof{R} \in \mathcal{R}^\succeq_\ell$, $f_{j_\ell}(\prof{R}) \notin Y$.
\end{claim}
\begin{subproof}[Proof of Claim~\ref{thm2:Deviator cannot get an object in Y}]
Let $\prof{R} \in \mathcal{R}^\succeq_\ell$ and set $\prof{R}^+ \coloneqq (\succeq, \prof{R}_{-j_\ell})$. Then $\prof{R}^+ \in \mathcal{R}^\succeq_{\ell + 1}$, so the descending induction hypothesis gives $f_{j_\ell}(\prof{R}^+) = \mu^\succeq_{j_\ell}$. Since $\prof{R} = (R_{j_\ell}, \prof{R}^+_{-j_\ell})$ differs from $\prof{R}^+$ only in the report of agent~$j_\ell$, \prop{strategy-proofness} implies $\mu^\succeq_{j_\ell} = f_{j_\ell}(\prof{R}^+) \mathrel{\succeq} f_{j_\ell}(\prof{R})$. Because every object in $Y$ is strictly preferred to $\mu^\succeq_{j_\ell}$ under $\succeq$, it follows that $f_{j_\ell}(\prof{R}) \notin Y$.
\end{subproof}

\begin{claim}\label{thm2:later agents outside Y}
For every $\prof{R} \in \mathcal{R}^\succeq_\ell$ and every index $ t > \ell$, $f_{j_t}(\prof{R}) \notin Y$.
\end{claim}

\begin{subproof}
Fix $\prof{R} \in \mathcal{R}^\succeq_\ell$ and $t > \ell$. Suppose toward contradiction that $f_{j_t}(\prof{R}) \in Y$. Then
$f_{j_t}(\prof{R}) = \mu^\succeq_{j_s}$ for some index $s < \ell$. Denote $ z \coloneqq f_{j_\ell}(\prof{R})$. By Claim~\ref{thm2:Deviator cannot get an object in Y}, $z \notin Y$.

Choose a preference $\succeq' \in \mathcal{R}$ such that
\begin{enumerate}
    \item $\mu^\succeq_{j_s}$ is the most-preferred object;
    \item every object in $Y$ is ranked above $z$; and
    \item $z$ is ranked above every object in $O \setminus (Y \cup \{z\})$.
\end{enumerate}
Let $\prof{R}' \coloneqq (\succeq', \prof{R}_{-j_\ell})$. Since $\prof{R}' \in \mathcal{R}^\succeq_\ell$, Claim~\ref{thm2:Deviator cannot get an object in Y} implies $f_{j_\ell}(\prof{R}') \notin Y$. At the profile $\prof{R}'$, agent~$j_\ell$ could obtain $z$ by reporting $R_{j_\ell}$. Hence \prop{strategy-proofness} and the definition of $R'_{j_\ell} = \succeq'$ imply that
$f_{j_\ell}(\prof{R}') = z = f_{j_\ell}(\prof{R})$. By \prop{non-bossiness}, $f(\prof{R}') = f(\prof{R})$; in particular,
$$f_{j_t}(\prof{R}') = \mu^\succeq_{j_s}.$$

Construct profile $\prof{R}'' \coloneqq (\succeq', \succeq', \prof{R}_{N \setminus \{j_\ell, j_t\}})$ from $\prof{R}'$ by replacing agent~${j_t}$'s preference with $R'_{j_\ell} = \succeq'$. Because  $R''_{j_t} = \succeq'$ ranks $\mu^\succeq_{j_s} = f_{j_t}(\prof{R}')$ first, $\prof{R}''$ is a monotonic transformation of $\prof{R}'$ at $f(\prof{R}')$. \prop{Monotonicity} therefore gives $f(\prof{R}'') = f(\prof{R}')$. Hence
\begin{equation}\label{thm:PTB contradiction 1}
R''_{j_\ell} = R''_{j_t} = \succeq' \qquad \text{and} \qquad f_{j_t}(\prof{R}'') = \mu^\succeq_{j_s} \mathrel{\succ'} f_{j_\ell}(\prof{R}'').
\end{equation}

Now construct profile $\hat{\prof{R}} \coloneqq (\succeq, \succeq, \prof{R}_{N \setminus \{j_\ell, j_t\}})$ from $\prof{R}$ by replacing the preferences of agents $j_\ell$ and $j_t$ with $\succeq$. Then $\prof{R} \in \mathcal{R}^\succeq_\ell$ implies that $\hat{\prof{R}} \in \mathcal{R}^\succeq_{\ell+1}$. Consequently, the descending induction hypothesis implies
$$ f_{j_r}(\hat{\prof{R}}) = \mu^\succeq_{j_r} \quad \text{for every index } r < \ell + 1.$$
The allocation $f(\hat{\prof{R}})$ is injective, and it extends $\varsigma$ because $S(k)$ holds and $\varsigma^\kappa_k(\hat{\prof{R}}) = \varsigma$. Consequently, agent~$j_t$ receives an assignment $f_{j_t}(\hat{\prof{R}}) \in A(\varsigma) \setminus \{\mu^\succeq_{j_1}, \dots, \mu^\succeq_{j_\ell}\}$. Hence
\begin{equation}\label{thm:PTB contradiction 2}
\hat{R}_{j_\ell} = \hat{R}_{j_t} = \succeq \qquad \text{and} \qquad f_{j_\ell}(\hat{\prof{R}}) = \mu^\succeq_{j_\ell} \mathrel{\succ} f_{j_t}(\hat{\prof{R}}).
\end{equation}
Together, (\ref{thm:PTB contradiction 1}) and (\ref{thm:PTB contradiction 2}) contradict \prop{pairwise tie-breaking}.
\end{subproof}

\begin{claim}\label{thm2:Y assigned to early agents}
For every $\prof{R} \in \mathcal{R}^\succeq_\ell$, $\{f_{j_1}(\prof{R}), \dots, f_{j_{\ell-1}}(\prof{R})\} = Y$.
\end{claim}

\begin{subproof}[Proof of Claim~\ref{thm2:Y assigned to early agents}]
Let $\prof{R} \in \mathcal{R}^\succeq_\ell$ and consider the allocation $f(\prof{R})$. By Claims~\ref{thm2:Deviator cannot get an object in Y} and~\ref{thm2:later agents outside Y}, no agent in $\{j_\ell, \dots, j_m\}$ is assigned an object in $Y$. Furthermore, the agents in $N(\varsigma)$ receive the objects in $O(\varsigma)$ because $S(k)$ holds and $\varsigma^\kappa_k(\prof{R}) = \varsigma$.

If some object $y \in Y$ were unassigned at $f(\prof{R})$, then some agent $j \in \{j_1, \dots, j_{\ell-1}\}$ would receive an assignment $f_j(\prof{R}) \in A(\varsigma) \setminus Y$. Since $R_j = \succeq$ ranks $y$ above every object in $A(\varsigma) \setminus Y$, this would contradict \prop{non-wastefulness}. Hence every object in $Y$ is assigned to an agent in $\{j_1, \dots, j_{\ell-1}\}$. Since the two sets have the same cardinality, we have $\{f_{j_1}(\prof{R}), \dots, f_{j_{\ell-1}}(\prof{R})\}~=~Y$.
\end{subproof}

\begin{claim}\label{thm2:early agents fixed}
For every $\prof{R} \in \mathcal{R}^\succeq_\ell$ and every index $t < \ell$, $f_{j_t}(\prof{R}) = \mu^\succeq_{j_t}$.
\end{claim}

\begin{subproof}[Proof of Claim~\ref{thm2:early agents fixed}]
Fix $\prof{R} \in \mathcal{R}^\succeq_\ell$ and consider the allocation $f(\prof{R})$. By Claim~\ref{thm2:Y assigned to early agents}, the agents in $\{j_1, \dots, j_{\ell-1}\}$ receive exactly the objects in $Y$. Suppose toward contradiction that one of these agents does not receive her assignment at $\mu^\succeq$. Let $r < \ell$ be the smallest index such that $f_{j_r}(\prof{R}) \neq \mu^\succeq_{j_r}$. Then
$ f_{j_r}(\prof{R}) = \mu^\succeq_{j_s}$ for some index $s \in \{r+1, \dots, \ell -1\}.$ Denote $z \coloneqq f_{j_\ell}(\prof{R})$. By Claim~\ref{thm2:Deviator cannot get an object in Y}, $z \notin Y$.

Let $\prof{R}^+ \coloneqq (\succeq, \prof{R}_{-j_\ell})$. Since $\prof{R}^+ \in \mathcal{R}^\succeq_{\ell + 1}$, the descending induction hypothesis gives
\begin{equation}\label{thm2:R plus assignments}
f_{j_t}(\prof{R}^+) = \mu^\succeq_{j_t} \quad \text{for every index } t < \ell +1.
\end{equation}

Let $R'_{j_\ell} \in \mathcal{R}$ be obtained from $\succeq$ by moving the single object~$z$ upward until it lies immediately below $\mu^\succeq_{j_{\ell-1}}$, while leaving the relative order of all other objects unchanged. Thus $$R'_{j_\ell}|_{A(\varsigma)}: \mu_{j_1}^\succeq, \dots, \mu^\succeq_{j_{\ell-1}}, z, \dots.$$
Since $(R'_{j_\ell}, \prof{R}_{-j_\ell}) \in \mathcal{R}^\succeq_{\ell}$,  Claim~\ref{thm2:Deviator cannot get an object in Y} implies that $f_{j_\ell}(R'_{j_\ell}, \prof{R}_{-j_\ell}) \notin Y$. At the profile $(R'_{j_\ell}, \prof{R}_{-j_\ell})$, agent~$j_\ell$ could obtain $z$ by reporting $R_{j_\ell}$. Hence \prop{strategy-proofness} and the definition of $R'_{j_\ell}$ imply that $f_{j_\ell}(R'_{j_\ell}, \prof{R}_{-j_\ell}) = z = f_{j_\ell}(\prof{R})$. By \prop{non-bossiness}, $f(R'_{j_\ell}, \prof{R}_{-j_\ell}) = f(\prof{R})$; in particular,
\begin{equation}\label{thm2:jr at dagger}
f_{j_r}(R'_{j_\ell}, \prof{R}_{-j_\ell}) = \mu^\succeq_{j_s}.
\end{equation}

Now choose a preference $R'_{j_r} \in \mathcal{R}$ for agent~$j_r$ such that 
$$R'_{j_r}:\mu^\succeq_{j_r}, z, \mu^\succeq_{j_s}, \dots,$$
and consider the profile $(R'_{j_r}, \prof{R}^+_{-j_r})$. Because $R'_{j_r}$ is a monotonic transformation of $R^+_{j_r}$ at $f_{j_r}(\prof{R}^+) = \mu^\succeq_{j_r}$, \prop{monotonicity} implies $f(R'_{j_r}, \prof{R}^+_{-j_r}) = f(\prof{R}^+)$; in particular,
\begin{equation}\label{thm2:ddagger assignments}
f_{j_r}(R'_{j_r}, \prof{R}^+_{-j_r}) = \mu^\succeq_{j_r} \qquad \text{and} \qquad f_{j_\ell}(R'_{j_r}, \prof{R}^+_{-j_r}) = \mu^\succeq_{j_\ell}.
\end{equation}

Finally, let $\prof{R}' \coloneqq (R'_{j_r}, R'_{j_\ell}, \prof{R}_{N \setminus \{j_r, j_\ell \}})$. At profile $\prof{R}'$, agent~$j_r$ could report $R_{j_r} = \succeq$ and obtain $\mu^\succeq_{j_s}$ by (\ref{thm2:jr at dagger}). Hence \prop{strategy-proofness} implies $f_{j_r}(\prof{R}') \mathrel{R'_{j_r}} \mu^\succeq_{j_s}$, and the definition of $R'_{j_r}$ implies $f_{j_r}(\prof{R}') \in \{\mu^\succeq_{j_r}, z, \mu^\succeq_{j_s}\}$.
Conversely, at profile $(R'_{j_\ell}, \prof{R}_{-j_\ell}) \in \mathcal{R}^\succeq_\ell$, where agent~$j_r$ reports $R_{j_r} = \succeq$, agent~$j_r$ could instead report $R'_{j_r}$ and obtain $f_{j_r}(\prof{R}')$. Hence \prop{strategy-proofness} implies $\mu^\succeq_{j_s} \mathrel{R_{j_r}} f_{j_r}(\prof{R}')$, and since $\mu^\succeq_{j_r} \mathrel{P_{j_r}} \mu^\succeq_{j_s}$ we must have $f_{j_r}(\prof{R}') \neq \mu^\succeq_{j_r}$. It follows that
$$f_{j_r}(\prof{R}') \in \{z, \mu^\succeq_{j_s} \}.$$
Consider two cases.

\paragraph{Case 1:} Suppose first that $f_{j_r}(\prof{R}') = z$. Since $f(\prof{R}')$ is injective, $f_{j_\ell}(\prof{R}') \neq z$. The preference $\succeq$ is obtained from $R'_{j_\ell}$ by moving the single object $z$ downward and leaving the relative order of all other objects unchanged. Hence $\succeq$ is a monotonic transformation of $R'_{j_\ell}$ at $f_{j_\ell}(\prof{R}')$, so \prop{monotonicity} implies that $f(\succeq, \prof{R}'_{-j_\ell}) = f(\prof{R}')$. But $(\succeq, \prof{R}'_{-j_\ell}) = (R'_{j_r}, \prof{R}^+_{-j_r})$, so (\ref{thm2:ddagger assignments}) yields
$$f_{j_r}(\prof{R}') = f_{j_r}(R'_{j_r}, \prof{R}^+_{-j_r}) = \mu^\succeq_{j_r},$$
a contradiction. 

\paragraph{Case 2:} Suppose instead that $f_{j_r}(\prof{R}') = \mu^\succeq_{j_s}$. The profiles $\prof{R}'$ and $(R'_{j_\ell}, \prof{R}_{-j_\ell})$ differ only in the report of agent~$j_r$, and (\ref{thm2:jr at dagger}) shows that $j_r$ receives the same object at both profiles. Hence \prop{non-bossiness} implies $f(\prof{R}') = f(R'_{j_\ell}, \prof{R}_{-j_\ell})$; in particular 
$$f_{j_\ell}(\prof{R}') = f_{j_\ell}(R'_{j_\ell}, \prof{R}_{-j_\ell}) = z.$$
By the definition of $\prof{R}' = (R'_{j_r}, R'_{j_\ell}, \prof{R}_{N \setminus \{j_r, j_\ell \}})$ and the fact that $s \leq \ell - 1$, we obtain
$$ f_{j_\ell}(\prof{R}') = z \mathrel{P'_{j_r}} \mu^\succeq_{j_s} = f_{j_r}(\prof{R}') \qquad \text{and} \qquad f_{j_r}(\prof{R}') = \mu^\succeq_{j_s} \mathrel{P'_{j_\ell}} z = f_{j_\ell}(\prof{R}').$$ Thus, at profile $\prof{R}'$, the two agents $j_r$ and $j_\ell$ can both benefit by swapping their assignments, contradicting \prop{pairwise efficiency}.

Both cases yield contradictions. Therefore $f_{j_t}(\prof{R}) = \mu^\succeq_{j_t}$ for every index $t < \ell$.
\end{subproof}

Claim~\ref{thm2:early agents fixed} establishes the descending induction step. By induction, for every $\ell \in \{2, \dots, m+1 \}$ and every $\prof{R} \in \mathcal{R}^\succeq_\ell$, $$f_{j_t}(\prof{R}) = \mu^\succeq_{j_t} \quad \text{for every index } t < \ell.$$
This completes the proof of Lemma~\ref{thm2:persistence claim}.
\end{proof}

For each $\succeq \in \mathcal{R}$, we call $j_1(\succeq)$ the \defterm{$\succeq$-dictator after $\varsigma$}. Taking $\ell = 2$ in Lemma~\ref{thm2:persistence claim} and using (\ref{There is a Q dictator}), we obtain
\begin{equation}\label{thm2:Q dictator property}
f_{j_1(\succeq)}(\prof{R}) = \max_{\succeq} A(\varsigma) \quad \text{for every } \prof{R} \in \mathcal{R}^\succeq_2.
\end{equation}

\paragraph{Step 2: The $\succeq$-dictator after $\varsigma$ is independent of $\succeq$.}

\begin{claim}\label{thm2:Q dictator constant}
For all $\succeq, \succeq' \in \mathcal{R}$, $j_1(\succeq) = j_1(\succeq')$.
\end{claim}

\begin{subproof}[Proof of Claim~\ref{thm2:Q dictator constant}]
Suppose toward contradiction that there exist $\succeq, \succeq' \in \mathcal{R}$ such that $i \coloneqq j_1(\succeq) \neq j_1(\succeq') \eqqcolon j$. Recall the profile $\prof{R}^\succeq = (\prof{R}^\varsigma_{N(\varsigma)}, \prof{R}(\succeq)_{\overline{N(\varsigma)}})$ and construct $\prof{R}' \coloneqq (\succeq', \succeq', \prof{R}^\succeq_{N \setminus \{i, j\}})$ from $\prof{R}^\succeq$ by replacing the preferences of agents~$i$ and~$j$ with $R'_i = R'_j = \succeq'$. Both $\prof{R}^\succeq$ and $\prof{R}'$ reach $\varsigma$. Moreover, $\prof{R}^\succeq \in \mathcal{R}^\succeq_2$ and $\prof{R}' \in \mathcal{R}^{\succeq'}_2$.
By the definition of $i$ and $j$, we have 
$$f_{i}(\prof{R}^\succeq) = \max_{R^\succeq_{i}} A(\varsigma) \quad \text{and} \quad f_{j}(\prof{R}') = \max_{R'_{j}} A(\varsigma).$$
Because $i \neq j$ and both $f_{j}(\prof{R}^\succeq)$ and $f_{i}(\prof{R}')$ belong to $A(\varsigma)$, we have
$$f_{i}(\prof{R}^\succeq) \mathrel{P^\succeq_{i}} f_{j}(\prof{R}^\succeq) \quad \text{and} \quad f_{j}(\prof{R}') \mathrel{P'_{i}} f_{i}(\prof{R}').$$
This contradicts \prop{pairwise tie-breaking}.
\end{subproof}

\paragraph{Step 3: The constant $\succeq$-dictator is the next dictator after $\varsigma$.}

By Claim~\ref{thm2:Q dictator constant}, there exists an agent $d(\varsigma) \in N \setminus N(\varsigma)$ such that
$$ d(\varsigma) = j_1(\succeq) \quad \text{for every } \succeq \in \mathcal{R}.$$
Let $\prof{R} \in \mathcal{R}^N$ be such that $\varsigma^\kappa_k(\prof{R}) = \varsigma$, and let $\succeq \coloneqq R_{d(\varsigma)}$. Since $d(\varsigma) = j_1(\succeq)$, we have $\prof{R} \in \mathcal{R}^{\succeq}_2$. Hence (\ref{thm2:Q dictator property}) implies
$$f_{d(\varsigma)}(\prof{R}) = \max_{R_{d(\varsigma)}} A(\varsigma).$$
Since $\prof{R}$ was an arbitrary profile reaching $\varsigma$, it follows that
$$f_{d(\varsigma)}(\prof{R}) = \max_{R_{d(\varsigma)}} A(\varsigma) \quad \text{for every } \prof{R} \in \mathcal{R}^N \text{ such that } \varsigma^\kappa_k(\prof{R}) = \varsigma.$$

\paragraph{Conclusion.}
We have established that (\ref{local dictator property}) holds for agent $d(\varsigma) \in N \setminus N(\varsigma)$. By Lemma~\ref{recursive extension lemma}, $\kappa$ extends to a partial dictator mapping $\overline{\kappa} : \mathcal{S}_{<k+1} \to N$ satisfying $S(k+1)$. Replacing $\kappa$ by $\overline{\kappa}$ completes the inductive step.

Iterating from $k= 0$ through to $k = n-1$, we obtain a dictator mapping $\kappa : \mathcal{S}_{<n} \to N$ such that $$f(\prof{R}) \supseteq \varsigma^\kappa_n(\prof{R}) \quad \text{for every } \prof{R} \in \mathcal{R}^N.$$
Since both $f(\prof{R})$ and $\varsigma^\kappa_n(\prof{R})$ assign one object to every agent, the inclusion is an equality. Therefore 
$$f(\prof{R}) = \varsigma^\kappa_n(\prof{R}) = f^\kappa(\prof{R}) \quad \text{for every } \prof{R} \in \mathcal{R}^N.$$
Hence $f = f^\kappa$ is a sequential dictatorship. \hfill \qed

\subsubsection{Proof of Proposition~\ref{prop:seq + IPLB}}
Let $f = f^\kappa$ be the sequential dictatorship associated with the dictator mapping
$\kappa$, and suppose that $f$ satisfies the \prop{identical-preference lower bound}.
We show that $f$ is a serial dictatorship.

Specifically, we prove by induction that, for each $k \in \{1,\dots, n\}$,
there is a sequence of distinct agents $(\pi_1, \dots, \pi_k)$ such that, for every $t \in \{1,\dots,k\}$ and every $\prof{R} \in \mathcal{R}^N$, $i_t^\kappa(\prof{R}) = \pi_t$.

\paragraph{Base case.}
For the base case $k =1$, observe that agent~$\pi_1 \coloneqq \kappa(\varnothing)$ satisfies $i_1^\kappa(\prof{R}) = \pi_1$ for every $\prof{R} \in \mathcal{R}^N$.

\paragraph{Induction hypothesis.} Suppose that, for some $k \in \{1, \dots, n-1\}$, there is a sequence of distinct agents $(\pi_1, \dots, \pi_k)$ such that, for every $t \in \{1,\dots,k\}$ and every $\prof{R} \in \mathcal{R}^N$, $i_t^\kappa(\prof{R}) = \pi_t$.

\paragraph{Inductive step.}

The induction hypothesis implies that every $\kappa$-reachable suballocation in $\mathcal{S}_k$ has domain $I \coloneqq \{\pi_1,\dots, \pi_k\}$. Moreover, every suballocation with domain $I$ is $\kappa$-reachable.

\begin{claim}\label{Adjacent suballocations have the same dictator}
Let $\varsigma^*,\varsigma' \in \mathcal{S}_k$ be $\kappa$-reachable. If $\varsigma^*$ and $\varsigma'$ differ for exactly one agent in $I$ (i.e., there is a unique $\pi_r \in I$ with $\varsigma^*(\pi_r) \neq \varsigma'(\pi_r)$), then $\kappa(\varsigma^*) = \kappa(\varsigma')$.
\end{claim}

\begin{subproof}[Proof of Claim~\ref{Adjacent suballocations have the same dictator}]
Let $\varsigma^*,\varsigma' \in \mathcal{S}_k$ be $\kappa$-reachable, and suppose $\pi_r \in I$ is the unique agent with $\varsigma^*(\pi_r) \neq \varsigma'(\pi_r)$. Thus $\varsigma^*(\pi_t) = \varsigma'(\pi_t)$ for every $\pi_t \in I \setminus \{\pi_r\}$.

Suppose toward contradiction that $\kappa(\varsigma^*) \neq \kappa(\varsigma')$. Let $j^* \coloneqq \kappa(\varsigma^*)$ and $j' \coloneqq \kappa(\varsigma')$. Then $j^*$ and $j'$ are distinct agents in $N \setminus \{\pi_1, \dots, \pi_k\}$. Let $x^* \coloneqq \varsigma^*(\pi_r)$ and $x' \coloneqq \varsigma'(\pi_r)$.

Let $\prof{R}' \in \mathcal{R}^N$ be an identical-preference profile at which every agent reports
$$
R'_1 : \varsigma'(\pi_1), \dots, \varsigma'(\pi_k),x^*,\dots.
$$

Similarly, let $\prof{R}^* \in \mathcal{R}^N$ be an identical-preference profile at which every agent reports
$$
R^*_1 : \varsigma^*(\pi_1), \dots, \varsigma^*(\pi_k),x',\dots.
$$

Since the first~$k$ dictators are $\pi_1,\dots,\pi_k$, we have $\varsigma^\kappa_k(\prof{R}') = \varsigma'$ and $\varsigma^\kappa_k(\prof{R}^*) = \varsigma^*$. Furthermore, $j' = \kappa(\varsigma')$ implies 
\begin{equation}\label{IPLB contradiction}
    f_{j'}(\prof{R}') = \max_{R'_{j'}}A(\varsigma') =  x^*.
\end{equation}
Now consider the profile $\hat{\prof{R}} \coloneqq (\prof{R}^*_I, \prof{R}'_{N \setminus I})$. Since the first $k$ dictators $\pi_1,\dots,\pi_k$ report the same preferences under $\hat{\prof{R}}$ and $\prof{R}^*$, we have $\varsigma^\kappa_k(\hat{\prof{R}}) = \varsigma^*$ as well. Furthermore, $j^* = \kappa(\varsigma^*)$ is the next dictator after $\varsigma^*$. Since $j^*$ reports $R'_{j^*}$ under $\hat{\prof{R}}$ and $O(\varsigma') \setminus O(\varsigma^*) = \{x'\}$, we have 
$$
f_{j^*}(\hat{\prof{R}}) = \max_{R'_{j^*}} A(\varsigma^*) =  x'.
$$
Hence the objects weakly preferred to $x^*$ under $R'_{j'}$, namely those in $$\{o \in O \mid o \mathrel{R'_{j'}} x^*\} = O(\varsigma') \cup \{x^*\} = O(\varsigma^*) \cup \{x'\},$$
are all assigned at $f(\hat{\prof{R}})$ to the agents in $I \cup \{j^*\}$. Since $j' \notin I \cup \{j^*\}$ and $j'$ reports $R'_{j'}$ under $\hat{\prof{R}}$, it follows that $x^* \mathrel{P'_{j'}} f_{j'}(\hat{\prof{R}})$. Hence, by the \prop{identical-preference lower bound}, we obtain
\begin{equation}
    x^* \mathrel{P'_{j'}} f_{j'}(\hat{\prof{R}}) \mathrel{R'_{j'}} f_{j'}(\prof{R}'),
\end{equation}
which contradicts (\ref{IPLB contradiction}).
\end{subproof}

We now extend Claim~\ref{Adjacent suballocations have the same dictator} from ``adjacent'' suballocations to all $\kappa$-reachable suballocations.

\begin{claim}\label{Same dictator globally}
    Let $\varsigma^*,\varsigma' \in \mathcal{S}_k$ be $\kappa$-reachable. Then $\kappa(\varsigma^*) = \kappa(\varsigma')$.
\end{claim}

\begin{subproof}[Proof of Claim~\ref{Same dictator globally}]
    Let $\varsigma^*, \varsigma' \in \mathcal{S}_k$ be distinct $\kappa$-reachable suballocations. Then there exists a finite path $(\varsigma^0, \varsigma^1,\dots, \varsigma^T)$ of $\kappa$-reachable suballocations in $\mathcal{S}_k$ such that $\varsigma^0 = \varsigma^*$, $\varsigma^T = \varsigma'$, and, for each $t = 1, \dots, T$, $\varsigma^t$ and $\varsigma^{t-1}$ differ for exactly one agent. To see this, construct $\varsigma^t$ from $\varsigma^{t-1}$ as follows. Let $\ell \in \{1,\dots,k\}$ be the smallest index such that $\varsigma^{t-1}(\pi_\ell) \neq \varsigma'(\pi_\ell)$. If $\varsigma'(\pi_\ell) \in O(\varsigma^{t-1})$, then necessarily $\varsigma'(\pi_\ell) = \varsigma^{t-1}(i)$ for some $i \in \{\pi_{\ell+1}, \dots, \pi_{k}\}$; assign agent~$i$ to an arbitrary object in $A(\varsigma^{t-1})$ (such an object exists because $|\varsigma^{t-1}|=k \leq n-1<|O|$) and leave all other agents' assignments unchanged. If $\varsigma'(\pi_\ell) \notin O(\varsigma^{t-1})$, then assign agent~$\pi_\ell$ to object $\varsigma'(\pi_\ell)$ and leave all other agents' assignments unchanged.
    
    An iterative application of Claim~\ref{Adjacent suballocations have the same dictator} implies that $\kappa(\varsigma^*) = \kappa(\varsigma')$.
\end{subproof}

By Claim~\ref{Same dictator globally}, any two $\kappa$-reachable suballocations in $\mathcal{S}_k$ are assigned the same dictator. Hence there is an agent~$\pi_{k+1} \in N \setminus \{\pi_1,\dots,\pi_k\}$ such that, for every $\prof{R} \in \mathcal{R}^N$, $i^\kappa_{k+1}(\prof{R}) = \pi_{k+1}$. This completes the inductive step.

Finally, let $\pi$ be the priority that orders agents according to the sequence $(\pi_1, \dots, \pi_n)$. Since 
$$ i^\kappa_t(\prof{R}) = \pi_t \quad \text{for every } t \in \{1,\dots, n\} \text{ and every } \prof{R} \in \mathcal{R}^N,$$
the sequential dictatorship $f^\kappa$ coincides with the serial dictatorship $f^\pi$. Hence, $f = f^\kappa = f^\pi$. \hfill \qed

\subsubsection{Proof of Theorem~\ref{thm:iplb}}

The if direction follows from Theorem~\ref{thm0} and Proposition~\ref{prop:IPLB}, since \prop{fair treatment of equals} implies \prop{pairwise tie-breaking}. It remains to prove the only-if direction.

Let $f$ be a mechanism satisfying \prop{non-wastefulness}, \prop{strategy-proofness}, \prop{identical-preference lower bound}, and \prop{pairwise tie-breaking}. We first show that $f$ is a sequential dictatorship. Proposition~\ref{prop:seq + IPLB} then implies that $f$ is a serial dictatorship.

\subsubsection*{The mechanism $f$ is a sequential dictatorship.}

We shall recursively construct a dictator mapping $\kappa : \mathcal{S} \setminus \mathcal{M} \to N$ such that $f = f^\kappa$. Specifically, we prove by induction that, for each $k \in \{0, 1, \dots, n\}$, there exists a partial dictator mapping $\kappa : \mathcal{S}_{<k} \to N$ satisfying
$$S(k): \qquad f(\prof{R}) \supseteq \varsigma^\kappa_k(\prof{R}) \quad \text{for every } \prof{R} \in \mathcal{R}^N.$$

\paragraph{Base case.} First note that $\mathcal{S}_{<0} = \emptyset$ and the empty function $\kappa : \mathcal{S}_{<0} \to N$ is vacuously a partial dictator mapping. Furthermore, $\kappa$ satisfies $S(0)$ because every allocation extends $\varnothing$.

\paragraph{Induction hypothesis.} Suppose that, for some $k \in \{0, 1, \dots, n-1\}$, we have already constructed a partial dictator mapping $\kappa : \mathcal{S}_{<k} \to N$ satisfying
$$S(k): \qquad f(\prof{R}) \supseteq \varsigma^\kappa_k(\prof{R}) \quad \text{for every } \prof{R} \in \mathcal{R}^N.$$

\paragraph{Inductive step.} We will extend the partial dictator mapping $\kappa$ from $\mathcal{S}_{<k}$ to a partial dictator mapping on $\mathcal{S}_{<k+1}$ satisfying condition $S(k+1)$. By Lemma~\ref{recursive extension lemma}, it suffices to identify a next dictator after every $\kappa$-reachable suballocation in $\mathcal{S}_k$. 

To this end, let $\varsigma \in \mathcal{S}_k$ be $\kappa$-reachable. By Lemma~\ref{suballocation determines sequence}, $\varsigma$ determines a unique history $H^\kappa(\varsigma) = \left((i_1, o_1), \dots, (i_k, o_k) \right)$ of dictator--assignment pairs. Hence $$N(\varsigma) = \{i_1, \dots, i_k\}, \quad O(\varsigma) = \{o_1, \dots, o_k\},$$ and every profile reaching $\varsigma$ has this same history. 
(When $k = 0$, the history is empty and $N(\varsigma) = O(\varsigma) = \emptyset$.)

For any linear order $L$ on $A(\varsigma)$, let $\succeq^L \in \mathcal{R}$ be the preference that lists objects $o_1, \dots, o_k$ first, in that order, and then orders the remaining objects according to $L$. Thus $\succeq^L$ coincides with $L$ on $A(\varsigma)$. Consider the associated identical-preference profile $\prof{R}(\succeq^L)$.

For every linear order $L$ on $A(\varsigma)$, we have $\varsigma^\kappa_k(\prof{R}(\succeq^L)) = \varsigma$: the recursion assigns $o_1$ to $i_1 = \kappa(\varnothing)$, then $o_2$ to $i_2 = \kappa(\{(i_1, o_1)\})$, and so on. By \prop{non-wastefulness}, $f(\prof{R}(\succeq^L))$ assigns object $\max_{\succeq^L}A(\varsigma)$ to some agent, and since $S(k)$ implies $f(\prof{R}(\succeq^L)) \supseteq \varsigma^\kappa_k(\prof{R}(\succeq^L)) = \varsigma$, this recipient belongs to $N \setminus N(\varsigma)$. Let $j(L) \in N \setminus N(\varsigma)$ be the agent with $f_{j(L)}(\prof{R}(\succeq^L)) = \max_{\succeq^L} A(\varsigma)$. We show that $j(L)$ does not depend on $L$.

\begin{claim}\label{claim:unique next dictator}
For all linear orders $L, L'$ on $A(\varsigma)$, $j(L) = j(L')$.  
\end{claim}

\begin{subproof}[Proof of Claim~\ref{claim:unique next dictator}]
Toward contradiction, let $L$ and $L'$ be linear orders on $A(\varsigma)$ with $j(L) \neq j(L')$. Write 
$$
\prof{R} \coloneqq \prof{R}(\succeq^L), \qquad \prof{R}' \coloneqq \prof{R}(\succeq^{L'}), \qquad j \coloneqq j(L), \qquad j' \coloneqq j(L').$$ 
Then $f_j(\prof{R}) \mathrel{P_j} f_{j'}(\prof{R})$ and $f_{j'}(\prof{R}') \mathrel{P'_j} f_j(\prof{R}')$.

Let $\hat{\prof{R}} \coloneqq (\succeq^L,\succeq^L, \prof{R}'_{N \setminus \{j,j'
\}})$ be the profile at which agents~$j$ and~$j'$ report $\succeq^L$ and every other agent reports $\succeq^{L'}$. Then agents~$j$ and~$j'$ report identically at $\hat{\prof{R}}$ and at $\prof{R}'$, while $\hat{\prof{R}}_{N \setminus \{j,j'
\}} = \prof{R}'_{N \setminus \{j,j'
\}}$. Since $f_{j'}(\prof{R}') \mathrel{P'_j} f_j(\prof{R}')$, \prop{pairwise tie-breaking} implies that 
\begin{equation}\label{PTB for j'}
f_{j'}(\hat{\prof{R}}) \mathrel{R_j} f_j(\hat{\prof{R}}).
\end{equation}
The \prop{identical-preference lower bound} for agent~$j$ gives
\begin{equation}\label{j gets the best at hat}
f_j(\hat{\prof{R}}) \mathrel{R_j} f_j(\prof{R}) = \max_{\succeq^L}A(\varsigma).
\end{equation}
Since $\hat{\prof{R}}_{N(\varsigma)} = \prof{R}'_{N(\varsigma)}$ and $\varsigma^\kappa_k(\prof{R}') = \varsigma$, we have $\varsigma^\kappa_k(\hat{\prof{R}}) = \varsigma$ as well. 

The hypothesis $S(k)$ implies $f(\hat{\prof{R}}) \supseteq \varsigma^\kappa_k(\hat{\prof{R}}) = \varsigma$. Hence the agents in $N \setminus N(\varsigma)$ are assigned objects in $A(\varsigma)$ under $f(\hat{\prof{R}})$. Consequently, $f_j(\hat{\prof{R}}) \in A(\varsigma)$ and (\ref{j gets the best at hat}) imply that $f_j(\hat{\prof{R}}) = \max_{\succeq^L} A(\varsigma)$. Furthermore, $f_{j'}(\hat{\prof{R}}) \in A(\varsigma)$ and (\ref{PTB for j'}) imply that $f_{j'}(\hat{\prof{R}}) = f_j(\hat{\prof{R}})$. Since $j \neq j'$, this contradicts the fact that the allocation $f(\hat{\prof{R}})$ is injective.
\end{subproof}

By Claim~\ref{claim:unique next dictator}, there exists an agent $d(\varsigma) \in N \setminus N(\varsigma)$ such that $d(\varsigma) = j(L)$ for every linear order $L$ on $A(\varsigma)$. We show that
$$f_{d(\varsigma)}(\prof{R}) = \max_{R_{d(\varsigma)}}A(\varsigma) \quad \text{for every } \prof{R} \in \mathcal{R}^N \text{ such that } \varsigma^\kappa_k(\prof{R}) = \varsigma.$$  

Fix such a profile $\prof{R} \in \mathcal{R}^N$ with $\varsigma^\kappa_k(\prof{R}) = \varsigma$. Let $L$ be the restriction of $R_{d(\varsigma)}$ to $A(\varsigma)$, and denote $x \coloneqq \max_{R_{d(\varsigma)}}A(\varsigma)$. Since $d(\varsigma) = j(L)$,
we have $$f_{d(\varsigma)}(\prof{R}(\succeq^L)) = x.$$

Now consider the profile $\prof{R}' \coloneqq (\succeq^L, \prof{R}_{-d(\varsigma)})$. Since $d(\varsigma) \in N \setminus N(\varsigma)$ and $\varsigma^\kappa_k(\prof{R}) = \varsigma$, we have $\varsigma^\kappa_k(\prof{R}') = \varsigma$ as well. The \prop{identical-preference lower bound} gives
$$f_{d(\varsigma)}(\prof{R}') \mathrel{\succeq^L} f_{d(\varsigma)}(\prof{R}(\succeq^L)) = x.$$
The hypothesis $S(k)$ implies $f(\prof{R}') \supseteq \varsigma^\kappa_k(\prof{R}') = \varsigma$, and hence $f_{d(\varsigma)}(\prof{R}') \in A(\varsigma)$. It follows that $f_{d(\varsigma)}(\prof{R}') = x$.

Finally, at the original profile $\prof{R}$, agent~$d(\varsigma)$ could obtain object~$x$ by reporting $\succeq^L$. Hence \prop{strategy-proofness} implies that
$$ f_{d(\varsigma)}(\prof{R}) \mathrel{R_{d(\varsigma)}} f_{d(\varsigma)}(\prof{R}') = x.$$
Again, the hypothesis $S(k)$ implies $f(\prof{R}) \supseteq \varsigma^\kappa_k(\prof{R}) = \varsigma$, and hence $f_{d(\varsigma)}(\prof{R}) \in A(\varsigma)$. By the definition of $x$, $f_{d(\varsigma)}(\prof{R}) = x$.

We have established that (\ref{local dictator property}) holds for agent $d(\varsigma) \in N \setminus N(\varsigma)$. By Lemma~\ref{recursive extension lemma}, $\kappa$ extends to a partial dictator mapping $\overline{\kappa} : \mathcal{S}_{<k+1} \to N$ satisfying $S(k+1)$. Replacing $\kappa$ by $\overline{\kappa}$ completes the inductive step.

Iterating from $k= 0$ through to $k = n-1$ yields a dictator mapping $\kappa : \mathcal{S}_{<n} \to N$ such that $$f(\prof{R}) \supseteq \varsigma^\kappa_n(\prof{R}) \quad \text{for every } \prof{R} \in \mathcal{R}^N.$$
Since both $f(\prof{R})$ and $\varsigma^\kappa_n(\prof{R})$ specify an assignment for every agent, the inclusion is an equality. Therefore 
$$f(\prof{R}) = \varsigma^\kappa_n(\prof{R}) = f^\kappa(\prof{R}) \quad \text{for every } \prof{R} \in \mathcal{R}^N.$$
Hence $f = f^\kappa$ is a sequential dictatorship, and Proposition~\ref{prop:seq + IPLB} implies that $f$ is a serial dictatorship. \hfill \qed

\section{Proofs for the outside-option model}\label{appendix:proofs-extension}
Throughout this appendix, the notation refers to the model of Section~\ref{sec:exten}: $\mathcal{R}$ is the set of linear orders on $O_0$, $\mathcal{M}$ is the extended allocation set, suballocations may assign $0$ to any number of agents, and $A(\varsigma)$ is given by (\ref{eq:remaining-set}). With $\mathcal{R}$, $\mathcal{M}$, $\mathcal{S}$, and $A(\varsigma)$ given these extended meanings, the definitions of Appendix~\ref{sec: Partial dictator mappings}, the elementary observation stated there, and Lemmas~\ref{suballocation determines sequence}--\ref{recursive extension lemma} carry over verbatim.

We first recover the priority formulation of \prop{fair treatment of equals} under the additional properties \prop{individual rationality} and \prop{non-wastefulness}.

\begin{lemma}\label{lemma:priority-exists}
    Let a mechanism $f$ satisfy \prop{individual rationality}, \prop{non-wastefulness}, and \prop{fair treatment of equals} (that is, condition (\ref{FTE tie-break formulation})). Then there is a priority $\pi \in \Pi$ satisfying (\ref{priority exists}).
\end{lemma}

\begin{proof}[Proof of Lemma~\ref{lemma:priority-exists}]
    Fix $x \in O$ and preference relations $\succeq^x, \succeq^0 \in \mathcal{R}$ with
    $$ \succeq^x : x, 0,\dots; \qquad \succeq^0: 0, \dots.$$
    We construct $\pi_1, \dots, \pi_n$ recursively. 
    
    Let $\prof{R}^1 \coloneqq \prof{R}(\succeq^x)$. \prop{Individual rationality} and \prop{non-wastefulness} imply that $f(\prof{R}^1)$ assigns $x$ to exactly one agent; let $\pi_1$ be the agent with $x = f_{\pi_1}(\prof{R}^1)$ and set $\prof{R}^2 \coloneqq (\succeq^0, \prof{R}^1_{-\pi_1})$. Inductively, suppose $k \in \{2, \dots, n\}$ and there are distinct agents $\pi_1, \dots, \pi_{k-1}$ and profiles $\prof{R}^1, \dots, \prof{R}^{k}$ such that, for each $\ell < k$,
    \begin{equation}\label{eq:constructed-priority}
    f_{\pi_\ell}(\prof{R}^\ell) = x \quad \text{and} \quad \prof{R}^{\ell + 1} \coloneqq (\prof{R}(\succeq^0)_{\{\pi_1, \dots, \pi_{\ell}\}}, \prof{R}(\succeq^x)_{N \setminus \{\pi_1, \dots, \pi_{\ell}\}}).
    \end{equation}
    \prop{Individual rationality} and \prop{non-wastefulness} imply that $f(\prof{R}^k)$ assigns $x$ to exactly one agent in $N \setminus \{\pi_1, \dots, \pi_{k-1}\}$ and assigns $0$ to every other agent; let $\pi_k$ be the agent with $x = f_{\pi_k}(\prof{R}^{k})$. Iterating from $k = 1$ through to $k = n$, we obtain a priority $\pi = (\pi_1, \dots, \pi_n)$ such that (\ref{eq:constructed-priority}) holds for each $\ell = 1 ,\dots, n$.

    We now show that $\pi$ satisfies (\ref{priority exists}). Fix $\prof{R} \in \mathcal{R}^N$ and distinct $i, j \in N$ with $i \mathrel{\pi} j$ and $R_i = R_j$. Then $i = \pi_s$ and $j = \pi_t$ for some $s < t$. By construction, agent $i$ receives $x$ at $f(\prof{R}^s)$, where $ \prof{R}^s = (\prof{R}(\succeq^0)_{\{\pi_1, \dots, \pi_{s-1}\}}, \prof{R}(\succeq^x)_{\{\pi_s, \dots, \pi_n \}})$. Since $f_i(\prof{R}^s) = x \succ^x 0 = f_j(\prof{R}^s)$ and agents~$i$ and $j$ both report $\succeq^x$ at $\prof{R}^s$, \prop{fair treatment of equals} implies that $f_i(\prof{R}) \mathrel{R_i} f_j(\prof{R})$. Thus $\pi$ satisfies (\ref{priority exists}).
\end{proof}

\subsection{Proof of Theorem~\ref{thm:serial-dictatorship-null}}
Every serial dictatorship satisfies the stated properties. For the only-if direction, let $f$ satisfy \prop{individual rationality}, \prop{non-wastefulness}, \prop{strategy-proofness}, and \prop{fair treatment of equals}. By Lemma~\ref{lemma:priority-exists}, there exists a priority $\pi$ satisfying (\ref{priority exists}). Without loss of generality, let $\pi$ be the identity priority.

As in the proof of Theorem~\ref{thm0}, we establish the recursion (\ref{eq:serial-recursion}). Specifically, we prove by induction on $k$ that
$$
\text{for every } \prof{R} \in \mathcal{R}^N, \quad f_k(\prof{R}) = \max_{R_k} \left( \left(O \setminus \bigcup_{i = 1}^{k-1}\{f_i(\prof{R})\} \right) \cup \{0\} \right)
$$
holds for $k = 1,\dots,n.$

\paragraph{Base case.} Let $\prof{R} \in \mathcal{R}^N$ and $o \coloneqq \max_{R_1}(O_0)$. If $o = 0$, then \prop{individual rationality} gives $f_1(\prof{R}) = 0$. Let $o \in O$ and suppose toward contradiction that $f_1(\prof{R}) \neq o$. By \prop{non-wastefulness}, $f_j(\prof{R}) = o$ for some agent $j > 1$. Let $R'_j \coloneqq R_1$ and consider $\prof{R}' \coloneqq (R'_j, \prof{R}_{-j})$. \prop{Strategy-proofness} at $\prof{R}'$ implies $f_j(\prof{R}') \mathrel{R'_j} f_j(\prof{R}) = o$, and since $o = \max_{R'_j}(O_0)$ we must have $f_j(\prof{R}') = o$. As a result, $f_1(\prof{R}') \neq o$ and hence $f_j(\prof{R}') \mathrel{P'_1} f_1(\prof{R}')$. Since $1 \mathrel{\pi} j$ and $R'_1 = R_1 = R'_j$, this contradicts (\ref{priority exists}).

\paragraph{Induction hypothesis.}
    Let $k \geq 2$, and assume that
    $$
    \text{for every } \prof{R} \in \mathcal{R}^N, \quad f_\ell(\prof{R}) = \max_{R_\ell} \left( \left(O \setminus \bigcup_{i = 1}^{\ell-1}\{f_i(\prof{R})\} \right) \cup \{0\} \right)
    $$
    holds for $\ell = 1, \dots, k-1$.

\paragraph{Inductive step.}
Let $\prof{R} \in \mathcal{R}^N$, $A \coloneqq \left(O \setminus \bigcup_{i = 1}^{k-1}\{f_i(\prof{R})\} \right) \cup \{0\}$, and $o \coloneqq \max_{R_k}(A)$. We show that $f_k(\prof{R}) = o = \max_{R_k}(A)$.

Feasibility implies $f_k(\prof{R}) \in A$. If $o = 0$, then \prop{individual rationality} gives $f_k(\prof{R}) = 0$. Let $o \in O$ and suppose toward contradiction that $f_k(\prof{R}) \neq o$. By \prop{non-wastefulness}, $f_j(\prof{R}) = o$ for some agent~$j > k$. Let $R'_j \coloneqq R_k$ and consider $\prof{R}' \coloneqq (R'_j, \prof{R}_{-j})$. By the induction hypothesis,
$$f_\ell(\prof{R}') = \max_{R'_\ell} \left( \left(O \setminus \bigcup_{i = 1}^{\ell-1}\{f_i(\prof{R}')\} \right) \cup \{0\} \right)$$
holds for $\ell = 1,\dots, k-1$. Since $R'_\ell = R_\ell$ for every agent $\ell \in \{1, \dots, k - 1\}$, it follows that $f_\ell(\prof{R}) = f_\ell(\prof{R}')$ for every $\ell \in \{1, \dots, k-1\}$. Therefore
$$ A = \left(O \setminus \bigcup_{i = 1}^{k-1}\{f_i(\prof{R'})\} \right) \cup \{0\},$$
and hence both $f_j(\prof{R}')$ and $f_k(\prof{R}')$ belong to $A$.

\prop{Strategy-proofness} at $\prof{R}'$ implies that $f_j(\prof{R}') \mathrel{R'_j} f_j(\prof{R}) = o$. Therefore, since $o = \max_{R'_j}(A)$ and $f_j(\prof{R}') \in A$, we must have $f_j(\prof{R}') = o$. As a result, $f_k(\prof{R}') \neq o$ and hence $f_j(\prof{R}') \mathrel{P'_k} f_k(\prof{R}')$. Since $k \mathrel{\pi} j$ and $R'_k = R_k = R'_j$, this contradicts (\ref{priority exists}).

This completes the inductive step. Iterating from $k = 2$ through to $k = n$, we see that $f$ is the serial dictatorship associated with $\pi$.\hfill \qed

\subsection{Proof of Theorem~\ref{thm:sequential-dictatorship-null}}

Every sequential dictatorship satisfies the stated properties. For the only-if direction, let $f$ satisfy \prop{individual rationality}, \prop{non-wastefulness}, \prop{strategy-proofness}, \prop{non-bossiness}, and \prop{pairwise tie-breaking}. By Fact~\ref{fact}, $f$ is \prop{monotonic}.

We shall recursively construct a dictator mapping $\kappa : \mathcal{S} \setminus \mathcal{M} \to N$ such that $f = f^\kappa$. Specifically, we prove by induction that, for each $k \in \{0, 1, \dots, n\}$, there exists a partial dictator mapping $\kappa : \mathcal{S}_{<k} \to N$ satisfying
$$ S(k) : \qquad f(\prof{R}) \supseteq \varsigma^\kappa_k(\prof{R}) \quad \text{for every } \prof{R} \in \mathcal{R}^N. $$

\paragraph{Base case.}
The empty function on $\mathcal{S}_{<0} = \emptyset$ is a partial dictator mapping, denoted $\kappa$, and it satisfies $S(0)$ because every allocation extends $\varnothing$.

\paragraph{Induction hypothesis.}
Suppose that, for some $k \in \{0, 1, \dots, n-1\}$, we have constructed a partial dictator mapping $\kappa : \mathcal{S}_{<k} \to N$ satisfying $S(k)$.

\paragraph{Inductive step.}
By Lemma~\ref{recursive extension lemma}, it suffices to identify a next dictator after every $\kappa$-reachable suballocation in $\mathcal{S}_k$. Fix a $\kappa$-reachable suballocation $\varsigma \in \mathcal{S}_k$. By Lemma~\ref{suballocation determines sequence}, $\varsigma$ determines a unique history $H^\kappa(\varsigma) = ((i_1, o_1), \dots, (i_k, o_k))$ of dictator--assignment pairs, in which some of the objects $o_t$ may equal $0$. Hence
$N(\varsigma) = \{i_1, \dots, i_k\}$, $O(\varsigma) = \{o_1, \dots, o_k\} \setminus \{0\}$, and every profile reaching $\varsigma$ has this same ordered history.  (When $k = 0$, the sequence $H^\kappa(\varsigma)$ is empty and $N(\varsigma) = O(\varsigma) = \emptyset$.) Write
$$ \overline{N(\varsigma)} \coloneqq N \setminus N(\varsigma) \qquad \text{and} \qquad m \coloneqq |\overline{N(\varsigma)}| = n - k.$$

If $A(\varsigma) \cap O = \emptyset$, then feasibility and $S(k)$ imply that every agent in $\overline{N(\varsigma)}$ receives $0$ at every profile reaching $\varsigma$, so we may choose any $d(\varsigma) \in \overline{N(\varsigma)}$. Henceforth assume $A(\varsigma) \cap O \neq \emptyset$.

Fix a profile $\prof{R}^\varsigma \in \mathcal{R}^N$ that reaches $\varsigma$. For each $x \in A(\varsigma) \cap O$, fix a preference relation $\succeq^x \in \mathcal{R}$ with $\succeq^x : x, 0, \dots$, and fix $\succeq^0 \in \mathcal{R}$ with $\succeq^0 : 0, \dots$. For each $x \in A(\varsigma) \cap O$, let $$\prof{R}^x \coloneqq (\prof{R}^\varsigma_{N(\varsigma)}, \prof{R}(\succeq^x)_{\overline{N(\varsigma)}}).$$
Every such $\prof{R}^x$ reaches $\varsigma$. By $S(k)$, \prop{individual rationality}, and \prop{non-wastefulness}, $f(\prof{R}^x)$ assigns $x$ to exactly one agent in $\overline{N(\varsigma)}$, and it assigns $0$ to all other agents in $\overline{N(\varsigma)}$. For each $x \in A(\varsigma) \cap O$, let $j(x) \in \overline{N(\varsigma)}$ be the agent with $f_{j(x)}(\prof{R}^x) = x$.

\begin{claim}\label{claim:x-dictator-constant}
    For all $x, x' \in A(\varsigma) \cap O$, $j(x) = j(x')$.
\end{claim}

\begin{subproof}[Proof of Claim~\ref{claim:x-dictator-constant}]
    Toward contradiction, let $x, x' \in A(\varsigma) \cap O$ be such that $j(x) \neq j(x')$. Write
    $$
    I \coloneqq \overline{N(\varsigma)} \setminus \{j(x), j(x')\}, \quad \prof{R} \coloneqq (\prof{R}(\succeq^0)_I, \prof{R}^x_{N \setminus I}), \quad \prof{R}' \coloneqq (\prof{R}(\succeq^0)_{I}, \prof{R}^{x'}_{N \setminus I})
    $$
    Since every agent in $I$ receives $0$ at $f(\prof{R}^x)$ and at $f(\prof{R}^{x'})$, the profile $\prof{R}$ is a monotonic transformation of $\prof{R}^x$ at $f(\prof{R}^x)$, while $\prof{R}'$ is a monotonic transformation of $\prof{R}^{x'}$ at $f(\prof{R}^{x'})$. Hence \prop{monotonicity} implies that $f(\prof{R}) = f(\prof{R}^x)$ and $f(\prof{R}') = f(\prof{R}^{x'})$. Agents $j(x)$ and $j(x')$ both report $\succeq^x$ at $\prof{R}$ and $\succeq^{x'}$ at $\prof{R}'$, while $\prof{R}_{N \setminus \{j(x), j(x')\}} = \prof{R}'_{N \setminus \{j(x), j(x')\}}$. However, 
    $$
    f_{j(x)}(\prof{R}) = x \mathrel{\succ^x} 0 = f_{j(x')}(\prof{R}) \quad \text{and} \quad f_{j(x')}(\prof{R}') = x' \mathrel{\succ^{x'}} 0 = f_{j(x)}(\prof{R}'),
    $$
    which contradicts \prop{pairwise tie-breaking}.
\end{subproof}

By Claim~\ref{claim:x-dictator-constant}, there exists an agent $d(\varsigma) \in N \setminus N(\varsigma)$ such that 
$$ d(\varsigma) = j(x) \quad \text{for every } x \in A(\varsigma) \cap O.$$ 
We show that $d(\varsigma)$ is the next dictator after $\varsigma$.

\begin{claim}\label{claim:candidate-dictator-gets-best-real}
    For every $x \in A(\varsigma) \cap O$ and every $\prof{R} \in \mathcal{R}^N$ with $\varsigma^\kappa_k(\prof{R}) = \varsigma$, $$R_{d(\varsigma)} = \succeq^x \implies f_{d(\varsigma)}(\prof{R}) = x.$$
\end{claim}

\begin{subproof}[Proof of Claim~\ref{claim:candidate-dictator-gets-best-real}]
    We prove the claim simultaneously for every $x \in A(\varsigma) \cap O$ by induction on the number of agents in $\overline{N(\varsigma)} \setminus \{d(\varsigma)\}$ whose report differs from $\succeq^x$. For every $\ell \in \{0, 1, \dots, m-1\}$ and every $x \in A(\varsigma) \cap O$, let $\mathcal{R}^x_\ell$ be the set of profiles reaching $\varsigma$ at which agent $d(\varsigma)$ reports $\succeq^x$ and at most $\ell$ agents in $\overline{N(\varsigma)} \setminus \{d(\varsigma)\}$ report a preference that differs from $\succeq^x$; that is, 
    $$ \mathcal{R}^x_\ell \coloneqq \{\prof{R} \in \mathcal{R}^N \mid \varsigma^\kappa_k(\prof{R}) = \varsigma,\ R_{d(\varsigma)} = \succeq^x, \text{ and } |I^x(\prof{R})| \leq  \ell\},$$
    where $I^x(\prof{R}) \coloneqq \{ i \in \overline{N(\varsigma)} \mid i \neq d(\varsigma) \text{ and } R_i \neq \succeq^x  \}$. We show by induction on $\ell$ that, for every $\ell \in \{0, \dots, m - 1\}$,
    \begin{equation}\label{eq:simultaneous-induction}
    \text{for every } x \in A(\varsigma) \cap O \text{ and every } \prof{R} \in \mathcal{R}^x_\ell, \quad f_{d(\varsigma)}(\prof{R}) = x.
    \end{equation}

    For the base case, let $x \in A(\varsigma) \cap O$ and $\prof{R} \in \mathcal{R}^x_0$. Then $\prof{R}$ and $\prof{R}^x$ reach $\varsigma$ under $\kappa$ and agree on $\overline{N(\varsigma)}$, so Lemma~\ref{tail assignments depend only on tail agents} gives $f(\prof{R}) = f(\prof{R}^x)$. Consequently, the definition of $d(\varsigma)$ implies $f_{d(\varsigma)}(\prof{R}) = x$.

    For the inductive step, assume that (\ref{eq:simultaneous-induction}) holds for some $\ell \in \{0, \dots, m-2\}$. Fix $x \in A(\varsigma) \cap O$ and $\prof{R} \in \mathcal{R}^x_{\ell + 1}$, and suppose toward contradiction that $f_{d(\varsigma)}(\prof{R}) \neq x$. The induction hypothesis implies that $|I^x(\prof{R})| = \ell + 1 \geq 1$, so $I^x(\prof{R})$ is nonempty. Since $R_{d(\varsigma)} = \succeq^x$, \prop{individual rationality} implies $f_{d(\varsigma)}(\prof{R}) = 0$.

    For every $i \in I^x(\prof{R})$, define $\prof{R}^i \coloneqq (\succeq^x, \prof{R}_{-i})$. Then $\prof{R}^i \in \mathcal{R}^x_{\ell}$, so the induction hypothesis gives $f_{d(\varsigma)}(\prof{R}^i) = x$. By \prop{individual rationality} and feasibility, we must have $f_i(\prof{R}^i) = 0$. \prop{Strategy-proofness} at $\prof{R}^i$ therefore implies $0 = f_i(\prof{R}^i) \succeq^x f_i(\prof{R})$, and hence $f_i(\prof{R}) \neq x$. Moreover, $f(\prof{R}^i) \neq f(\prof{R})$ and \prop{non-bossiness} imply that $f_i(\prof{R}) \neq f_i(\prof{R}^i) = 0$. As this holds for every $i \in I^x(\prof{R})$, we have $$f_i(\prof{R}) \in \left(A(\varsigma) \cap O \right) \setminus \{x\} \quad \text{for every } i \in I^x(\prof{R}).$$
    By \prop{non-wastefulness},  $f(\prof{R})$ assigns object $x$, so its recipient is an agent $j \in \overline{N(\varsigma)} \setminus (I^x(\prof{R}) \cup \{d(\varsigma)\})$. In particular, $R_j = \succeq^x$ and $f_j(\prof{R}) = x$.

    Let $\overline{\prof{R}} \in \mathcal{R}^N$ be such that $\overline{\prof{R}}_{N \setminus I^x(\prof{R})} = \prof{R}_{N \setminus I^x(\prof{R})}$ and $\overline{R}_i = \succeq^{f_i(\prof{R})}$ for every $i \in I^x(\prof{R})$. Then $\overline{\prof{R}}$ is a monotonic transformation of $\prof{R}$ at $f(\prof{R})$, so \prop{monotonicity} implies $f(\overline{\prof{R}}) = f(\prof{R})$. In particular, $$f_{d(\varsigma)} (\overline{\prof{R}}) = 0, \qquad f_j(\overline{\prof{R}}) = x.$$
    Because the agents in $\overline{N(\varsigma)} \setminus (I^x(\prof{R}) \cup \{d(\varsigma), j\})$ all report $\succeq^x$ at $\overline{\prof{R}}$, \prop{individual rationality} and feasibility imply that these agents all receive $0$ at $f(\overline{\prof{R}})$.

    Let $j^* \in I^x(\prof{R})$ and $x^* \coloneqq f_{j^*}(\prof{R})$. Construct $\widetilde{\prof{R}} \in \mathcal{R}^N$ such that
    $$
    \widetilde{R}_i = \begin{cases}
    R_i, & \text{if } i \in N(\varsigma),\\
    \succeq^{f_i(\prof{R})}, & \text{if } i \in I^x(\prof{R}),\\
    \succeq^{x^*}, & \text{if }i \in \overline{N(\varsigma)} \setminus I^x(\prof{R}).
    \end{cases}
    $$
    Then $\varsigma^\kappa_k(\widetilde{\prof{R}}) =  \varsigma$, $\widetilde{R}_{d(\varsigma)} = \succeq^{x^*}$, and $|I^{x^*}(\widetilde{\prof{R}})| \leq |I^x(\prof{R})| - 1$ because $\widetilde{R}_{j^*} = \succeq^{x^*}$ and $f_{j^*}(\prof{R}) = x^*$. Hence $\widetilde{\prof{R}} \in \mathcal{R}^{x^*}_\ell$, and the induction hypothesis (\ref{eq:simultaneous-induction}) implies that $f_{d(\varsigma)}(\widetilde{\prof{R}}) = x^*$. Because the agents in $\overline{N(\varsigma)} \setminus I^x(\prof{R})$ all report $\succeq^{x^*}$ at $\widetilde{\prof{R}}$, \prop{individual rationality} and feasibility imply that all of these agents other than $d(\varsigma)$ receive $0$ at $f(\widetilde{\prof{R}})$. In particular, $f_{j}(\widetilde{\prof{R}}) = 0$.

    Write $J \coloneqq \overline{N(\varsigma)} \setminus (I^x(\prof{R}) \cup \{d(\varsigma), j\})$. Construct the profile $\overline{\prof{R}}^0 \coloneqq (\prof{R}(\succeq^0)_{J},\overline{\prof{R}}_{N \setminus J})$ from $\overline{\prof{R}}$ by replacing the report of every agent in $J$ with $\succeq^0$. Since $f(\overline{\prof{R}})$ assigns $0$ to every agent in $J$, $\overline{\prof{R}}^0$ is a monotonic transformation of $\overline{\prof{R}}$ at $f(\overline{\prof{R}})$. Likewise, $\widetilde{\prof{R}}^0 \coloneqq (\prof{R}(\succeq^0)_{J},\widetilde{\prof{R}}_{N \setminus J})$ is a monotonic transformation of $\widetilde{\prof{R}}$ at $f(\widetilde{\prof{R}})$. Hence \prop{monotonicity} implies $$f(\overline{\prof{R}}^0) = f(\overline{\prof{R}}) \quad \text{and} \quad f(\widetilde{\prof{R}}^0) = f(\widetilde{\prof{R}}).$$
    
    Observe that (i) $\overline{R}^0_i = \widetilde{R}^0_i = R_i$ for every $i \in N(\varsigma)$, (ii) $\overline{R}^0_i = \widetilde{R}^0_i = \succeq^{f_i(\prof{R})}$ for every $i \in I^x(\prof{R})$, and (iii) $\overline{R}^0_i = \widetilde{R}^0_i =\succeq^0$ for every $i \in J$. Hence $\overline{\prof{R}}^0_{N \setminus \{d(\varsigma), j\}} = \widetilde{\prof{R}}^0_{N \setminus \{d(\varsigma), j\}}$. Moreover, agents $d(\varsigma)$ and $j$ both report $\succeq^x$ at $\overline{\prof{R}}^0$ and $\succeq^{x^*}$ at $\widetilde{\prof{R}}^0$. However,
    $$f_j(\overline{\prof{R}}^0) = x \mathrel{\succ^x} 0 =  f_{d(\varsigma)}(\overline{\prof{R}}^0) \qquad \text{and} \qquad f_{d(\varsigma)}(\widetilde{\prof{R}}^0) = x^* \mathrel{\succ^{x^*}} 0 =  f_{j}(\widetilde{\prof{R}}^0),$$
    which contradicts \prop{pairwise tie-breaking}. This completes the induction and the proof of the claim.
   
\end{subproof}

It remains to show that $d(\varsigma)$ is the next dictator after $\varsigma$. Fix any $\prof{R} \in \mathcal{R}^N$ with $\varsigma^\kappa_k(\prof{R}) = \varsigma$. Write $x \coloneqq \max_{R_{d(\varsigma)}}A(\varsigma)$. Condition $S(k)$ and feasibility imply that $f_{d(\varsigma)}(\prof{R}) \in A(\varsigma)$. If $x = 0$, then \prop{individual rationality} implies that $f_{d(\varsigma)}(\prof{R}) = x$. 

If $x \in O$, then the profile $\prof{R}' \coloneqq (\succeq^x, \prof{R}_{-d(\varsigma)})$ also satisfies $\varsigma^\kappa_k(\prof{R}') = \varsigma$. Claim~\ref{claim:candidate-dictator-gets-best-real} implies that $f_{d(\varsigma)}(\prof{R}') = x$. \prop{Strategy-proofness} at $\prof{R}$ gives $f_{d(\varsigma)}(\prof{R}) \mathrel{R_{d(\varsigma)}} f_{d(\varsigma)}(\prof{R}') = x$, and hence $f_{d(\varsigma)}(\prof{R}) = x$ by the definition of $x$. 

Thus (\ref{local dictator property}) holds for agent $d(\varsigma)$. By Lemma~\ref{recursive extension lemma}, $\kappa$ extends to a partial dictator mapping $\overline{\kappa} : \mathcal{S}_{<k+1} \to N$ satisfying $S(k+1)$. Replacing $\kappa$ by $\overline{\kappa}$ completes the inductive step.

At $k = n$, $S(n)$ implies that $f(\prof{R}) \supseteq \varsigma^\kappa_n(\prof{R})$ for every $\prof{R} \in \mathcal{R}^N$. Since both sides specify an assignment for every agent, the inclusion is an equality. Hence $f = f^\kappa$. \hfill \qed

\subsection{Proof of Proposition~\ref{prop:bounded-sequential}}

Call a suballocation $\varsigma$ \defterm{real} if it assigns a real object to every agent in $N(\varsigma)$, that is, if $\varsigma(N(\varsigma)) \subseteq O$. In particular, the empty suballocation is real.

\begin{lemma}\label{lemma:bounded-real-serial}
    Let $\kappa$ be a $\pi$-bounded dictator mapping. If $\varsigma \in \mathcal{S} \setminus \mathcal{M}$ is real and $\kappa$-reachable, and $A(\varsigma) \cap O \neq \emptyset$, then $\kappa(\varsigma) = \pi_{|\varsigma| + 1}$.
\end{lemma}

\begin{proof}[Proof of Lemma~\ref{lemma:bounded-real-serial}]
    Let $\varsigma \in \mathcal{S} \setminus \mathcal{M}$ be real and $\kappa$-reachable, and let $A(\varsigma) \cap O \neq \emptyset$. Write $k \coloneqq |\varsigma|$. Since $\varsigma$ is real, $|O(\varsigma)| = k = |N(\varsigma)|$, so (\ref{eq:bounded}) implies $N(\varsigma) = \{\pi_1, \dots, \pi_k\}$. Hence $\kappa(\varsigma) = \pi_\ell$ for some $\ell \geq k + 1$. Since $\pi_\ell \notin N(\varsigma)$, there is a profile $\prof{R}$ that reaches $\varsigma$ under $\kappa$ and at which $R_{\pi_\ell}$ ranks some object $x \in A(\varsigma) \cap O$ first. Then $\varsigma ' \coloneqq \varsigma \cup \{(\pi_\ell, x)\}$ is $\kappa$-reachable, $|O(\varsigma')| = k + 1$, and $N(\varsigma') = \{\pi_1, \dots, \pi_k\} \cup \{\pi_\ell\}$. Applying the $\pi$-boundedness condition (\ref{eq:bounded}) at $\varsigma'$, we obtain
    $$\{\pi_1, \dots, \pi_{k+1}\} \subseteq N(\varsigma') = \{\pi_1, \dots, \pi_k\} \cup \{\pi_\ell\}.$$
    Thus $\ell = k + 1$ and $\kappa(\varsigma) = \pi_{k+1}$.
\end{proof}

We now prove the proposition. Let $f = f^\kappa$ be a sequential dictatorship.

\paragraph{If direction.} 

Let $\kappa$ be $\pi$-bounded. Fix $\prof{R} \in \mathcal{R}^N$ and an agent $i \in N$. Let $i = \pi_k$ and $q \coloneqq |\{o \in O \mid o \mathrel{P_i} 0\}|$. 

At the identical-preference profile $\prof{R}(R_i)$, the first $\min \{q, n\}$ objects assigned are real. By Lemma~\ref{lemma:bounded-real-serial}, these objects are assigned to $\pi_1, \dots, \pi_{\min \{q, n\}}$ in that order. Consequently, $f_i(\prof{R}(R_i))$ is the $k$th-ranked real object under $R_i$ if $k \leq q$, and $0$ otherwise.

If $k > q$, then \prop{individual rationality} implies $f_i(\prof{R}) \mathrel{R_i} 0 = f_i(\prof{R}(R_i))$. If $k \leq q$, let $\varsigma$ be the suballocation that $\prof{R}$ reaches immediately before $i$ is called, i.e., $\varsigma = \varsigma^\kappa_{|\varsigma|}(\prof{R})$ and $\kappa(\varsigma) = i$. Since $i \notin N(\varsigma)$, (\ref{eq:bounded}) implies $|O(\varsigma)| \leq k - 1$. Consequently, at most $k-1$ real objects are unavailable after $\varsigma$, so $A(\varsigma)$ contains at least one of agent~$i$'s $k$ most-preferred real objects under $R_i$. Hence $f_i(\prof{R}) = \max_{R_i}A(\varsigma)  \mathrel{R_i} f_i(\prof{R}(R_i))$.

\paragraph{Only-if direction.}

Suppose that $f = f^\kappa$ satisfies the \prop{identical-preference lower bound}. Write $q \coloneqq \min \{|O|, n \}$.

\paragraph{Step 1: Along real histories, $\kappa$ follows a common priority.}

We prove by induction that, for each $k \in \{1, \dots, q\}$, there exist distinct agents $\pi_1, \dots, \pi_k$ such that
\begin{equation}\label{eq:priority-for-real-histories}
    \kappa(\varsigma) = \pi_{|\varsigma|+1} \quad \text{for every real and $\kappa$-reachable } \varsigma \in \mathcal{S}_{<k}.
\end{equation}
For the base case $k = 1$, set $\pi_1 \coloneqq \kappa(\varnothing)$. For the inductive step, let $k \in \{1, \dots, q-1\}$ and suppose that $\pi_1, \dots, \pi_k$ are distinct agents satisfying (\ref{eq:priority-for-real-histories}). Then every real and $\kappa$-reachable suballocation of size $k$ has domain $I \coloneqq \{\pi_1, \dots, \pi_k\}$. As in Claim~\ref{Adjacent suballocations have the same dictator} in the proof of Proposition~\ref{prop:seq + IPLB}, we first compare two such suballocations that differ for exactly one agent.

\begin{claim}\label{claim:null-adjacent}
    Let $\varsigma^*, \varsigma' \in \mathcal{S}_k$ be real and $\kappa$-reachable. If $\varsigma^*$ and $\varsigma'$ differ for exactly one agent in $I$, then $\kappa(\varsigma^*) = \kappa(\varsigma')$.
\end{claim}

\begin{subproof}[Proof of Claim~\ref{claim:null-adjacent}]
Let $\pi_r \in I$ be the unique agent with $\varsigma^*(\pi_r) \neq \varsigma'(\pi_r)$. Write $x^* \coloneqq \varsigma^*(\pi_r)$ and $x' \coloneqq \varsigma' (\pi_r)$. Suppose toward contradiction that $j^* \coloneqq \kappa(\varsigma^*)$ and $j' \coloneqq \kappa (\varsigma')$ are distinct. Let $\prof{R}' \coloneqq \prof{R}(\succeq')$ and $\prof{R}^* \coloneqq \prof{R}(\succeq^*)$ be the identical-preference profiles associated with
$$
\succeq' : \varsigma'(\pi_1), \dots, \varsigma'(\pi_k), x^*, \dots, 0 \qquad \text{and} \qquad \succeq^* : \varsigma^*(\pi_1), \dots, \varsigma^*(\pi_k), x', \dots, 0,
$$
where every omitted real object is ranked above $0$. By (\ref{eq:priority-for-real-histories}), $\prof{R}'$ reaches $\varsigma'$ and $\prof{R}^*$ reaches $\varsigma^*$. Since $j' = \kappa(\varsigma')$, we have
$$
f_{j'}(\prof{R}') = \max_{\succeq'} A(\varsigma') = x^*.
$$
Now consider the profile $\hat{\prof{R}} \coloneqq (\prof{R}^*_I, \prof{R}'_{N \setminus I})$. This profile also reaches $\varsigma^*$, after which $j^* = \kappa(\varsigma^*) \in N \setminus I$ receives $\max_{\succeq'} A(\varsigma^*) = x'$. The objects weakly preferred to $x^*$ under $\succeq'$ form the set
$$ O(\varsigma') \cup \{x^*\} = O(\varsigma^*) \cup \{x'\},$$
and all of these objects are assigned to agents in $I \cup \{j^*\}$ at $f(\hat{\prof{R}})$. Since $j' \notin I \cup \{j^*\}$ reports $\succeq'$ at $\hat{\prof{R}}$, she receives an object strictly worse than $x^*$ under $\succeq'$, so $f_{j'}(\prof{R}') = x^* \mathrel{P'_{j'}} f_{j'}(\hat{\prof{R}})$, contradicting the \prop{identical-preference lower bound}. Therefore $\kappa(\varsigma^*) = \kappa(\varsigma')$.
\end{subproof}

We now extend Claim~\ref{claim:null-adjacent} beyond ``adjacent'' real suballocations in $\mathcal{S}_k$ as in Claim~\ref{Same dictator globally} in the proof of Proposition~\ref{prop:seq + IPLB}. 

\begin{claim}\label{claim:null-beyond-adjacent}
    Let $\varsigma^* , \varsigma' \in \mathcal{S}_k$ be real and $\kappa$-reachable. Then $\kappa(\varsigma^*) = \kappa(\varsigma')$.
\end{claim}

\begin{subproof}[Proof of Claim~\ref{claim:null-beyond-adjacent}]
    Let $\varsigma^*, \varsigma' \in \mathcal{S}_k$ be distinct real and $\kappa$-reachable suballocations. Then there exists a finite path $(\varsigma^0, \varsigma^1,\dots, \varsigma^T)$ of real and $\kappa$-reachable suballocations in $\mathcal{S}_k$ such that $\varsigma^0 = \varsigma^*$, $\varsigma^T = \varsigma'$, and, for each $t = 1, \dots, T$, $\varsigma^t$ and $\varsigma^{t-1}$ differ for exactly one agent. To see this, construct $\varsigma^t$ from $\varsigma^{t-1}$ as follows. Let $\ell \in \{1,\dots,k\}$ be the smallest index such that $\varsigma^{t-1}(\pi_\ell) \neq \varsigma'(\pi_\ell)$. If $\varsigma'(\pi_\ell) \in O(\varsigma^{t-1})$, then necessarily $\varsigma'(\pi_\ell) = \varsigma^{t-1}(i)$ for some $i \in \{\pi_{\ell+1}, \dots, \pi_{k}\}$; assign agent~$i$ to an arbitrary real object in $A(\varsigma^{t-1}) \cap O$ (such an object exists because $|\varsigma^{t-1}|=k < q \leq |O|$) and leave all other agents' assignments unchanged. If $\varsigma'(\pi_\ell) \notin O(\varsigma^{t-1})$, then assign agent~$\pi_\ell$ to object $\varsigma'(\pi_\ell)$ and leave all other agents' assignments unchanged.
    
    An iterative application of Claim~\ref{claim:null-adjacent} implies that $\kappa(\varsigma^*) = \kappa(\varsigma')$.
\end{subproof}

By Claim~\ref{claim:null-beyond-adjacent}, any two real and $\kappa$-reachable suballocations in $\mathcal{S}_k$ are assigned the same dictator; let $\pi_{k+1}$ be this common value. Since $\kappa(\varsigma) \notin N(\varsigma) = I$ for every such suballocation $\varsigma$, the agents $\pi_1, \dots, \pi_{k+1}$ are distinct and satisfy (\ref{eq:priority-for-real-histories}) with $k+1$ in place of $k$. This completes the induction. 

Let $\pi \in \Pi$ be a priority that ranks $\pi_1, \dots, \pi_q$ first, in that order; if $q = n$, then $\pi = (\pi_1, \dots, \pi_n)$.

\paragraph{Step 2: $\kappa$ is $\pi$-bounded.}

Suppose toward contradiction that (\ref{eq:bounded}) fails at some $\kappa$-reachable suballocation $\varsigma$. Then there exists $k \in \{1, \dots, |O(\varsigma)|\}$ with $\pi_k \notin N(\varsigma)$. In particular, $k \leq q$.

Choose distinct real objects $x_1, \dots, x_k \in O(\varsigma)$ and a preference $\succeq : x_1, \dots, x_k, \dots, 0$, placing all omitted real objects after $x_k$ and before $0$. By (\ref{eq:priority-for-real-histories}) (Step 1), the first $k$ dictators at the identical-preference profile $\prof{R}(\succeq)$ are $\pi_1, \dots, \pi_k$; hence $\prof{R}(\succeq)$ reaches the real suballocation $\varsigma ' \coloneqq \{(\pi_1, x_1), \dots, (\pi_k, x_k)\}$ under $\kappa$, and $f_{\pi_k}(\prof{R}(\succeq)) = x_k$.

Now fix a profile $\prof{R}$ that reaches $\varsigma$ under $\kappa$, and let $\widetilde{\prof{R}} \coloneqq (\succeq, \prof{R}_{-\pi_k})$. Since $\pi_k \notin N(\varsigma)$, the profile $\widetilde{\prof{R}}$ also reaches $\varsigma$. Thus $f(\widetilde{\prof{R}})$ assigns objects $x_1, \dots, x_k \in O(\varsigma)$ to agents in $N(\varsigma)$. Consequently,  $f_{\pi_k}(\prof{R}(\succeq)) = x_k \succ f_{\pi_k}(\widetilde{\prof{R}})$, contradicting the \prop{identical-preference lower bound}. Hence $\kappa$ is $\pi$-bounded, and $f$ is bounded. \qed


\end{document}